\documentclass[letter,11pt]{article}
\pdfoutput=1

\usepackage{tcs}

\newcommand{\snote}[1] {} %

\bibliography{hopkins-lifetime-refs}
\bibliography{main}

\begin{document}

\title{Additive Approximation Schemes for \\
Low-Dimensional Embeddings}

\author{
Prashanti Anderson\footnote{Supported by NSF award no. 2238080 and Google.}\\
\texttt{paanders@csail.mit.edu} \\
MIT
\and
Ainesh Bakshi\footnote{Supported by the NSF TRIPODS program (award DMS-2022448).} \\
\texttt{aineshbakshi@nyu.edu} \\
NYU
\and
Samuel B. Hopkins\footnote{Supported by NSF award no. 2238080 and Google.} \\
\texttt{samhop@mit.edu} \\
MIT
}
\maketitle

\begin{abstract}
We consider the task of fitting low-dimensional embeddings to high-dimensional data.
In particular, we study the $k$-Euclidean Metric Violation problem (\textsf{$k$-EMV}), where the input is $D \in \mathbb{R}^{\binom{n}{2}}_{\geq 0}$ and the goal is to find the closest vector $X \in \mathbb{M}_{k}$, where $\mathbb{M}_k \subset  \mathbb{R}^{\binom{n}{2}}_{\geq 0}$ is the set of all $k$-dimensional Euclidean metrics on $n$ points, and closeness is formulated as the following optimization problem, where $\| \cdot \|$ is the entry-wise $\ell_2$ norm:
\begin{equation*}
    \OPT_{\textrm{EMV}} = \min_{X \in \mathbb{M}_{k} } \norm{ D - X }_2^2\,.
\end{equation*}
Cayton and Dasgupta~\cite{cayton2006robust} showed that this problem is NP-Hard, even when $k=1$. Dhamdhere~\cite{dhamdhere2004approximating} obtained a $\bigO{\log(n)}$-approximation for \textsf{$1$-EMV} and leaves finding a PTAS for it as an open question (reiterated recently by Lee~\cite{lee2025mincsp}).
Although \textsf{$k$-EMV} has been studied in the statistics community for over 70 years, under the name ``multi-dimensional scaling,'' there are no known efficient approximation algorithms for $k > 1$, to the best of our knowledge. 

We provide the first polynomial-time additive approximation scheme for \textsf{$k$-EMV}.
In particular, we obtain an embedding with objective value $\OPT_{\textrm{EMV}} + \eps \norm{D}_2^2$ in $(n\cdot B)^{\poly(k, \eps^{-1})}$ time, where each entry in $D$ can be represented by $B$ bits.
We believe our algorithm is a crucial first step towards obtaining a PTAS for \textsf{$k$-EMV}.
Our key technical contribution is a new analysis of \emph{correlation rounding} for Sherali-Adams / Sum-of-Squares relaxations, tailored to low-dimensional embeddings.
We also show that our techniques allow us to obtain additive approximation schemes for two related problems: a weighted variant of \textsf{$k$-EMV} and $\ell_p$ low-rank approximation for $p>2$.

\end{abstract}
\thispagestyle{empty}

\newpage
\thispagestyle{empty}

\tableofcontents

\newpage

\setcounter{page}{1}

\section{Introduction}

\emph{Dimensionality reduction} aims to find low-dimensional approximations to high-dimensional datasets while preserving the geometric structure of the original dataset, as well as possible.
Dimensionality reduction is pervasive in theory and practice, and is often the first step taken with a high-dimensional noisy dataset before running further learning or inference algorithms \cite{Wikipedia,VanDerMaaten2009}.
It is widely used to reduce noise \cite{Wikipedia}, select training features \cite{Wikipedia2025}, accelerate learning and data-retrieval procedures \cite{Rahimi2007,HIM12}, and visualize high-dimensional data.

In this paper, we study dimensionality reduction algorithms which fit low-dimensional metrics to high-dimensional data.
Our main focus is \textsf{$k$-Euclidean Metric Violation} (\textsf{$k$-EMV}), where given a set of $\binom{n}{2}$ nonnegative numbers $D \in \R^{\binom{n}{2}}_{\geq 0}$, the goal is find the closest $k$-dimensional Euclidean metric $X \in \mathbb{M}_{k}$, where $\mathbb{M}_{k} \subseteq \R^{\binom{n}{2}}_{\geq 0}$ is the set of all $k$-dimensional Euclidean metrics on $n$ points.
If $D_{ij} = \|y_i - y_j\|$ for some high-dimensional dataset $y_1,\ldots,y_n$, this amounts to a non-linear dimension-reduction problem.\footnote{However, in general, the set $D$ does not necessarily correspond to a metric and might violate triangle inequality.} 
In this work, we measure closeness via mean-squared error, so that $\min_{X \in  \mathbb{M}_{k}}  \hspace{0.05in} \norm{ D - X }^2_2$ captures the projection of $D$ to $\mathbb{M}_k$.  

\textsf{$k$-EMV} and several closely related problems have a long history in computer science and beyond.
There has been extensive work on finding the nearest (possibly non-Euclidean) metric to $D$ (``\textsf{Metric Violation}'') \cite{brickell2008metric,FanRB18,FanGRSB20, cohen2025fitting,fan2018metric}.
An even older problem is finding the nearest metric from a specific class of structured metrics.
For instance, finding the nearest ultra-metric or tree metric has been studied for nearly 60 years, originally because of its application to phylogenetic tree reconstruction~\cite{cavalli1967phylogenetic,kvrivanek1986np,day1987computational,farach1993robust,cohen2024fitting,ailon2011fitting,charikar2024improved}.
These works study closeness under a variety of entrywise $\ell_p$ norms $\|D - X\|_p$, with $p=2$ playing a prominent role from the start \cite{cavalli1967phylogenetic}.

\textsf{$k$-EMV} has also been studied in the statistics and data science literature for over 70 years, under the name \emph{multidimensional scaling}, and is broadly used as a nonlinear dimensionality reduction technique for feature extraction, visualization, and more~\cite{DIMACS2001,cox2000multidimensional,kruskal1978multidimensional,borg2012applied,young2013multidimensional,borg2005modern,kruskal1964nonmetric,handbook-of-statistics-MDS,shepard1962analysis-1,shepard1962analysis-2,de2011multidimensional,torgerson1952multidimensional,richardson1938multidimensional,kruskal1964multidimensional,kruskal1964multidimensional,demaine2021multidimensional, bakshi2025metric,badoiu2005approximation}.\footnote{The term ``multidimensional scaling'' typically encompases a number objective functions measuring closeness of a $k$-Euclidean metric to an input $D$, as well as a variety of heuristic methods for fitting such metrics.} 
Implementations of heuristic methods for multidimensional scaling are present in major software packages such as R and scikit-learn \cite{scikit-learn,de2011multidimensional}.

Despite extensive literature on fitting structured metrics to data, our theoretical understanding of \textsf{$k$-EMV} remains fairly limited.
Cayton and Dasgupta~\cite[Section 5.1.1]{cayton2006robust} show NP-hardness, even when $k=1$, under both $\ell_1$ and $\ell_2^2$ closeness. Dhamdhere obtains an $\bigO{\log(n)}$-approximation to \textsf{$1$-EMV} and leaves finding a PTAS as an open question~\cite[Section 5]{dhamdhere2004approximating}, reiterated recently by Lee~\cite{lee2025mincsp}.
For the case $p = \infty$ and $k=1$, H{\aa}stad, Ivansson, and Lagergren~\cite{haastad2003fitting} obtain a $2$-approximation, and show hardness of $(1.4-\delta)$ approximation for any $\delta > 0$.
Bartal, Fandina, and Neiman \cite{bartal2019dimensionality} observe that when $k=n$ the problem becomes convex and can be solved efficiently via semidefinite programming.

Our main result is the first polynomial-time \emph{additive approximation scheme} for \textsf{$k$-EMV}. 
We obtain our approximation scheme via a new analysis of the Sherali-Adams linear programming hierarchy and Sum of Squares semidefinite programming hierarchy \cite{sherali1990hierarchy,lasserre2001global,parrilo2003semidefinite}, generalizing the \emph{global correlation rounding} paradigm.
Global correlation rounding was originally developed in the context of rounding linear (LP) and semidefinite (SDP) programs for constraint satisfaction problems \cite{Barak2011,GS12} and has since become a staple of the toolkit for rounding LPs and SDPs in combinatorial optimization \cite{jain2019mean,davies2020scheduling,Raghavendra2012,CohenAddad2023}.
But, as we discuss in~\cref{sec:techniques}, there are a number of basic obstacles to using global correlation rounding to solve optimization problems over continuous and unbounded domains.
We introduce several new techniques to obtain strong approximation and running-time guarantees when using global correlation rounding for low-dimensional embedding.

Our techniques extend to two further problems, for which we also obtain new additive approximation schemes.
The first problem is a weighted variant of \textsf{$k$-EMV}, where each pair $i,j$ is weighted by a non-negative $w_{ij} \in [0,1]$ in the objective function $\E_{ij} w_{ij} (d_{ij} - \|x_i - x_j\|)^2$.
The weighted variant is long-studied in the multi-dimensional scaling literature, where it has a natural interpretation -- if the input consists of pairwise similarity data among $n$ objects, then only a subset of similarities might be observed~\cite{saeed2018survey}.
Charikar and Guo recently introduced a similar weighted variant for \textsf{Ultra-metric Violation}~\cite{charikar2024improved}.
The weighted case brings significant new technical challenges, which we overcome with a new application of decompositions of dense graphs into few expanders, following Oveis Gharan and Trevisan~\cite{Gharan2014} (see Section~\ref{sec:techniques} for details).

The second problem is \textsf{Entrywise Low-rank Approximation} of matrices (\textsf{$\ell_p$-LRA}), which has a rich algorithmic history \cite{chierichetti2017algorithms,Mahankali2021,Ban2019,cohen2024ptas,ke2003robust,brooks2013l1,kwak2008principal,brooks2012pcal1,park2016iteratively,Clarkson2013,dan2019optimal}.
Given a matrix $A \in \R^{n \times m}$, the goal is to find a rank-$k$ matrix $B \in \R^{n \times m}$ minimizing $\|A-B\|_p$, the entrywise $\ell_p$ norm of $A - B$.
When $p=2$ this problem is solved in polynomial time by singular value decomposition.
But $p=2$ is not always the best choice.  
Larger or smaller $p$ can yield a better low-dimensional approximation of $A$, depending on the respect in which $A$ could be close to low rank -- many entries perturbed mildly compared to a low-rank matrix, in which case large $p$ is appropriate, or few entries perturbed drastically, in which case small $p$ is appropriate.

Both the $p > 2$ and $p < 2$ cases have received significant attention. \textsf{$\ell_p$-LRA} is known to be NP-hard when $p = 0$, $1$ or $\infty$~\cite{Gillis2017,Gillis2015}.
In the extreme case $p=\infty$, the problem is closely related to computing the \textsf{Approximate Rank} of $A$, which has connections across theoretical computer science \cite{Alon2013}.
After a long line of work gradually improving approximation factors, the setting where $p \in [0,2]$ now admits a PTAS~\cite{Ban2019,cohen2024ptas}.
In this work, we focus on the $p > 2$ case, which is comparatively less understood, and as we discuss below, there are some serious roadblocks to translating the techniques from the PTAS in the $p < 2$ case to $p > 2$.
Instead, the state-of-the-art is a $(3+\eps)$ approximation algorithm due to Ban, Bhattiprolu, Bringmann, Kolev, Lee and Woodruff~\cite{Ban2019}.
As a first step towards a PTAS for the $p > 2$ case, we obtain a quasi-polynomial time additive approximation scheme, a result which seems out of reach for prior techniques.

Our work opens several avenues for further progress.
One is to pursue PTASs for all the problems we consider here: we believe that the same rounding algorithms we use here are actually likely to yield PTASs (though possibly with different discretization schemes), but this will require significant new analytical ideas.
The second, more broadly, is to expand the reach of LP/SDP hierarchies to solve approximation algorithms problems beyond combinatorial optimization, by extending global correlation rounding beyond discrete optimization.

\subsection{Results}

\paragraph{$k$-Euclidean Metric Violation.}

We begin by formally defining the problem of finding the closest $k$-dimensional Euclidean metric under the $\ell_2^2$-norm.
We write $\E_{i \sim [n]} f(i)$ for the uniform average of a function $f \, : \, [n] \rightarrow \R$.

\begin{problem}[Euclidean Metric Violation]
\label{problem:rs}
Given a set of $\binom{n}{2}$ distances $\Set{d_{ij}}_{ij \in[n]}$, the \textsf{$k$-Euclidean Metric Violation} (\textsf{$k$-EMV}) objective finds an embedding $x_1,\ldots,x_n$ in $\mathbb{R}^k$ such that 
\begin{equation*}
    \OPT_{\textrm{EMV}}(d) = \min_{ x_1, \ldots , x_n \in \R^k } \E_{ij \sim [n]} \Paren{  d_{ij} - \norm{x_i - x_j }_2 }^2
\end{equation*}
\end{problem}

We aim to design approximation algorithms in the regime $k = O(1)$ -- recall, even the $k=1$ case is NP-hard to solve exactly \cite{cayton2006robust}.
Aside from being of interest in its own right, the constant $k$ case is most relevant to visualization applications.
Furthermore, as we discuss below, linear dimensionality reduction techniques like random linear projections give particularly poor guarantees when $k=O(1)$.
This seems to force any algorithm with strong provable guarantees for \textsf{$k$-EMV} to find non-linear dimension reductions, a goal of interest in its own right in the long-term project of obtaining a theoretical understanding of widespread methods like Isomap, t-SNE, locally linear embeddings, and autoencoders \cite{hinton2002stochastic,Hinton2006,tenenbaum2000global,roweis2000nonlinear}.

Despite its natural formulation, almost all algorithmic work on $\textsf{$k$-EMV}$ is heuristic -- there are only a handful of approaches with provable guarantees.
Beyond the $\bigO{\log n}$-approximation of \cite{dhamdhere2004approximating} when $k=1$, we are aware of two possible approaches for the constant-$k$ regime.
First, one can take the optimal $n$-dimensional solution (computed via SDP \cite{bartal2019dimensionality}) and then use random projections to dimension-reduce the resulting $y_1,\ldots,y_n \in \R^n$.
This yields a solution of objective value at most $\OPT_{\textrm{EMV}}(d) + \bigO{\E_{ij} d_{ij}^2/{k}}$.
Since $\OPT_{\textrm{EMV}}(d) \leq \E_{ij} d_{ij}^2$ (consider the embedding that maps all the $x_i$'s to the origin), the overall scaling of this additive error is sensible, but the error cannot be driven towards zero unless $k \rightarrow \infty$.

Another approach, due to Demaine, Hesterberg, Koehler, Lynch, and Urschel \cite{demaine2021multidimensional}, produces a non-linear embedding by adapting a polynomial-time approximation scheme (PTAS) for solving dense constraint satisfaction problems.\footnote{The main result of \cite{demaine2021multidimensional} is phrased for a different objective function, ``Kamada-Kawai''. Here we state the guarantees that their proof approach would naturally provide using the same ideas for EMV.}
They give an additive approximation scheme, whose additive error can be driven to zero by spending additional running time, but with the serious drawback that the running time depends \emph{doubly-exponentially on the input size}.
Concretely, if each $d_{ij}$ can be represented using $B$ bits, then for any $\eps > 0$ the algorithm of \cite{demaine2021multidimensional} runs in time $n^{O(1)} \cdot \exp((k \cdot 2^B/\eps)^{\mathcal{O}(1)})$\footnote{The doubly exponential dependence on bit complexity of the approach of \cite{demaine2021multidimensional} is phrased as an exponential dependence on $\Delta$, or the aspect ratio, in the original work.} and returns a solution of quality $\OPT_{\textrm{EVM}}(d) + \eps \cdot \E_{ij} d_{ij}^2$.\footnote{In fact, once we are spending double-exponential time with respect to $B$, we can take the additive error to be $\eps$ rather than $\eps \E_{ij} d_{ij}^2$, but we prefer the latter normalization in our work.}
At a high level, this double-exponential dependence arises because algorithms for CSPs, and the usual analyses thereof, are not naturally suited to finding real-valued solutions to non-Boolean optimization problems.

Our main result is the first truly polynomial-time additive approximation scheme for \textsf{$k$-EMV}:

\begin{theorem}[Additive Approximation Scheme for \textsf{$k$-EMV}, see~\cref{thm:raw-stress-main}]
\label{thm:RS-main-intro}
    Given input $\calD = \{d_{ij}\}_{i,j \in [n]}$, where each $d_{ij}$ can be represented by $B$ bits, and $0<\eps<1$, there exists an algorithm that runs in time $B^{\poly(k,\eps^{-1})} \cdot \poly(n)$,  and with probability at least $0.99$, outputs an embedding in $\R^k$ with \textsf{$k$-EMV} objective value at most $$\OPT_{\textrm{EMV}}(d) + \eps \cdot \E_{ij} d_{ij}^2 \, .$$
\end{theorem}

Recent work by Bakshi, Cohen, Hopkins, Jayaram, and Lattanzi~\cite{bakshi2025metric} gave algorithms for a related problem, the Kamada-Kawai objective, which reduces this double-exponential dependence in $B$ to a single-exponential dependence. However, their approximation guarantees are quantitatively weak; they produce a solution of value $B\cdot \left(\OPT\right)^{1/k} + \eps$ (where for Kamada-Kawai $\OPT \leq 1$ and thus $\OPT^{1/k} \ggg \OPT$). Furthermore, their techniques do not straightforwardly extend to the $\textsf{$k$-EMV}$ problem. Intuitively, this is due to the fact that Kamada-Kawai and \textsf{$k$-EMV} emphasize different desiderata of the resulting embedding; $\textsf{$k$-EMV}$ emphasizes preserving all distances well up to additive error (which may result in a large multiplicative gap between the input $d_{ij}$ and the embedded distance for small $d_{ij}$), while Kamada-Kawai emphasizes preserving all distances well up to multiplicative factors and favors contractions over expansions.

Our second result is an algorithm for the significantly more challenging \emph{weighted} \textsf{$k$-EMV}, where the input consists of $\{d_{ij}\}_{ij \in [n]}$ and weights $w_{ij} \in [0,1]$, and the objective function is $$\OPT_{\textrm{WEMV}}(w,d) = \min_{x_1, \ldots, x_n \in \R^k} \hspace{0.1in} \E_{ij \sim w} (d_{ij} - \|x_i - x_j\|)^2 \,,$$
where the notation $\E_{ij \sim w}$ means that a pair $(i,j)$ is drawn with weight proportional to $w_{ij}$.
We obtain an algorithm with approximation guarantees and running time depending on the density of the weights, when viewed as a weighted graph.
In the weighted case, the running time of our algorithm depends slightly super-polynomially on the magnitudes of the distances $d_{ij}$, when the weighted graph is dense and regular -- this still makes an exponential improvement over the double-exponential bit-complexity dependence incurred by the approach of \cite{demaine2021multidimensional}.
We discuss in~\cref{sec:techniques} the origin of these dependencies on density and magnitudes of the distances.

\begin{theorem}[Weighted \textsf{EMV}, see~\cref{thm:regular-raw-stress}]
\label{thm:RS-weighted-intro}
Given input $\{d_{ij}\}_{i,j \in [n]}$, weights $w_{ij} \in [0,1]$ such that each row/column sum of the matrix $\{w_{ij}\}_{ij \in [n]}$ is equal to $\delta n$, $d_{ij}$ and $w_{ij}$ can be represented by $B$ bits, and $0<\eps<1$,  there exists an algorithm that runs in $\left(n\cdot 2^B\right)^{\poly(k,\eps^{-1},\delta^{-1})}$ time and outputs an embedding $x_1,\ldots,x_n \in \R^k$ with weighted \textsf{$k$-EMV} objective value at most 
$$\OPT_{\textrm{WEMV}}(w,d) + \eps \cdot \E_{ij \sim w} d_{ij}^2,$$
\end{theorem}
\noindent The algorithm in the theorem above also runs in time $\left(n\cdot B \cdot 2^\kappa \right)^{\poly(k,\eps^{-1},\delta^{-1})}$, where $\kappa = \sqrt{ \max_{i, j \in [n]} d_{ij}^2 / \E_{i, j \sim [n]} d_{ij}^2}$. This can be significantly faster than $\left(n\cdot 2^B\right)^{\poly(k,\eps^{-1},\delta^{-1})}$ time on certain instances (for example, when most of the distances are very large), but there exist worst-case instances where $\kappa = \Theta(B)$.

By comparison, using the techniques of \cite{demaine2021multidimensional} in the weighted setting would also give an additive approximation scheme with exponential running-time dependence on $\delta^{-1}$, while incurring a double-exponential running-time dependence on input size, though allowing for the regularity assumption on the weights to be removed.

\paragraph{Rank-One Approximation.}
The new analysis of global correlation rounding which leads to Theorems~\ref{thm:RS-main-intro} and~\ref{thm:RS-weighted-intro} also leads to a new algorithm for entry-wise \textsf{$\ell_p$-LRA} in the challenging $p > 2$ regime.

\begin{problem}[$\ell_p$-Rank-One Approximation (LRA)]
\label{prob:lra}
Given an $n \times m$ matrix $A$, find $u \in \R^n, v \in \R^m$ minimizing $\|A - uv^\top\|_p$, the entrywise $\ell_p$ norm of $A - uv^\top$.
We define
\begin{equation*}
    \OPT_{\textrm{LRA}} = \min_{ u\in \R^{n} , v \in \R^{m} } \Norm{  A - u v^\top  }_{p}^{p}\,,
\end{equation*}   
\end{problem}

\textsf{$\ell_p$-LRA} has a PTAS for $p \in [0,2]$, sometimes under bit-complexity assumptions \cite{Ban2019,cohen2024ptas}, but for $p > 2$, the best known efficient algorithm is a $(3+\eps)$-approximation in time $(nm)^{\poly(\eps^{-1})}$ for matrices with entries bounded by $\poly(n,m)$, again due to \cite{Ban2019}.
The PTAS for $p \in (0,2)$ relies on the existence of $p$-stable random variables, which don't exist for $p > 2$, a major roadblock to using the same techniques in the $p > 2$ case \cite{wiki:stable_distribution}. 
(Both of these algorithms extend to rank-$k$ approximation for $k > 1$, with exponential and doubly exponential dependence on $k$ in the running time respectively.)\footnote{Bicriteria approximations, which in this case would allow outputting a rank-$3$ matrix which competes with the best rank-one solution in objective value, are known which can avoid these bit complexity assumptions \cite{Mahankali2021}.}

We give a quasi-polynomial time additive approximation scheme for $\ell_p$ rank-one approximation when $p$ is any even integer, without assumptions on the magnitudes of $A$'s entries.

\begin{theorem}[$\ell_p$-LRA Additive Approximation, see~\cref{thm:lra}]
\label{thm:LRA-intro}
Given as input a matrix $A \in \R^{n \times m}$ such that each entry can be represented by $B$ bits, for every even $p \in \N$ and every $\eps > 0$ there exists an algorithm that runs in time $(mn\cdot B)^{(\eps^{-1} \log n)^{C_p}}$, where $C_p$ is a constant depending only on $p$, and outputs $\hat{u} \in \R^n$ and $\hat{v} \in \R^m$ such that with probability at least $0.99$, 
  \[
  \Norm{A - \hat{u} \hat{v}^\top}_p^p \leq \OPT_{\textrm{LRA}} + \eps \cdot \Norm{A}_p^p \, .
  \]
\end{theorem}

We are not aware of prior work obtaining a comparable guarantee even in subexponential time.
Although strictly speaking Theorem~\ref{thm:LRA-intro} is incomparable to the $(3+\eps)$-approximation of \cite{Ban2019}, we hope that it represents a step towards a PTAS.
Indeed, for the related problem of $\ell_0$ low rank approximation, an additive approximation scheme was one of the key innovations required to obtain a PTAS \cite{cohen2024ptas}.

\subsection{Related Work}

\paragraph{Dense CSPs.}
Designing algorithms for dense and structured max-CSPs remains a central topic in theoretical computer science, leading to a proliferation of techniques, including sampling-based algorithms~\cite{arora1995polynomial, alon2003random, MS08}, algorithms leveraging regularity lemmas~\cite{frieze1996regularity, coja2010efficient}, and correlation rounding of convex hierarchies~\cite{Barak2011, GS12, yoshida2014approximation}. While the first two approaches so far lead to the fastest running times~\cite{yaroslavtsev2014going}, correlation rounding can extend to sparse CSPs where the underlying graph is an expander~\cite{Barak2011,Alev2019}.
In contrast, very little is known about min-CSPs, even in the complete case. Recently, Anand, Lee, and Sharma~\cite{anand2025min} obtained a constant factor approximation for \textsf{Min $2$-SAT} in the complete case. 
Our work borrows techniques from the dense CSP toolkit.

\paragraph{Metric Violation Distance and Metric Embeddings.}
\textsf{Metric Violation} was first introduced by Brickell, Dhillon, Sra and Tropp~\cite{brickell2008metric}, and they show that when $\mathbb{M}$ is the set of all metrics on $n$ points, there exists a polynomial-time algorithm to find the closest metric under any $\ell_p$ norm ($1\leq p\leq \infty$), via convex programming. For the $\ell_0$ variant of \textsf{Metric Violation} the state-of-the-art is a $\bigO{\log(n)}$-approximation due to Cohen-Addad, Fan, Lee, and De Mesmay~\cite{cohen2025fitting} and APX-Hardness of $2$ under the Unique Games conjecture due to Fan, Raichel, and Van Buskirk~\cite{fan2018metric}.  

Fitting structured metrics, such as ultra-metrics or tree metrics, has been studied for even longer, due to their application in reconstructing phylogenetic trees.
Cavalli-Sforza and Edwards~\cite{cavalli1967phylogenetic} introduce the problem and consider finding the closest ultra-metric in $\ell_2^2$ norm.
Finding the closest ultra-metric under various $\ell_p$ norms has been well-studied in theoretical computer science for over three decades. K{\v{r}}iv{\'a}nek and Mor{\'a}vek~\cite{kvrivanek1986np} and Day~\cite{day1987computational} showed NP-Hardness for \textsf{Ultra-metric Violation} when $p = 1, 2$. Farach-Colton, Kannan and Warnow~\cite{farach1993robust} gave an exact polynomial time algorithm when $p = \infty$, Cohen-Addad, Das, Kipouridis, Parotsidis, and Thorup~\cite{cohen2024fitting} gave a constant factor approximation when $p =1$ and Ailon and Charikar~\cite{ailon2011fitting} gave a $\bigO{ (\log(n) \log\log(n))^{1/p}  }$-approximation for any $p \geq 1$. Recently, Charikar and Gao~\cite{charikar2024improved} obtained a $5$-approximation when $p = 0$. This problem is also closely related to correlation clustering, see~\cite{cohen2025fitting} and references therein. 

\paragraph{$\ell_p$-Low Rank Approximation.} Entry-wise low-rank approximation has been well-studied in the theory and machine learning communities~\cite{koller1993constructing, ke2003robust, Clarkson2013, markopoulos2013some, musco2015randomized}. The special case of $p=2$ results in a convex optimization problem, which admits an input-sparsity time algorithm~\cite{Clarkson2013}. It is known that the problem is NP-hard when $p = 0$, $1$ or $\infty$~\cite{Gillis2017,Gillis2015}. For $p \in (0,2)$~\cite{Ban2019} obtained a PTAS under bit-complexity assumptions.
They also obtained a $(3+\eps)$-approximation for any $p>2$.
For $p\in [1,2)$, Mahankali and Woodruff obtained an additive approximation scheme with running time $\poly(n, m) + 2^{\poly(k/\eps)}$~\cite{Mahankali2021}, without bit complexity assumptions.
Finally, for $p=0$, \cite{cohen2024ptas} recently obtained a PTAS that scales polynomially in the bit complexity of the input.

\section{Technical Overview}
\label{sec:techniques}
We now give an overview of our additive approximation schemes, which rely on a new way to analyze \emph{global correlation rounding} \cite{Barak2011}.
Throughout this section, we take $k=1$.
We focus on \textsf{$1$-EMV} here, since our results for \textsf{$\ell_p$-LRA} employ largely the same techniques.

\subsection{Warmup: $n^{2^B}$-time Euclidean Metric Violation via Sherali Adams}
We start by describing an ``orthodox'' use of techniques from dense constraint satisfaction problems (CSPs) to obtain an additive approximation scheme for Euclidean Metric Violation whose running time is polynomial in $n$ but doubly-exponential in the input size, à la \cite{demaine2021multidimensional}.
This guarantee can be obtained by black-box use of any PTAS for dense Max-CSPs, but it will be convenient for us to describe a Sherali-Adams linear programming approach, which we will later modify to obtain our algorithms.

\paragraph{Discretization and Sherali-Adams LP.}
Recall,  $\{d_{ij}\}_{ij \in [n]}$ is an arbitrary matrix such that each $d_{i,j} \geq 0$. 
To obtain an LP relaxation of $\min_{x_1,\ldots,x_n \in \R} \E_{ij} (d_{ij} - |x_i - x_j|)^2$,
we start with a standard discretization argument.
Without loss of generality, we can assume that either $d_{i,j}= 0$ or $1 \leq d_{ij} \leq 2^B$ for a bit-complexity parameter $B$.
It is not hard to show that there is a solution $x$ with objective value $\OPT_{\textrm{EMV}} + \eps \cdot \E_{ij} d_{ij}^2$ in which each $x_i$ assumes values in a discrete subset $S$ of $\R$ with around $\poly\left(2^B / \eps\right)$ points, each of magnitude at most $\poly\left(2^B / \eps\right)$.
(A slightly more sophisticated argument with a geometrically-spaced grid actually shows that $\bigO{B/\eps}$ values per variable $x_i$ suffices -- we will return to this later).

A solution to the degree-$t$ Sherali-Adams linear programming relaxation is a \emph{degree-$t$ pseudo-distribution} over $S^n$.
A degree-$t$ pseudo-distribution $\mu$ consists of \emph{local distributions} $\{ \mu_T \}_{T \in \binom{n}{t}}$, where each $\mu_T$ is a distribution on $S^t$, and the distributions are locally consistent, meaning that $\mu_T, \mu_U$ induce identical marginal distributions on $T \cap U$ -- we often write $\mu_{T \cap U}$ for this marginal.
We denote the objective value associated to $\mu$ by
\[
  \pE_{\mu} \E_{ij} (d_{ij} - |x_i - x_j|)^2 := \E_{ij} \E_{(x_i,x_j) \sim \mu_{ij}}(d_{ij} - |x_i - x_j|)^2 \, 
\]
where ``$\pE_\mu$'' is the ``pseudo-expectation operator'' associated to $\mu$.
The degree-$t$ Sherali-Adams LP has $(nS)^{\mathcal{O}(t)}$ variables and constraints, and so can be solved in $(nS)^{\mathcal{O}(t)}$ time.
Our goal is to round a solution $\mu$ to this relaxation to a solution $y_1,\ldots,y_n \in \R$. 
The crucial operation we can perform on a pseudo-distribution is \emph{conditioning} on the value of a variable (see \cref{sec:prelims} for details).
For any $x_i$ and any value $z \in \supp(\mu_i)$, we can obtain a degree $t-1$ pseudo-distribution $\mu'$ by conditioning $\mu$ on the event $x_i = z$.
This entails conditioning each of the local distributions $\mu_T$ on this event (throwing out those local distributions which do not contain $i$).
A rounding scheme which needs to condition on $s$ variables will result in an LP of size $(nS)^{\mathcal{O}(s)}$.

\paragraph{Global Correlation Rounding.}
The key insight of \cite{Barak2011} (anticipated by earlier works in statistical physics \cite{Montanari2008,Ioffe2000}) is (1) iterated conditioning produces approximately pairwise-independent pseudo-distributions, and (2) approximately pairwise-independent distributions are easy to round for objective functions which are sums of functions each depending on only two variables.
More precisely, a simple potential function argument shows the following:
\begin{lemma}[Pinning Lemma, Total Variation Version]
\label{lem:pinning-lemma-intro}
  Let $\mu$ be a degree-$t$ pseudodistribution on variables $x_1,\ldots,x_n$, taking values in a discrete set $S$.
  For every $s \leq t-2$, there exists a set $T \subseteq [n]$ with $|T| \leq s$ such that if $z_T \sim \mu_T$ and $\mu'$ is given by conditioning on the event $x_T = z_T$, then
  \[
  \E_{z_T \sim \mu_T} \E_{ij \sim [n]} \tv(\mu'_{ij}, \mu'_i \otimes \mu'_j) \leq \sqrt{\frac{\log |S|}{s}} \, .
  \]
\end{lemma}
\noindent Here $\tv$ is total variation distance, and $\mu'_i \otimes \mu'_j$ is the product of the distributions $\mu'_i$ and $\mu'_j$.

Suppose given a degree-$t$ pseudodistribution $\mu$ of degree $s$ with objective value at most $\OPT_{\textrm{EMV}} + \eps \cdot  \E_{ij} d_{ij}^2$, for some $s$ we will choose later.
By searching over all subsets $T$ of size at most $t-2$ and choices for values $z_T \in S^T$, we can find a conditioned pseudo-distribution $\mu'$ with objective value at most $\OPT_{\textrm{EMV}} + \eps \cdot \E_{ij} d_{ij}^2$ and $\E_{ij \sim [n]} \tv(\mu_{ij}', \mu_i' \otimes \mu_j') \leq \bigO { \log |S|) / \sqrt{t} }$. Then, we can obtain $y$  by sampling each $y_i \sim \mu'_i$ independently. Let $\mu^{' \otimes}$ denote the distribution $\mu'_1\otimes \mu'_2\otimes \ldots \otimes \mu'_n$. We can then use the pinning lemma to analyze the rounding cost as follows: 

\begin{equation*}
    \begin{split}
         \E_{z_T \sim \mu_T}  & \E_{ij \sim [n]} \pE_{\mu^{'\otimes}} \Paren{ d_{ij} - |x_i  - x_j |  }^2 \\
         & = \E_{z_T \sim \mu_T}  \E_{ij \sim [n]} \pE_{\mu_{ij}} \Paren{ d_{ij} - |x_i  - x_j |  }^2 + \Paren{\E_{z_T \sim \mu_T}  \E_{ij \sim [n]} \E_{\mu'_i \otimes \mu'_j} (d_{ij} - |x_i - x_j|^2) - \E_{\mu'_{ij}} (d_{ij} - |x_i - x_j|)^2}  \\
         & \leq \OPT_{\textrm{EMV}} + \eps \cdot \E_{ij} d_{ij}^2 + \Paren{\E_{z_T \sim \mu_T}  \E_{ij \sim [n]} \E_{\mu'_i \otimes \mu'_j} (d_{ij} - |x_i - x_j|^2) - \E_{\mu'_{ij}} (d_{ij} - |x_i - x_j|)^2} \\
         & \leq \OPT_{\textrm{EMV}} + \underbrace{\eps \cdot \E_{ij} d_{ij}^2}_{\text{discretization error}} + \underbrace{\frac{2^{2B}}{\eps^2} \cdot O\Paren{\E_{z_T \sim \mu_T} \E_{ij \sim [n]} \tv(\mu'_{ij}, \mu'_i \otimes \mu'_j) }}_{\text{rounding error}}\, ,
    \end{split}
\end{equation*}
where we used:
\begin{itemize}
    \item for every $x_i$ in the support of $\mu'_i$ and $x_j$ in the support of $\mu'_j$, the function $(d_{ij} - |x_i - x_j|)^2$ is at most $O(2^{2B} / \eps^2)$, and
    \item if $\alpha,\beta$ are distributions and $f$ is a real-valued function, $|\E_{x \sim \alpha} f(x) - \E_{x \sim \beta} f(x)| \leq \|f\|_{\infty} \cdot \tv(\alpha,\beta)$, by the definition of total variation.
\end{itemize}
Then we can apply the pinning lemma (Lemma~\ref{lem:pinning-lemma-intro}) to conclude that if we choose $T$ appropriately of size $2^{O(B)} / \eps^{O(1)}$, the rounded solution has expected objective value $\OPT +  \eps \cdot \E_{ij} d_{ij}^2$.
Since we had to pick $|T|$ exponentially large in $B$, and the size of the Sherali-Adams LP is exponential in $|T|$, we get double-exponential dependence on $B$.

Our rounding schemes always follow this plan: search over all subsets $T$ of coordinates of a given fixed size, draw a sample $z_T \sim \mu_T$, condition on $x_T = z_T$, and then draw independently from each of the $1$-local conditional distributions to obtain a rounded solution.
Our technical contributions are new analyses of this canonical rounding scheme, showing that it obtains much smaller error for low-dimensional embedding problems than the aforementioned analysis suggests.

\subsection{Euclidean Metric Violation in Polynomial Time}
Next we describe the ideas needed to reduce the running time to $n^{\poly(B,\eps^{-1})}$, removing one exponential in $B$, by showing that in the rounding scheme above we can actually take $|T| = \eps^{-\mathcal{O}(1)}$. Note, this still isn't polynomial time, and we eke out a further exponential improvement on $B$ using a better discretization argument later in this section. 
Our main goal is to avoid the pessimistic bound of $2^{\mathcal{O}(B)}/\eps^2$ on the maximum possible value of $(d_{ij} - |x_i - x_j|)^2$ used in the analysis above.\footnote{This by itself would not be enough to achieve the bound on $|T|$ given above but is the main obstacle in removing the doubly exponential dependence on $B$.}
Note that there are initial pseudo-distributions $\mu$ with optimal objective value for which the $2^B/\eps$ bound really could be tight.
For example, imagine that $\mu$ corresponds to an actual distribution over optimal solutions $x_1,\ldots,x_n \in \R$, obtained by taking a single optimal solution and then outputting $x_i + \sigma \cdot 2^B$, where $\sigma$ is a random sign (chosen identically for each $i$).
Then under the product distribution $\mu_i' \otimes \mu_j'$, we expect $|x_i - x_j| \approx 2^B$.

\paragraph{Anchoring the embedding by conditioning once.}
The first observation is that after conditioning on a single $x_a$, each local distribution $\mu_i'$ will place most of the probability mass of each $x_i$ in a ball of radius $\bigO{ \sqrt{ \E_{i,j \sim [n]} d_{ij}^2} }$, which can be significantly smaller than $2^B/\eps$.\footnote{While~\cite{bakshi2025metric} also utilizes an ``anchoring'' procedure, their approach requires that most points are very strongly anchored. In contrast, our initial anchoring procedure is much weaker and only serves as the first step in our analysis.}
To see why, it is convenient to think of $\mu$ as an actual probability distribution over $1$-dimensional metrics (i.e., $\mu$ supported on $\R^n$), and consider choosing a random $a \sim [n]$ and conditioning $\mu$ on the location of $x_a$.
This ``anchors'' the distribution to $x_a$, and $99\%$ of the probability mass of each $x_i$ will now be contained in a ball of radius $\bigO{\sqrt{\E_{\mu_{ia}} |x_i - x_a|^2}}$ centered at $x_a$, by Chebyshev's inequality.
This escapes the bad example of a distribution $\mu$ on solutions which are randomly shifted by $\pm 2^B$ -- conditioning on a single $x_a$ fixes such a global random shift.
By comparing $\E_{i\sim [n]} \E_\mu |x_i - x_a|^2$ to the objective function, we can also conclude that these balls are not too big on average: $\E_i \E_a \E_{\mu} |x_i - x_a|^2 \leq  \bigO{ \E_{ij} d_{ij}^2}$, since
\begin{equation}
\label{eqn:avg-variance-intro}
    \begin{split}
        \E_{a \sim [n]} \E_{i \sim [n]} \E_{\mu_{ia}} (| x_i - x_a | \pm d_{ij} )^2 \leq 2 \pE_{\mu} \E_{i, a \sim [n]}  \Paren{ d_{ia} - |x_i - x_a|}^2 + 2 \E_{i, a \sim [n]} d_{ia}^2 \leq 4 \E_{ij} d_{ij}^2 \, ,
    \end{split}
\end{equation}
using in the last step that $\OPT_{\textrm{EMV}} \leq \E_{i, a \sim [n]} d_{ia}^2$.
Furthermore, the same argument holds even if $\mu$ is a Sherali-Adams pseudo-distribution.

Optimistically, if we imagine that after conditioning on $x_a$, each $x_i$ is with probability $1$ contained in a ball of radius $\bigO{\sqrt{\E_\mu |x_i - x_a|}}$, and furthermore, for each $i$, we have $\E_{\mu_{ia}} |x_i - x_a|^2 \leq \bigO{ \E_{ij} d_{ij}^2}$, then we can use the same $\tv$-based analysis as above, but now we get to use the bound $\bigO{\E_{ij} d_{ij}^2}$ in place of the $2^{2B} / \eps^2$ upper bound on $(d_{ij} - |x_i - x_j|)^2$.
This leads to rounded objective value
\[
\OPT_{\textrm{EMV}} + \underbrace{\eps \E_{ij} d_{ij}^2}_{\text{discretization error}} + \underbrace{\E_{ij} d_{ij}^2 \cdot \mathcal{O} \Paren{\E_{z_T \sim \mu_T} \E_{ij \sim [n]} \tv(\mu'_{ij}, \mu'_i \otimes \mu'_j) }}_{\text{rounding error}}\, .
\]
Then we would deploy the pinning lemma with $|T| = (B/\eps)^{\mathcal{O}(1)}$, which, under these optimistic assumptions, would lead to objective value $\OPT_{\textrm{EMV}} + \eps \cdot \E_{ij} d_{ij}^2$ in time $n^{\poly(B,\eps^{-1})}$.

\paragraph{Splitting into Linear and Quadratic Terms.}
The sketch above makes an unrealistic assumption: that after we condition on a single $x_a$, there is a ball of radius $\bigO{ \sqrt{\E_{ij} d_{ij}^2}}$ containing the entire support of each $x_i$.
But really we only know a bound on the second moment $\E_{\mu_{ia}} |x_i - x_a|^2$, not a probability-1 upper bound on $|x_i - x_a|$, and even that only on average over $i \in [n]$ (recall \cref{eqn:avg-variance-intro}).
So, we need a version of the argument using only these second-moment bounds.

The first step is to expand each term of our objective function:
\[
(d_{ij} - |x_i - x_j|)^2 = d_{ij}^2 - 2 d_{ij} |x_i - x_j | + x_i^2 + x_j^2 - 2 x_i x_j \, .
\]
Since $\E_{\mu'_{ij}} d_{ij}^2 = \E_{\mu'_i \otimes \mu'_j} d_{ij}^2$ and similarly for $x_i^2$ and $x_j^2$, our rounding scheme incurs zero expected error on these terms.
The remaining rounding error can then be broken down into:
\begin{align*}
\underbrace{\E_{ij} d_{ij} \cdot \Abs{ \E_{\mu'_{ij}} |x_i - x_j| - \E_{\mu_i' \otimes \mu_j'} |x_i - x_j| } }_{\text{linear term}} \quad \text{ and } \quad \E_{ij} \Abs{ \E_{\mu'_{ij}} x_i x_j - \E_{\mu'_i \otimes \mu'_j} x_i x_j } = \underbrace{\E_{ij} \Abs{ \Cov_{\mu'_{ij}}(x_i,x_j)}}_{\text{quadratic term}}\, .
\end{align*}

\paragraph{The Linear Term: Splitting into Tail and Bulk Contributions.}
For the linear term, we will impose the probability-one upper bound on $|x_i - x_j|$ ``by force'', by splitting each $x_i$ up into a bulk part and a tail part.
Let $E_i$ be the event that $|x_i - x_a| \leq \eps^{-\mathcal{O}(1)} \sqrt{\E_{\mu_{ia}} |x_i - x_a|^2}$, where $x_a$ is the anchor point we conditioned on previously.
We can split the function $|x_i - x_j|$ up as
\[
|x_i - x_j| = |x_i - x_j| \cdot \Ind(E_i, E_j) + |x_i - x_j| \cdot \Ind(\overline{E_i} \text{ or } \overline{E_j}) \, .
\]

Using Chebyshev's inequality, we can show that the tail parts can't contribute much to the objective value of $\mu$ or the objective value of the rounded solution we obtain by from each $\mu_i'$ independently:
\[
\E_{ij} d_{ij}\cdot  \E_{\mu'_{ij}} |x_i - x_j| \cdot \Ind(\overline{E_i} \text{ or } \overline{E_j}) \leq \eps  \E_{ij} d_{ij}^2 \quad \text{ and } \quad
\E_{ij} d_{ij} \cdot \E_{\mu'_i \otimes \mu'_j} |x_i - x_j| \cdot \Ind(\overline{E_i} \text{ or } \overline{E_j}) \leq \eps  \E_{ij} d_{ij}^2 \, .
\]

At the same time, for the ``bulk part'' $|x_i - x_j| \cdot \Ind(E_i,E_j)$, there is a probability-$1$ upper bound of $\eps^{-\mathcal{O}(1)} (\sqrt{\E_{\mu_{ia}} |x_i - x_a|^2} + \sqrt{\E_{\mu_{ja}} |x_j - x_a|^2} )$, so we can bound
\begin{equation*}
    \begin{split}
        & (\text{rounding error for }   |x_i - x_j| \cdot \Ind(E_i,E_j))  \\ 
        & \qquad \leq \eps^{-\mathcal{O}(1)} \cdot \Paren{\sqrt{\E_{\mu_{ia}} |x_i - x_a|^2} + \sqrt{\E_{\mu_{ja}} |x_j - x_a|^2} } \cdot \tv(\mu_{ij}', \mu_i' \otimes \mu_j') \, .
    \end{split}
\end{equation*}
Putting things together, to get rounding error down to $\bigO{\eps \cdot \E_{ij} d_{ij}^2 }$ for the linear term $\E_{ij} d_{ij} |x_i - x_j|$, our task now boils down to proving a variant of the pinning lemma which offers a bound on a weighted average of $\tv(\mu'_{ij},\mu'_i \otimes \mu'_j)$.
Substituting an appropriately-chosen method of selecting $T$ into the potential-function argument used to prove Lemma~\ref{lem:pinning-lemma-intro}, along with a tighter bound on the \emph{effective alphabet size} dependence\footnote{Naively, we would still require $|T|$ to scale with $B$, which arises due to an initial bound on the average entropy of $x_i$. By instead analyzing relevant functions of the $x_i$ in the potential function argument, we avoid this dependence.} in the number of rounds of conditioning, we are able to show that after conditioning on $\delta^{-\mathcal{O}(1)}$ choices of $x_i$,
\[
\E_{ij} d_{ij} \cdot \Paren{\sqrt{\E_{\mu_{ij}} |x_i - x_a|^2} + \sqrt{\E_{\mu_{ij}} |x_j - x_a|^2} } \cdot \tv(\mu_{ij}', \mu_i' \otimes \mu_j') \leq \delta \E_{ij} d_{ij}^2\,
\]
at which point we can choose $\delta = \eps^{\mathcal{O}(1)}$ (see \cref{lem:rs-gcr-error} and the lemmas that lead up to it).\footnote{The specific approximate pairwise independence statement in~\cref{sec:rs} is slightly different then what is presented here. The above statements are simplified for the sake of exposition.}
The conclusion is that there is a set of $\eps^{-O(1)}$ indices $x_i$ we can condition on so that afterwards independently sampling each $y_i \sim \mu_i'$ leads to rounding error $ \bigO{ \eps \cdot  \E_{ij} d_{ij}^2}$ on the linear term.

\paragraph{The Quadratic Term.}
For the quadratic term, we can use another variant of the pinning lemma, due to \cite{Barak2011}, which shows via a distinct potential function argument for any pseudo-distribution $\mu$ on variables $x_1,\ldots,x_n$, there is a set $T \subseteq [n]$ of indices such that if $\mu'$ is the result of conditioning $\mu$ on $x_T = z_T$, with $z_T\sim \mu_T$,
\[
\E_{z_T \sim \mu_T} \E_{ij} |\Cov_{\mu'_{ij}}(x_i,x_j)| \leq \eps \cdot \E_{i\sim[n]} \Var_\mu(x_i) \, .
\]
In our case, if $\mu$ is the pseudo-distribution we obtain after the anchoring step, $\E_{i \sim [n]} \Var_{\mu}(x_i) \leq \E_{i \sim [n]} \E_{\mu_{ia}} |x_i - x_a|^2 \leq  \bigO{ \E_{ij} d_{ij}^2}$.
This means that by conditioning on $\eps^{-\mathcal{O}(1)}$ indices, we can push the rounding error on the quadratic term to $\eps \cdot \E_{ij} d_{ij}^2$.

To make the overall argument successfully, we have to show that we can condition on a single set of $\eps^{-\mathcal{O}(1)}$ indices to get small rounding error simultaneously for the linear and quadratic terms.
Luckily, we can combine the potential functions from both variants of the pinning lemma into a single potential function (see \cref{lem:gcr_linear_combo} for a formal statement).

\paragraph{From $n^B$ to $\poly(n,B)$ Running Time.} 
Now we describe how to improve the algorithm to run in truly polynomial time, for fixed $\eps > 0$.
There are two sources of the exponential dependence on $B$ in the analysis we described so far.
First, we discretized $\R$ into around $2^B / \eps$ points, meaning that our LP is at least exponential-in-$B$ sized.
Second, Lemma~\ref{lem:pinning-lemma-intro} requires taking our conditioning set $T$ to have size at least $\log(\text{support size of $x_i$}) \geq \log(2^B) \geq B$. 

To address the first issue, we show that there is actually a discretization of $\R$ containing only $\bigO{B/\eps}$ points which suffices to preserve solution value $\OPT_{\textrm{EMV}} + \eps \cdot \E_{ij} d_{ij}^2$.
We rely again on the anchoring argument to show that a grid with geometrically increasing interval length can be used, centered at a single anchor point -- this simple argument appears in Lemma~\ref{lem:discretization-rs}.

At this point, we could get an algorithm with running time $n^{\poly(\log B, \eps^{-1})}$ by using the rest of the analysis unchanged.
To get from quasi-polynomial to polynomial dependence on $B$, we use one more idea.
We apply the pinning lemma to a slightly different choice of variables -- rather than $x_i$, we use the ``truncated'' versions of $x_i$, i.e. $\tilde{x}_i = (x_i -\E x_i )\cdot E_i + \E x_i$.
This has the result that each $x_i$ now assumes values in a smaller set, and is enough to show that we can use $|T| = \eps^{-\mathcal{O}(1)}$, with no dependence at all on $B$, leading to our truly polynomial-time algorithm.

\subsection{Weighted $k$-Euclidean Metric Violation}
We turn next to Theorem~\ref{thm:RS-weighted-intro}, on weighted $k$-EMV with dense weights.
The algorithmic guarantees for Max-CSPs (whose toolkit we are borrowing here) differ little between complete instances, where all possible tuples of variables participate in a constraint, and densely-weighted instances. In fact, the analysis even extends to sparse instances, as long as the underlying constraint graph enjoys some mild expansion properties.
Since our algorithm relies on tools from the CSP literature, one might expect that similarly the analysis of our algorithm could remain essentially unchanged when generalizing to the weighted $k$-EMV case, but this is not to be the case for our algorithm.
Rather, the key anchoring step in our argument breaks down. 
To see why, consider two simple examples, both allowable weight matrices under our assumptions.

The first example: the weights $w_{ij} \in [0,1]$ are a disjoint union of two cliques, on $1,\ldots,n/2$ and $n/2+1,\ldots,n$.
Then consider the following distribution $\mu$ over optimal solutions $x_1,\ldots,x_n$ which we could be required to round.
Start with a single optimal solution $x_1,\ldots,x_n$, and pick two random signs $\sigma,\sigma'$.
Shift $x_1,\ldots,x_{n/2}$ by $2^B \sigma$ and $x_{n/2+1},\ldots,x_n$ by $2^B \sigma'$.
At best we can hope to anchor half of the points by conditioning on any single $x_a$.
In this example, the remedy is clear: condition on two points, $x_a,x_b$, to anchor both clusters.

The second example: the weights $w_{ij} \in [0,1]$ form an expander graph.
In this case our algorithm can be modified so that our anchoring step works by conditioning on $x_a$, but we need to replace our Sherali-Adams LP with a Sum of Squares SDP.
Once we do so, we can argue as follows.
If we condition on a randomly-chosen $x_a$, then a typical $i$ is anchored to a ball of radius around $\sqrt{\E_{i,a \sim [n]} \E_{\mu_{ia}} |x_i - x_a|^2}$.
Previously we were able to compare this quantity to the objective function to conclude a bound of $\bigO{\sqrt{\E_{ij} d_{ij}^2}}$.
But in the weighted case, the objective function only allows us to bound $\E_{i,a \sim w} \E_{\mu_{ia}} |x_i - x_a|^2$.
If $\mu$ were an actual distribution over $\R^n$, and $w$ a spectral expander, we could resolve our issue via
\[
  \E_{\mu} \E_{i,a \sim [n]} |x_i - x_a|^2 \leq \bigO{1} \cdot \E_{\mu} \E_{i,a \sim w} |x_i - x_a|^2 \, ,
\]
since the right-hand side is the quadratic form of the Laplacian of $w$.
It turns out that even if $\mu$ is an SoS-pseudo-distribution, this inequality still holds (see \cref{fact:nonnegative-quadratic}), allowing us to complete our anchoring argument in the expander case.\footnote{While the above sketch might make it sound like our analysis would go through if $w$ is any expander our analysis will work, and this is true of the anchoring step, the ``linear term'' step requires $w$ to be dense. Density can be circumvented in the corresponding step of Max-CSP algorithms \cite{Barak2011}, but we do not know how to translate the ``local-to-global'' technique used in those analyses to our setting.}

Now we return to the case of general weights $w$. We require a rather stringent assumption of $\delta$-regularity (every row sums to $\delta n$), which does not appear in the CSP literature. Under this assumption, we can use results on partitioning low threshold-rank graphs into expanders~\cite{Gharan2014} to show that the weight matrix, viewed as a graph, can be partitioned into $\delta^{-\mathcal{O}(1)}$ node-induced expanding subgraphs.
This implies that the ideas in the preceding examples are enough to show that by conditioning on $\delta^{-\mathcal{O}(1)}$ many $x_a$'s, we can obtain a conditioned pseudo-distribution $\mu'$ such that $\E_i \Var_{\mu'} (x_i) \leq  \delta^{-\mathcal{O}(1)} \cdot \E_{ij} d_{ij}^2$ (see \cref{cor:magnitude-points-wrs} for details). While this is larger than the initial variance bound for \textsf{$k$-EMV}, we can drive down the rounding error at the expense of running time, resulting in a $n^{\poly(B,\eps^{-1},\delta^{-1})}$ time algorithm.
The improved discretization argument we use in the non-weighted case to remove the single-exponential dependence on $B$, however, does not extend to the weighted case.

\section{Preliminaries}
\label{sec:prelims}

Throughout the paper we use the $\expecf{i,j\sim [n]}{\cdot}$ to denote $ \sum_{i,j \in [n]} (\cdot) / n^2$. For a $n \times n$ matrix $W$ with entries $0 \leq w_{i,j} \leq 1$, we use the notation $\expecf{i, j \sim W }{\cdot}$ to denote $ \sum_{i,j \in [n]} w_{ij}(\cdot) / ( \sum_{i,j \in [n]} w_{ij} )$. We also use $\norm{x} = \Paren{ \sum_{i} x_i^2 }^{1/2} $ to denote the $\ell_2$ norm of a vector and $\norm{\cdot}_p = \Paren{ \sum_{i} x_i^p }^{1/p} $ to denote the $\ell_p$ norm. We also use the notation $ \mathbb{1}[E]$ to denote the indicator for the event $E$ and $\overline{E}$ to denote the complement.  %

\paragraph{Information theory.}
For random variables $x, y$ with alphabet size $q$,  we use $\textsf{H}(x)$ to denote the $q$-ary entropy of $x$ and $\mathsf{I}(x;y) = \mathsf{H}(x) - \mathsf{H}(x \mid y)$ to denote mutual information. Assuming $x \sim \mu_1$ and $y \sim \mu_2$, we use the notation $\textsf{TV}(x,y) =\frac{1}{2} \sup_{ \norm{f}_{\infty} \leq 1 } \expecf{x \sim \mu_1}{ f(x) } -  \expecf{y \sim \mu_2}{ f(y) } $ to denote total variation distance between $\mu_1$ and $\mu_2$.

\begin{fact}[Pinsker's Inequality]
\label{fact:pinsker's}
Given random variables $x\sim \mu_1$ and $y \sim \mu_2$, $\tv(\mu_1, \mu_2 )^2  \leq \mathsf{I}(x, y)$.
\end{fact}

Next, recall the following standard fact for lower bounding the decrease in variance between random variables $u, v$ by conditioning on one of them. 

\begin{fact}[Variance Reduction]
\label{fact:var-reduction}
    Given scalar random variables $u, v$ drawn from the joint distribution $\calD$, the decrease in variance of $u$ when conditioning on $v$ can be lower bounded as follows:
    \begin{equation*}
        \Var(u) - \E_v \Var(u \vert v) \geq \frac{\Cov^2(u, v)}{4\cdot \Var(v)}\,.
    \end{equation*}
\end{fact}

\paragraph{Sherali-Adams Linear Programming.}
We review the local-distribution view of the Sherali-Adams linear programming hierarchy.
Let $\Omega$ be a finite set and $n \in \N$.
For $t \in \N$, the level-$t$ Sherali-Adams LP relaxation of $\Omega^n$ is as follows.
(For a more formal account, see e.g. \cite{Barak2011}.)

\paragraph{Variables and constraints.} For each $T \in \binom{n}{t}$, we have $|\Omega|^t$ non-negative variables, $\mu_T(x)$ for $x \in \Omega^t$, collectively denoted $\mu_T$.
For each $T$, we constrain $\sum_x \mu_T(x) = 1$ so that $\mu_T$ is a probability distribution over $\Omega^t$.
Then, for every $S,T \in \binom{n}{t}$, we add linear constraints so that the distribution $\mu_T$ restricted to $S \cap T$ is identical to the distribution $\mu_S$ restricted to $S \cap T$.
Overall, we have $(|\Omega|n)^{O(t)}$ variables and constraints.
We sometimes call the collection of local distributions $\{ \mu_T \}$ satisfying these constraints a ``degree-$t$ Sherali-Adams pseudo-distribution''.

\paragraph{Pseudo-expectation view.}
If we have a Sherali-Adams solution which is clear from context, for any function $f$ of at most $t$ variables out of $x_1,\ldots,x_n$ (a.k.a. ``$t$-junta''), we write $\pE f(x_1,\ldots,x_n)$ for $ \E_{x \sim \mu_T} f(x_T)$, where $\mu_T$ is the local distribution on the subset $T \subseteq [n]$ of variables on which $f$ depends.
We can extend $\pE$ to a linear operator on the linear span of $t$-juntas.
Observe that $\pE f$ is a \emph{linear} function of the underlying LP variables.

\paragraph{Conditioning.}
We will frequently pass from a Sherali-Adams solution $\{\mu_T\}$ to a \emph{conditional} solution $\{\mu_T \, | \, x_{i_1} = a_1,\ldots,x_{i_{t'}} = a_{t'} \}$, for $t' < t$.
We obtain the conditional local distribution on $x_{j_1},\ldots,x_{j_{t-t'}}$ by conditioning the local distribution on $x_{j_1},\ldots,x_{j_{t-t'}}, x_{i_1},\ldots,x_{i_{t'}}$ on the event $\{ x_{i_1} = a_1,\ldots,x_{i_{t'}} = a_{t'} \}$.

\begin{definition}
\label{def:product-marginals-pe}
    Let $\mu$ be a pseudo-distribution with indeterminates $x_1, \ldots, x_n$ where $x_i$ is $\R^k$-valued.
    We define $\mu^{\otimes}$ to be the distribution which is the independent product of marginals on each $x_i$. 
\end{definition}

\begin{definition}
    Let $\mu$ be a pseudo-distribution on $x_1,\ldots,x_n$.
    For any set $\calT \subseteq [n]$, we define $\mu_{\calT}$ to be the pseudo-distribution generated by sampling $\hat{x_i}$ jointly for all $i \in \calT$ and conditioning on $x_i = \hat{x_i}$ for $i \in \calT$.
    Hence $\mu_\calT$ is itself a random variable, assuming values in the set of pseudo-distributions.
    Furthermore, let $\E_{\calT}$ denote the expectation over the randomness in the sampling of the values of $x_i$ for $i \in \calT$.
\end{definition}

 Note that since conditioning on a value drawn from the marginal of a single variable maintains the pseudo-distribution in expectation, we have the following fact:

 \begin{fact}
     For any pseudo-distribution $\mu$ and any set $\calT \subseteq [n]$ we have that $\E_{\calT} \pE_{\mu_{\calT}}  = \pE_{\mu}$.
 \end{fact}

We note that we can extend the pseudo-expectation view to the Lasserre/sum-of-squares hierarchy (see~\cite{barak2016proofs, fleming2019semialgebraic} for a formal treatment). 

\begin{definition}[Constrained pseudo-distributions]
\label{def:constrained-pseudo-distributions}
Let $\calA = \Set{ p_1\geq 0 , p_2\geq0 , \dots, p_r\geq 0}$ be a system of $r$ polynomial inequality constraints of degree at most $d$ in $m$ variables.
Let $\mu$ be a degree-$\ell$ pseudo-distribution over $\mathbb{R}^m$.
We say that $\mu$ \emph{satisfies} $\calA$ at degree $\ell \ge1$ if for every subset $\calS \subset [r]$ and every sum-of-squares polynomial $q$ such that $\deg(q) + \sum_{i \in \calS } \max\Paren{ \deg(p_i), d} \leq \ell$, $\pexpecf{\mu}{ q \prod_{i \in \calS} p_i } \geq 0$.
Further, we say that $\mu$ \emph{approximately satisfies} the system of constraints $\calA$ if the above inequalities are satisfied up to additive error $\pexpecf{\mu}{ q \prod_{i \in \calS} p_i } \geq -2^{-n^{\ell} } \norm{q} \prod_{i \in \calS} \norm{p_i}$, where $\norm{\cdot}$ denotes the Euclidean norm of the coefficients of the polynomial, represented in the monomial basis.  
\end{definition}

Crucially, there's an efficient separation oracle for moment tensors of constrained pseudo-distributions. 

\begin{fact}[\cite{shor1987approach, parrilo2000structured}]
    \label{fact:sos-separation-efficient}
    For any $m,\ell \in \N$, the following convex set has a $m^{\bigO{\ell}}$-time weak separation oracle, in the sense of \cite{grotschel1981ellipsoid}:\footnote{
        A separation oracle of a convex set $S \subset \R^M$ is an algorithm that can decide whether a vector $v \in \R^M$ is in the set, and if not, provide a hyperplane between $v$ and $S$.
        Roughly, a weak separation oracle is a separation oracle that allows for some $\eps$ slack in this decision.
    }:
    \begin{equation*}
        \Set{  \pexpecf{\mu(x)} { (1,x_1, x_2, \ldots, x_m)^{\otimes \ell } } \Big\vert \text{ $\mu$ is a degree-$\ell$ pseudo-distribution over $\R^m$}}
    \end{equation*}
\end{fact}

Given a system of polynomial constraints, denoted by $ \calA$, we say that it is \emph{explicitly bounded} if it contains a constraint of the form $\{ \|x\|^2 \leq 1\}$. Then, the following fact follows from  \cref{fact:sos-separation-efficient} and \cite{grotschel1981ellipsoid}:

\begin{theorem}[Efficient optimization over pseudo-distributions]
    \label{fact:eff-pseudo-distribution}
There exists an $(m+r)^{O(\ell)} $-time algorithm that, given any explicitly bounded and satisfiable system $ \calA$ of $r$ polynomial constraints in $m$ variables, outputs a degree-$\ell$ pseudo-distribution that satisfies $ \calA$ approximately, in the sense of~\cref{def:constrained-pseudo-distributions}.\footnote{
    Here, we assume that the bit complexity of the constraints in $ \calA$ is $(m+t)^{O(1)}$.
}
\end{theorem}
\begin{remark}[Bit complexity and approximate satisfaction]
    \label{rmk:tedium}
    We will eventually apply this result to a constraint system that can be defined with numbers with $\log(t)$ bits, where $t$ is the sample complexity of the algorithm (scaling polynomially with $\terms$).
    Consequently, we can run this algorithm efficiently, and the errors incurred here (which is exponentially small in $\qubits$) can be thought of as a ``machine precision'' error, and is dominated by the sampling errors incurred elsewhere.
    We can therefore safely ignore precision issues in the rest of our proof.

    The pseudo-distribution $D$ found will satisfy $\calA$ only approximately, but provided the bit complexity of the sum-of-squares proof of $\calA \sststile{r'}{} B$, i.e.\ the number of bits required to write down the proof, is bounded by $\terms^{\bigO{\ell}}$ (assuming that all numbers in the input have bit complexity $\terms^{\bigO{1}}$), we can compute to sufficiently good error in polynomial time that the soundness will hold approximately.
    All of our sum-of-squares proofs will have this bit complexity.
\end{remark}

Finally, 
we require that any spectral inequality admits a sum-of-squares proof:

\begin{fact}[Quadratic polynomial inequalities admit SoS proofs (see Fact 2.35 in~\cite{bakshi2024learning})]
    \label{fact:nonnegative-quadratic}
Let $p$ be a polynomial in the indeterminates $x \in \R^m$ such that $p$ has degree $2$ and $p \geq 0$ for all $x \in \mathbb{R}^m$. Then $\sststile{2}{x} \Set{ p(x) \geq 0  }$. 
\end{fact}

\section{A Global Correlation Framework}

In this section, we describe a simple framework for analyzing global correlation rounding.
This framework allows us a technically convienient way to obtain pseudoexpectations which simultaneously satisfy several distinct notions of approximate pairwise independence.

\begin{definition}[Pseudo-distribution potentials]
\label{def:pseudo-distribution-potential}
    A function $\Phi$ from pseudo-distributions over $\mathbb{R}^n$ to $\mathbb{R}$ is a pseudo-distribution potential if
    \begin{enumerate}
        \item For all pseudo-distributions $\mu$ we have that $\Phi(\mu) \geq 0$.
        \item For a fixed $i \in [n]$, let $\mu'$ be the pseudo-distribution given by sampling a random value $\hat{x}_i$ from the marginal of $x_i$ and conditioning on the event $\Set{ x_i = \hat{x}_i}$. Then,
        \[ \E_{x_i} \left[\Phi(\mu) - \Phi(\mu')\right] \geq 0\,.\]
    \end{enumerate}
\end{definition}

The typical potential functions used in global correlation rounding by \cite{Raghavendra2012,Barak2011} are $\E_{i \sim [n]} H(\{x_i\})$ and $\E_{i \sim [n]} \Var(x_i)$.
Both are pseudo-distribution potentials in the sense of Defintion~\ref{def:pseudo-distribution-potential}.

\begin{definition}[Potential aligned functions]
\label{def:potential-aligned}
    A function $f: (\mathbb{R}^k)^n \rightarrow \mathbb{R}$ is $(\delta, \delta')$-potential aligned with a pseudo-distribution potential $\Phi$, initial pseudo-distribution $\mu_0$, and family of polynomials $p_1,\ldots,p_m$ which have $\pE_\mu p_i = 0$ if for all distributions over pseudo-distributions $\mu$ such that 
    \begin{enumerate}
        \item $\E_{\{\mu\}} \mu = \mu_0$, 
        \item all $\mu'$ in the support of $\{\mu\}$ satisfy $\pE_{\mu'} p_i = 0$ for all $i \leq m$, and 
        \item $ \left\vert \E_{\{\mu\}} \pE_{\mu} f(x_1,\ldots, x_n) - \pE_{\mu^{\otimes}} f(x_1,\ldots, x_n) \right\vert \geq \delta\,,$
    \end{enumerate}
    we have that there exists an index $i \in [n]$ such that
    \[ \E_{ \{ \mu \}} \E_{\hat{x_i} \sim \{x_i \} } \left[\Phi(\mu) - \Phi(\mu') \right] \geq \delta'\,,\]
    where $\mu'$ is the pseudo-distribution given by sampling a random value $\hat{x}_i$ from the marginal of $x_i$ and conditioning on $x_i = \hat{x}_i$.
\end{definition}

In our language, prior works like \cite{Barak2011, Raghavendra2012} can be interpreted as showing that any  function $f$ of the form $\E_{i,j \sim [n]} f_{ij}(x_i,x_j)$ where $f_{ij}$'s are $[-1,1]$-valued is potential-aligned with respect to $\Phi = \E_{i} H(\{x_i\})$. 
And, functions of the form $f = \E_{i,j \sim [n]} A_{ij} \langle x_i, x_j\rangle$ are potential-aligned with respect to $\Phi = \E_{i \sim [n]} \Var(x_i)$; this follows almost directly from prior work on global correlation rounding (see~\cite{Barak2011}). 

\begin{lemma}[Degree 2 Polynomials are Potential Aligned]
\label{fact:deg-2-potential-aligned}
    Let $\epsilon > 0$, $k \in \mathbb{N}^{>0}$, and let $\mu_0$ be a pseudo-distribution over $\left(\mathbb{R}^k\right)^n$ such that $\E_i \tr(\widetilde{\Sigma}_i) \leq \eta$, where $\widetilde{\Sigma} \in \mathbb{R}^{k \times k}$ is the pseudo-covariance of $x_i$ with respect to $\mu_0$. Then for any matrix $A \in \R^{n \times n}$, the function $\E_{i,j \sim [n]} A_{ij} \langle x_i, x_j\rangle$ is 
    \[
    \Paren{ \epsilon \cdot \E_{i,j \sim[n]} A_{ij}^2 , \hspace{0.1in} \frac{\eps^2}{ \eta  k} \cdot  \E_{i,j \sim [n]} A_{ij}^2 }-\textrm{potential aligned}
    \]
    for the potential $\Phi = \E_i \tr(\widetilde{\Sigma}_i)$ and pseudo-distribution $\mu$.
\end{lemma}
\begin{proof}
Consider the following potential:
    \[ \Phi(\mu) = \E_{i\sim[n]} \tr(\tilde{\Sigma}_i(\mu)\,,\]
    where $\tilde{\Sigma}_i(\mu)$ is the pseudo-covariance of $x_i$ with respect to the pseudo-distribution $\mu$.
    Let $\{ \mu\}$ be a distribution over pseudo-distributions satisfying the hypotheses of the Defintion~\ref{def:potential-aligned}.
    We will aim to show that whenever the difference in expectation under $\mu$ versus $\mu^{\otimes}$ is at least $\epsilon \cdot \E_{i,j\sim[n]} A_{ij}^2$, then choosing $j$ at random and conditioning the value of $x_j$ significantly reduces the potential function $\Phi$ in expectation. In particular, this implies that there must also be a fixed $x_j$ which also achieves the same decrease in expectation over the random sample of $x_j$.
    
    For any fixed $\mu$, we have
    \begin{equation*}
        \Abs{  \E_{i,j\sim[n]} \pE_{\mu^{\otimes}} A_{ij} \langle x_i, x_j\rangle - \E_{i,j\sim[n]} \pE_\mu A_{ij} \langle x_i, x_j\rangle } \leq \E_{i,j\sim[n]} |A_{ij}| \sum_{\ell \in [k]} \Abs{\widetilde{\Cov}_\mu ((x_i)_\ell, (x_j)_\ell)}\,,
    \end{equation*}
    where we use the notation $\widetilde{\Cov}_\mu$ to denote the pseudo-covariance.
    Applying~\cref{fact:var-reduction} we have that
    \begin{equation}
    \label{eqn:canonical-two-cauchy-schwarz}
        \begin{split}
            & \E_{i,j\sim[n]}   |A_{ij}| \cdot \sum_{\ell \in [k]}  \Abs{\widetilde{\Cov} ((x_i)_\ell, (x_j)_\ell)} \\
            & \leq \E_{i,j\sim[n]} \left[|A_{ij}| \cdot \sum_{\ell \in [k]}  \E_{(x_j)_\ell \sim \mu_{j\ell}}\left[\Var_\mu((x_i)_\ell) - \Var_\mu((x_i)_\ell \vert (x_j)_\ell)\right] ^{1/2} \cdot \Var_\mu((x_j)_\ell)^{1/2} \right] \\
            & \leq \expecf{i \sim [n]}{  \Paren{ \expecf{j \sim [n]}{ |A_{ij}| \sum_{\ell \in [k]}  \E_{(x_j)_\ell \sim \mu_{j \ell}}\left[\Var_\mu((x_i)_\ell) - \Var_\mu((x_i)_\ell \vert (x_j)_\ell)\right]^{\frac12}  } }^2   }^{\frac12 } \expecf{j \sim [n]}{ \sum_{\ell \in [k]}\Var_\mu((x_j)_\ell) }^{\frac12}  \\
            & \leq \underbrace{\expecf{i \sim [n]}{ \expecf{j\sim [n]}{ A_{ij}^2} \expecf{j \sim [n]}{ \Paren{\sum_{\ell \in [k]} \E_{(x_j)_\ell \sim \mu_{j \ell}}\left[\Var_\mu((x_i)_\ell) - \Var_\mu((x_i)_\ell \vert (x_j)_\ell)\right]^{\frac12} }^2 }  }^{\frac12}}_{\eqref{eqn:canonical-two-cauchy-schwarz}.(1) }  \underbrace{ \expecf{j \sim [n]}{ \sum_{\ell \in [k]}\Var_\mu((x_j)_\ell) }^{\frac12}}_{\eqref{eqn:canonical-two-cauchy-schwarz}.(2)}  \,.
        \end{split}
    \end{equation}
    Writing the sum over $\ell$ as an average any applying Jensen's we can bound term \eqref{eqn:canonical-two-cauchy-schwarz}.(1) as follows:
    \begin{equation}
    \begin{split}
       \eqref{eqn:canonical-two-cauchy-schwarz}.(1) & \leq \sqrt{k} \cdot \expecf{i \sim [n]}{ \expecf{j\sim [n] }{A_{ij}^2 } \expecf{j \sim [n]}{\sum_{\ell \in [k] }{\expecf{(x_j)_\ell \sim \mu_{j \ell}}{\Var_\mu((x_i)_\ell) - \Var_\mu((x_i)_\ell \vert (x_j)_\ell) } } } }^{\frac12} \\
       & \leq  \sqrt{k} \cdot \expecf{i \sim [n]}{ \expecf{j\sim [n] }{A_{ij}^2 } \expecf{j \sim [n]}{ \tr(\widetilde{\Sigma}_i (\mu))  - \E_{x_j} \tr(\widetilde{\Sigma}_j(\mu \, | \, x_j)) } }^{\frac12} 
       \end{split}
    \end{equation}
    where the last inequality follows from the definition of trace and the fact that conditioning $(x_i)_\ell$ on $(x_j)_{\ell'}$ for $\ell' \neq \ell$ only decreases variance on average.

    Furthermore, by assumption we have that $ \E_{\{\mu\}} \eqref{eqn:canonical-two-cauchy-schwarz}.(2) \leq \left( \E_{\{\mu\}} \E_{i}\sum_{\ell \in [k]} \Var_\mu((x_j)_\ell)\right)^{1/2} \leq \sqrt{\eta}$. Thus, we have that when the difference in expectation of $\E_{i,j \sim [n]} A_{ij} \langle x_i, x_j \rangle$ under $\mu$ versus $\mu^{\otimes}$ is at least $\epsilon \cdot \expecf{i,j \sim [n]} {A_{ij}^2}$, then 
    \[ \epsilon \cdot \expecf{i,j \sim [n]}{ A_{ij}^2} \leq \sqrt{\eta \cdot k} \cdot \expecf{i \sim [n]}{ \expecf{j\sim [n] }{A_{ij}^2 } \expecf{j \sim [n]}{ \tr(\widetilde{\Sigma}_i(\mu))  - \E_{x_j} \tr(\widetilde{\Sigma}_i(\mu \mid x_j))} }^{\frac12}  \,.\]
    Note that we can rewrite the LHS expectation over $i$ drawn uniformly from $[n]$ instead as an expectation over $i$ drawn from $\psi$ which samples $i$ proportional to $\E_{j \sim [n]} A_{ij}^2$. This gives us that
    \[ \epsilon \expecf{i,j \sim [n]} {A_{ij}^2} \leq \sqrt{\eta \cdot k} \cdot \left( \expecf{i,j \sim [n]} {A_{ij}^2} \right)^{1/2} \left(\E_{i \sim \psi}  \left( \E_{j \sim [n]}  \tr(\widetilde{\Sigma}_j)  - \E_{x_j} \tr(\widetilde{\Sigma}_i \vert x_j)\right)\right)^{1/2}\,\]
    and by rearranging that
    \[ \expecf{i \sim \psi }{ \expecf{j \sim [n]}{ \tr(\widetilde{\Sigma}_i(\mu))  - \E_{x_j} \tr(\widetilde{\Sigma}_i(\mu \mid x_j))} }  \geq \Paren{ \frac{\epsilon^2 }{\eta k} } \cdot  \expecf{i,j\sim [n]}{A_{ij}^2} \,.\]
    Note that this implies that the expected decrease in potential when conditioning on some $i$ from this distribution is at least $\Paren{ \frac{\epsilon^2 }{\eta k} } \cdot  \expecf{i,j\sim [n]}{A_{ij}^2} $ and thus there must exist some $j$ such that when conditioning on $x_j$ the expected decrease over the randomness in the sample of $x_j$ is sufficiently large. Thus, we have shown that this function is potential aligned.
\end{proof}

We now show how to combine different potential functions to analyze correlation rounding.

\begin{lemma}[Handling several potential aligned functions]
\label{lem:gcr_linear_combo}
    Let $\epsilon \geq 0$. Let $\{f_{j}\}_{j \in [\ell]}$ be a collection of functions from $(\mathbb{R}^k)^n$ such that $f_j$ is $(\epsilon/\ell, \delta_j)$-potential aligned for some psuedo-distribution potential $\Phi_j$ and some initial pseudo-distribution $\mu$. Then there exists a set $\calT \subseteq [n]$ of size at most $\sum_{j \in [\ell]} \Phi_j(\mu)/\delta_j$ such that after conditioning on the values of $x_i \in \calT$,
    \[ \left\vert \E_{\{\mu_\calT\}} \pE_{\mu_{\calT}} \sum_{j \in [\ell]} f_j(x_1,\ldots, x_n) - \pE_{\mu^{\otimes}_{\calT}} \sum_{j \in [\ell]} f_j(x_1,\ldots, x_n) \right\vert \leq \epsilon \,.\]
\end{lemma}
\begin{proof}
    Suppose $\calT \subseteq [n]$ is such that
    \[ \left\vert \E_{\{\mu_\calT\}} \pE_{\mu_{\calT}} \sum_{j \in [\ell]} f_j(x_1,\ldots, x_n) - \pE_{\mu^{\otimes}_{\calT}} \sum_{j \in [\ell]} f_j(x_1,\ldots, x_n) \right\vert \geq \epsilon \,.\]
    Then, by averaging, there exists $j \in [\ell]$ such that
    \[ \left\vert \E_{\{\mu_\calT\}} \pE_{\mu_{\calT}} f_j(x_1,\ldots, x_n) - \pE_{\mu^{\otimes}_{\calT}} f_j(x_1,\ldots, x_n) \right\vert \geq \epsilon/\ell \,.\]
    Since $f_j$ is $(\epsilon/\ell, \delta_i)$ potential aligned and by definition of conditioning $\{ \mu_\calT \} = \mu$, by~\cref{def:potential-aligned} we have that there exists $i \in [n]$ such that
    \[ \E_{\{ \mu_\calT \}} \E_{x_i} \left[\Phi_j(\mu_\calT) - \Phi_j(\mu'_\calT) \right] \geq \delta_j\,,\]
    where $\mu'_{\calT}$ is $\mu_{\calT}$ conditioned on $x_i$.
    We will pick this $i$ to add to the set $\calT$ which we condition on. Since all $\Phi_j$ are psuedo-distribution potentials, conditioning on $x_i$ does not increase any of them in expectation. Thus, at each step of conditioning the potential functions either decrease by $\delta_j$ in expectation or are non-increasing. Therefore, each $f_j$ can only have a large difference in expectation on $\mu$ vs $\mu^{\otimes}$ at most $\Phi_j(\mu)/\delta_j$ many times, since after this happens mores times then the expectation of $\Phi_j$ is negative, which contradicts the fact that pseudo-distribution potential are always positive.
\end{proof}

\section{Euclidean Metric Violation}
\label{sec:rs}

In this section, we will show that there exist additive approximation for the \textsf{$k$-EMV} objective (\cref{problem:rs}) and its weighted variant (under regularity assumptions). Our algorithm is based on global correction, and the key technical innovation in both results is showing that Lipschitz functions are \emph{potential aligned}.

\subsection{Rounding Lipschitz Functions}
 In this section, we will show that Lipschitz functions are potential aligned, for a carefully chosen variant of the average entropy potential. These functions will appear in both the weighted and unweighted Raw Stress objective functions and showing that they satisfy the potential aligned definition will allow us to control their rounding error via~\cref{lem:gcr_linear_combo}.

\begin{lemma}[Lipschitz Functions are Potential Aligned]
\label{lem:lipschitz-potential-aligned}
    Let $k \in \mathbb{N}^{>0}$. Let $\{f_{ij}\}_{i,j \in [n]} : \left( \mathbb{R}^k \right)^2 \rightarrow \mathbb{R}$ be a collection of Lipschitz functions with Lipschitz constants $L_{ij}$. Let $\mu$ be a pseudo-distribution, and for all $i \in [n]$, let $E_i$ be the event that $\norm{x_i - \pE_{\mu} x_i} \leq C \cdot \sqrt{\pE_{\mu} \norm{x_i - \pE_{\mu} x_i}^2}$ and let 
    \[\widetilde{x_i} = x_i E_i + (\pE_{\mu}[x_i] ) \overline{E}_i \,.\]
    Then $\expecf{i,j \sim [n]}{f_{ij}(x_i, x_j)}$ is $$\Paren{ \bigO{\epsilon^{1/4}} \cdot \left(\E_{i \sim [n]} \pE_\mu \norm{x_i - \pE_\mu x_i}^2\right)^{1/2} \cdot \left( \E_{i,j\sim[n]} L_{ij}^2\right)^{1/2},  \hspace{0.1in}\epsilon }-\textrm{potential aligned}$$  for the average entropy potential, $\Phi = \expecf{i\sim [n]}{ \mathsf{H}(\tilde{x}_i)}$, and $\mu$. %
\end{lemma}

We will prove this via showing both that (a) the difference in expectation on $\mu$ versus $\mu^{\otimes}$ is tied to an approximate pairwise independence property, and then (b) showing that whenever this property fails the potential in question decreases in expectation when conditioning on some carefully chosen $x_i$.

\begin{lemma}[Characterizing Lipschitz Rounding Error]
\label{lem:lipschitz-rounding-error}
    Let $0 < \epsilon \leq 1$ and let $\mu$ be some pseudo-distribution. Suppose $\mu'$ is a random psuedodistribution such that $\E \mu' = \mu$.
    Let $f_{ij} \, : \, \left(\R^k\right)^2 \rightarrow \R$ be $L_{ij}$-Lipschitz.
    Further, let $E_i$ be the event that $\norm{x_i - \pE_\mu x_i} \leq C \cdot \sqrt{\pE_\mu \norm{x_i - \pE_\mu x_i}^2}$ and let 
    $$\widetilde{x_i} = x_i E_i + (\pE_{\mu}[ x_i] ) \overline{E}_i .$$
    And suppose that 
    \[ \E_{\{\mu'\}} \E_{i \sim [n]} \left[\left( \E_{j \sim [n]} L_{ij}^2 \right) \left( \E_{j \sim [n]} \tv \Paren{ \mu'_{\{\widetilde{x_i}, \widetilde{x_j}\}}, \mu'_{\{\widetilde{x_i}\} \otimes \{ \widetilde{x_j} \}}}^2 \right)\right] \leq \eps \cdot \E_{i,j \sim [n]} L_{ij}^2\,.\]
    Then
    \begin{align*}
        &\left\vert \E_{\{\mu'\}} \E_{i,j \sim [n]} \left[ \pE_{\mu'} f_{ij}(x_i, x_j) - \pE_{(\mu')^{\otimes}}f_{ij}(x_i, x_j) \right] \right\vert \\
        &\qquad\leq  \bigO{ 1/C + C \sqrt{\epsilon}} \cdot \left(\E_{i \sim [n]} \pE_\mu \norm{x_i - \pE_\mu x_i}^2\right)^{1/2} \cdot \left( \E_{i,j\sim[n]} L_{ij}^2\right)^{1/2} \,.
    \end{align*}
\end{lemma}

In order to prove~\cref{lem:lipschitz-rounding-error}, we will need to bound the difference in expectation between $f_{ij}$ under $x_i$ versus under $\tilde{x}_i$. 

\begin{lemma}[Truncated Variables can be rounded]
\label{lem:tilde-vs-not-diff}
    Let $0 < \epsilon \leq 1$ and let $\mu$ be some pseudo-distribution. Suppose $\mu'$ is a random psuedodistribution such that $\E \mu' = \mu$.
    Let $f_{ij} \, : \, \R^k  \times  \R^k  \rightarrow \R$ be $L_{ij}$-Lipschitz.
    Then
    \begin{align*}
        \left\vert \E_{\{\mu'\}} \E_{i,j \sim [n]} \pE_{\mu'} \left[\left( f_{ij}(x_i, x_j) - f_{ij}(\widetilde{x}_i,\widetilde{x}_j)\right)\right] \right\vert\ \leq O(1/C) \cdot \left(\E_{i \sim [n]} \pE_\mu \norm{x_i - \pE_\mu x_i}^2\right)^{1/2} \cdot \left( \E_{i,j\sim[n]} L_{ij}^2\right)^{1/2}
    \end{align*}
    and 
    \begin{align*}
        \left\vert \E_{\{\mu'\}} \E_{i,j \sim [n]} \E_{(\mu')^\otimes} \left[\left( f_{ij}(x_i, x_j) - f_{ij}(\widetilde{x}_i,\widetilde{x}_j)\right)\right] \right\vert\ \leq O(1/C) \cdot \left(\E_{i \sim [n]} \pE_\mu \norm{x_i - \pE_\mu x_i}^2\right)^{1/2} \cdot \left( \E_{i,j\sim[n]} L_{ij}^2\right)^{1/2}\,.
    \end{align*}
\end{lemma}

\begin{proof}
    We will bound 
    \[ \left\vert \E_{\{\mu'\}} \E_{i,j \sim [n]} \pE_{\mu'} \left[ \left( f_{ij}(\widetilde{x}_i,\widetilde{x}_j) - f_{ij} (x_i, x_j) \right) \right] \right\vert\,,\] 
    and the expression with $(\mu')^{\otimes}$ can be bounded via the same sequence of inequalities. We break up this term as follows:
    \begin{equation}
    \label{eqn:introduce-indicators-for-E}
    \begin{split}
        \left\vert \E_{\{\mu'\}} \E_{i,j \sim [n]} \pE_{\mu'} \left[ \left( f_{ij}(\widetilde{x}_i,\widetilde{x}_j) - f_{ij} (x_i, x_j) \right) \right] \right\vert &\leq \underbrace{ \left\vert \E_{\{\mu'\}} \E_{i,j \sim [n]} \pE_{\mu'} \left[ \mathbb{1}[E_i \wedge E_j]\left( f_{ij}(\widetilde{x}_i,\widetilde{x}_j) - f_{ij} (x_i, x_j) \right) \right] \right\vert}_{\eqref{eqn:introduce-indicators-for-E}.(1) } \\
        &\hspace{0.2in} \quad+ \underbrace{ \left\vert \E_{\{\mu'\}} \E_{i,j \sim [n]} \pE_{\mu'} \left[ \mathbb{1}[\overline{E}_i \vee \overline{E}_j]\left( f_{ij}(\widetilde{x}_i,\widetilde{x}_j) - f_{ij} (x_i, x_j) \right) \right] \right\vert}_{\eqref{eqn:introduce-indicators-for-E}.(2)} \,.
    \end{split}
    \end{equation}
    Note that when $E_i$ and $E_j$ both occur then $\widetilde{x}_i = x_i$ and $\widetilde{x}_j = x_j$, so thus $\eqref{eqn:introduce-indicators-for-E}.(1) =0$, and it remains to bound term \eqref{eqn:introduce-indicators-for-E}.(2). Using the Lipschitz property of $f_{ij}$ we have:
    \begin{align*}
        &\left\vert \E_{\{\mu'\}} \E_{i,j \sim [n]} \pE_{\mu'} \left[ \mathbb{1}[\overline{E}_i \vee \overline{E}_j]\left( f_{ij}(\widetilde{x}_i,\widetilde{x}_j) - f_{ij} (x_i, x_j) \right) \right] \right\vert \\
        &\leq \E_{\{\mu'\}} \E_{i,j \sim [n]} \pE_{\mu'} \left[ \mathbb{1}[\overline{E}_i \vee \overline{E}_j]\left( \vert f_{ij}(\widetilde{x}_i,\widetilde{x}_j) \vert + \vert f_{ij} (x_i, x_j) \vert \right) \right] \\
        &\leq \E_{\{\mu'\}} \E_{i,j \sim [n]} \pE_{\mu'} L_{ij} \left[ \mathbb{1}[\overline{E}_i \vee \overline{E}_j]\left( \norm{\widetilde{x_i} - \pE_\mu x_i} + \norm{\widetilde{x_j} - \pE_\mu x_j} + \norm{x_i - \pE_\mu x_i} + \norm{x_j - \pE_\mu x_j} \right) \right]\,.
    \end{align*}
    Note that since $\widetilde{x_i}$ just contracts values towards $\pE_\mu x_i$ when they fall outside some range, we have that $\norm{\widetilde{x_i} - \pE_\mu x_i} \leq \norm{x_i - \pE_\mu x_i}$ and thus
    \begin{equation}
    \label{eqn:}
    \begin{split}
        &\left\vert \E_{\{\mu'\}} \E_{i,j \sim [n]} \pE_{\mu'} \left[ \mathbb{1}[\overline{E}_i \vee \overline{E}_j]\left( f_{ij}(\widetilde{x}_i,\widetilde{x}_j) - f_{ij} (x_i, x_j) \right) \right] \right\vert \\
        &\leq O(1) \cdot \E_{\{\mu'\}} \E_{i,j \sim [n]} \pE_{\mu'} L_{ij} \left[ \mathbb{1}[\overline{E}_i \vee \overline{E}_j]\left(\norm{x_i - \pE_\mu x_i} + \norm{x_j - \pE_\mu x_j} \right) \right] \\
        &\leq O(1) \cdot \E_{\{\mu'\}} \E_{i,j \sim [n]} \pE_{\mu'} L_{ij} \left[ (\mathbb{1}[\overline{E}_i] + \mathbb{1}[ \overline{E}_j]) \cdot \left(\norm{x_i - \pE_\mu x_i} + \norm{x_j - \pE_\mu x_j} \right) \right] \,.
    \end{split}
    \end{equation}
    We focus on bounding
    \[ \E_{\{\mu'\}} \E_{i,j \sim [n]} \pE_{\mu'} L_{ij} \left[ \mathbb{1}[\overline{E}_i]  \cdot \left(\norm{x_i - \pE_\mu x_i} + \norm{x_j - \pE_\mu x_j} \right) \right]\,,\]
    and the other term follows analogously. Note that we have that by Cauchy-Schwarz,
    \begin{align*}
        &\E_{\{\mu'\}} \E_{i,j \sim [n]} \pE_{\mu'} L_{ij} \left[ \mathbb{1}[\overline{E}_i]  \cdot \left(\norm{x_i - \pE_\mu x_i} + \norm{x_j - \pE_\mu x_j} \right) \right] \\
        &\leq O(1) \cdot \E_{ij} L_{ij} \left( \E_{\{\mu'\}(T)} \pE_{\mu'} \mathbb{1}[\overline{E}_i]^2 \right)^{1/2} \left( \E_{\{\mu'\}} \pE_{\mu'} \norm{x_i - \pE_\mu x_i}^2 + \norm{x_j - \pE_\mu x_j}^2\right)^{1/2}\,.
    \end{align*}
    Furthermore, since $\E_{\{\mu'\}} \pE_{\mu'} = \pE_{\mu}$ we have that
    \begin{align*}
        &\E_{\{\mu'\}} \E_{i,j \sim [n]} \pE_{\mu'} L_{ij} \left[ \mathbb{1}[\overline{E}_i]  \cdot \left(\norm{x_i - \pE_\mu x_i} + \norm{x_j - \pE_\mu x_j} \right) \right] \\
        &\leq O(1) \cdot \E_{i,j \sim [n]} L_{ij} \left( \pE_{\mu} \mathbb{1}[\overline{E}_i]^2 \right)^{1/2} \left(\pE_{\mu} \norm{x_i - \pE_\mu x_i}^2 + \norm{x_j - \pE_\mu x_j}^2\right)^{1/2}\,.
    \end{align*}
    Note that we have that $\pE_{\mu} \mathbb{1}[\overline{E}_i]^2 = \Pr_{\mu} (\overline{E}_i) \leq 1/C^2$ for all $i \in [n]$ via Chebyshev, so 
    \begin{align*}
        &\E_{\{\mu'\}} \E_{i,j \sim [n]} \pE_{\mu'} L_{ij} \left[ \mathbb{1}[\overline{E}_i]  \cdot \left(\norm{x_i - \pE_\mu x_i} + \norm{x_j - \pE_\mu x_j} \right) \right] \\
        &\leq O(1) \cdot \E_{i,j \sim [n]} L_{ij}/C \left(\pE_{\mu} \norm{x_i - \pE_\mu x_i}^2 + \norm{x_j - \pE_\mu x_j}^2\right)^{1/2}\,.
    \end{align*}
    Applying Cauchy-Schwartz again, we have that
    \begin{align*}
        &\E_{\{\mu'\}} \E_{i,j \sim [n]} \pE_{\mu'} L_{ij} \left[ \mathbb{1}[\overline{E}_i]  \cdot \left(\norm{x_i - \pE_\mu x_i} + \norm{x_j - \pE_\mu x_j} \right) \right] \\
        &\leq O(1/C) \cdot \left(\E_{i,j \sim [n]} L_{ij}^2\right)^{1/2} \left( \E_{i \sim [n]} \pE_\mu \norm{x_i - \pE_\mu x_i}^2 \right)^{1/2} \,.
    \end{align*} 
\end{proof}

We now return to the proof of~\cref{lem:lipschitz-rounding-error}.
\begin{proof}[Proof of \cref{lem:lipschitz-rounding-error}]
    Without loss of generality, let $f_{ij}(\pE_\mu x_i, \pE_\mu x_j) = 0$; otherwise, you can apply an additive shift to each function without changing the quantity we are bounding. Furthermore, note that we can decompose the expression we want to bound as follows: 
    \begin{align*}
        &\left\vert \E_{\{\mu'\}} \E_{i,j \sim [n]} \left[\pE_{(\mu')^{\otimes}} f_{ij}(x_i,x_j) - \pE_{\mu'} f_{ij}(x_i,x_j) \right] \right\vert \\
        &= \Biggl\vert \E_{\{\mu'\}} \E_{i,j \sim [n]} \biggl[\left[\pE_{(\mu')^{\otimes}} f_{ij} (\widetilde{x}_i, \widetilde{x}_j) - \pE_{\mu'} f_{ij} (\widetilde{x}_i, \widetilde{x}_j)\right] \\
        &\qquad+ \pE_{\mu'} \left[ f_{ij}(\widetilde{x}_i,\widetilde{x}_j) - f_{ij} (x_i, x_j) \right] + \pE_{(\mu'){^\otimes}}\left[ f_{ij}(x_i, x_j) - f_{ij}(\widetilde{x}_i,\widetilde{x}_j)\right] \biggr]\Biggr\vert \\
        &\leq  \left\vert \E_{\{\mu'\}} \E_{i,j \sim [n]} \left[\pE_{(\mu')^{\otimes}} f_{ij} (\widetilde{x}_i, \widetilde{x}_j) - \pE_{\mu'} f_{ij} (\widetilde{x}_i, \widetilde{x}_j)\right] \right\vert + \left\vert \E_{\{\mu'\}} \E_{i,j \sim [n]} \pE_{\mu} \left[ \left( f_{ij}(\widetilde{x}_i,\widetilde{x}_j) - f_{ij} (x_i, x_j) \right) \right] \right\vert \\
        &\qquad+ \left\vert \E_{\{\mu'\}} \E_{i,j \sim [n]} \pE_{(\mu')^{\otimes}} \left[\left( f_{ij}(y_i, y_j) - f_{ij}(\widetilde{y}_i,\widetilde{y}_j)\right)\right] \right\vert\,. 
    \end{align*}
    Using that $\widetilde{x}_i = (x_i(T) -\expecf{0}{x_i}) E_i + \expecf{0}{x_i}$ and thus is within a $C \cdot \sqrt{\pE_\mu \norm{x_i - \pE_\mu x_i}^2}$ radius ball with probability $1$ we have always that
    \begin{equation}
    \label{eqn:shift-and-lipschitz}
    \begin{split}
        f_{i,j}(\widetilde{x}_i,\widetilde{x}_j) &\leq f_{i,j} \Paren{ \pE_\mu {x_i}, \pE_{\mu} x_j} + L_{ij} C \Paren{ \sqrt{\pE_\mu \norm{x_i - \pE_\mu x_i}^2} + \sqrt{\pE_\mu \norm{x_j - \E x_j}^2} } \\
        & \leq L_{ij} C \Paren{ \sqrt{\pE_\mu \norm{x_i - \pE_\mu x_i}^2} + \sqrt{\pE_\mu \norm{x_j - \E x_j}^2} }\,,
    \end{split}
    \end{equation}
    using that $f_{ij}$ is $L_{ij}$-Lipschitz. We can thus bound the first term as follows:
    \begin{align*}
        &\left\vert \E_{\{\mu'\}} \E_{i,j \sim [n]} \left[\pE_{(\mu')^{\otimes}} f_{ij} (\widetilde{x}_i, \widetilde{x}_j) - \pE_{\mu'} f_{ij} (\widetilde{x}_i, \widetilde{x}_j)\right] \right\vert \\
        &\leq \E_{\{\mu'\}} \E_{i,j \sim [n]} \left[ \tv \Paren{ \mu'_{\{\widetilde{x_i}, \widetilde{x_j}\}}, \mu'_{\{\widetilde{x_i}\} \otimes \{ \widetilde{x_j} \}}} \cdot 2L_{ij} \cdot C \cdot \left(\sqrt{\pE_\mu \norm{x_i - \pE_\mu x_i}^2} + \sqrt{\pE_\mu \norm{x_j - \E x_j}^2} \right) \right] \\
        &\leq 4C \E_{\{\mu'\}} \E_{i \sim [n]} \left[ \sqrt{\pE_\mu \norm{x_i - \pE_\mu x_i}^2} \cdot \E_{j \sim [n]} \left[ \tv \Paren{ \mu'_{\{\widetilde{x_i}, \widetilde{x_j}\}}, \mu'_{\{\widetilde{x_i}\} \otimes \{ \widetilde{x_j} \}}} \cdot L_{ij} \right] \right]\,,
    \end{align*}
    where the first inequality follows from \cref{eqn:shift-and-lipschitz}. Applying Cauchy-Schwarz twice, we have that
    \begin{align*}
        &\left\vert \E_{\{\mu'\}} \E_{i,j \sim [n]} \left[\pE_{(\mu')^{\otimes}} f_{ij} (\widetilde{x}_i, \widetilde{x}_j) - \pE_{\mu'} f_{ij} (\widetilde{x}_i, \widetilde{x}_j)\right] \right\vert \\
        &\leq 4C \left(\E_{\{\mu'\}} \E_{i \sim [n]} \pE_\mu \norm{x_i - \pE_\mu x_i}^2\right)^{1/2} \left( \E_{i \sim [n]} \left( \E_{j \sim [n]} \tv \Paren{ \mu'_{\{\widetilde{x_i}, \widetilde{x_j}\}}, \mu'_{\{\widetilde{x_i}\} \otimes \{ \widetilde{x_j} \}}} \cdot L_{ij} \right)^2\right)^{1/2} \\
        &\leq 4C \left(\E_{i \sim [n]} \pE_\mu \norm{x_i - \pE_\mu x_i}^2\right)^{1/2} \left( \E_{\{\mu'\}} \E_{i \sim [n]} \left(\E_{j \sim [n]} L_{ij}^2\right)  \left(\E_{j \sim [n]} \tv \Paren{ \mu'_{\{\widetilde{x_i}, \widetilde{x_j}\}}, \mu'_{\{\widetilde{x_i}\} \otimes \{ \widetilde{x_j} \}}}^2 \right)\right)^{1/2}\,.
    \end{align*}
    By our lemma assumption, we have that this can be further bounded by
    \begin{equation*}
        \begin{split}
            & \left\vert  \E_{\{\mu'\}}  \E_{i,j \sim [n]} \left[\pE_{(\mu')^{\otimes}} f_{ij} (\widetilde{x}_i, \widetilde{x}_j) - \pE_{\mu'} f_{ij} (\widetilde{x}_i, \widetilde{x}_j)\right] \right\vert \\
            & \quad \leq \bigO{ C \sqrt{\epsilon}} \cdot \left(\E_{i \sim [n]} \pE_\mu \norm{x_i - \pE_\mu x_i}^2\right)^{1/2} \cdot \left( \E_{i,j\sim[n]} L_{ij}^2\right)^{1/2}\,.
        \end{split}
    \end{equation*}
    It remains to bound the remaining two terms in the sum. Applying~\cref{lem:tilde-vs-not-diff} we have that they are both at most $O(1/C) \cdot \left(\E_{i,j \sim [n]} L_{ij}^2\right)^{1/2} \left( \E_{i \sim [n]} \pE_\mu \norm{x_i - \pE_\mu x_i}^2 \right)^{1/2}$, which completes the proof.
\end{proof}

Given~\cref{lem:lipschitz-rounding-error} we can now easily prove~\cref{lem:lipschitz-potential-aligned}.
\begin{proof}[Proof of~\cref{lem:lipschitz-potential-aligned}]
    Recall that we let $E_i$ be the event that $\norm{x_i - \pE_{\mu} x_i} \leq C \cdot \sqrt{\pE_{\mu} \norm{x_i - \pE_{\mu} x_i}^2}$ and let 
    \[\widetilde{x_i} = x_i E_i + (\pE_{\mu}[x_i] ) \overline{E}_i \,.\]
    The goal is to show that $\E_{i,j \in [n]} f_{ij}(x_i, x_j)$ is potential aligned.
    
    Consider a distribution over pseudo-distributions $\{\mu'\}$ such that $\E_{\{\mu'\}} \{\mu'\} = \mu$ for some initial pseudo-distribution $\mu$. When we have that
    \begin{align*}
        &\left\vert \E_{\{\mu'\}} \E_{i,j \sim [n]} \left[ \pE_{\mu'} f_{ij}(x_i, x_j) - \E_{(\mu')^{\otimes}}f_{ij}(x_i, x_j) \right] \right\vert \\
        &\qquad> O(\epsilon^{1/4}) \cdot \left(\E_{i \sim [n]} \pE_\mu \norm{x_i - \pE_\mu x_i}^2\right)^{1/2} \cdot \left( \E_{i,j\sim[n]} L_{ij}^2\right)^{1/2} \,,
    \end{align*}
    then by~\cref{lem:lipschitz-rounding-error} (with setting $C = \epsilon^{-1/4}$) we must have that 
    \[ \E_{\{\mu'\}} \E_{i \sim [n]} \left[\left( \E_{j \sim [n]} L_{ij}^2 \right) \left( \E_{j \sim [n]} \tv \Paren{ \mu'_{\{\widetilde{x_i}, \widetilde{x_j}\}}, \mu'_{\{\widetilde{x_i}\} \otimes \{ \widetilde{x_j} \}}}^2 \right)\right] > \eps \cdot \E_{i,j \sim [n]} L_{ij}^2\,.\]
    Let $\psi$ be the distribution over $[n]$ which samples $i$ proportional to $\E_{j \sim [n]} L_{ij}^2$. Then equivalently, we can write the above as follows:
    \[ \E_{\{\mu'\}} \E_{i \sim \psi} \left[\E_{j \sim [n]} \tv \Paren{ \mu'_{\{\widetilde{x_i}, \widetilde{x_j}\}}, \mu'_{\{\widetilde{x_i}\} \otimes \{ \widetilde{x_j} \}}}^2 \right] > \eps \,.\]
    By Pinskers, we have that $\tv \Paren{ \mu'_{\{\widetilde{x_i}, \widetilde{x_j}\}}, \mu'_{\{\widetilde{x_i}\} \otimes \{ \widetilde{x_j} \}}}^2 \leq I(\mu'_{\{\widetilde{x_i}, \widetilde{x_j}\}}; \mu'_{\{\widetilde{x_i}\} \otimes \{ \widetilde{x_j} \}})$ and thus
    \[ \E_{\{\mu'\}} \E_{i \sim \psi}  \E_{j \sim [n]} H(\widetilde{x}_j) - H(\widetilde{x}_j \vert \widetilde{x}_i) = \E_{\{\mu'\}} \E_{i \sim \psi} \left[\E_{j \sim [n]} I(\mu'_{\{\widetilde{x_i}, \widetilde{x_j}\}}; \mu'_{\{\widetilde{x_i}\} \otimes \{ \widetilde{x_j} \}}) \right] > \eps \,.\]
    We have therefore shown that whenever the rounding error is large, the potential function decreases in expectation for a random $i$. Notably, there must exist a fixed $i$ which achieves this decrease in expectation over the sample from $x_i$ and thus we have shown that Lipschitz functions are potential aligned with the given potential. 
\end{proof}

\subsection{Euclidean Metric Violation: Algorithm and Analysis}
In this section we will show that there exists an additive approximation scheme for \textsf{$k$-EMV}. 

\begin{theorem}[Additive Approximation Scheme for \textsf{$k$-EMV}]
\label{thm:raw-stress-main}
Given a set of distances $\calD = \Set{d_{i,j}}_{i,j \in [n]}$, $k \in \mathbb{N}$, where each $d_{ij}$ can be represented in $B$ bits, and $0<\epsilon <1$, there exists an algorithm that runs in $B^{\mathcal{O}( k^2\log(k/\eps)/\eps^4)} \cdot \poly(n)$ time and outputs points $\Set{ \hat{x}_1, \ldots, \hat{x}_n }_{i \in [n]} \in \R^k$ such that with probability at least $0.99$, 
\begin{equation*}
    \E_{i,j\sim[n]} \Paren{ d_{i,j} - \norm{ \hat{x}_i - \hat{x}_j }_2 }^2 \leq \OPT_{\textrm{EMV}} + \eps  \cdot \E_{i,j\sim[n]} d_{i,j}^2. 
\end{equation*}
\end{theorem}

In this section we will give an algorithm and analysis for an additive approximation scheme that runs in $ (nB)^{\mathcal{O}( k^2\log(k/\eps)/\eps^4)}$ time. In~\cref{sec:emv-speedup} we will show how to improve the runtime of the algorithm to $B^{\mathcal{O}( k^2\log(k/\eps)/\eps^4)} \cdot \poly(n)$. Without loss of generality, we will assume that $\min_{i,j} d_{ij} = 1$ and $\max_{i,j} d_{ij} = \Delta$ for the remainder of this section.

\begin{mdframed}
  \begin{algorithm}[Additive Approximation Scheme for \textsf{$k$-EMV}]
    \label{algo:raw-stress}\mbox{}
    \begin{description}
    \item[Input:] Non-negative numbers $\mathcal{D} = \{d_{ij}\}_{i,j\in[n]}$, target dimension $k \in \mathbb{N}$, target accuracy $0<\epsilon<1$.
    
    \item[Operations:]\mbox{}
    \begin{enumerate}
        \item For each $a \in [n]$ and set $\calT$ of size at most $\bigO{ k^2 \log (k/\epsilon) /\epsilon^4}$:
        \begin{enumerate}
            \item Let %
            \[ \Sigma = \{0\} \cup \left\{\eps \cdot \sqrt{\frac{1}{k}} (1+\epsilon)^t \, \vert \, t \in [12\log_{1+\epsilon} (k\Delta/\epsilon)] \right\}\,.\]
            Let $\mu$ be a $O(k^2 \log (k/\epsilon) /\epsilon^4)$-degree Sherali-Adams pseudo-distribution over $\left(\Sigma^k\right)^n$ such that $\pE_{\mu}$ optimizes 
            \begin{equation}
            \label{eqn:main-opt-rs}
            \begin{split}
                & \min_{\pE_{\mu} }   \hspace{0.1in}\pE_{\mu}  \E_{i,j\sim[n]} (d_{ij} - \norm{x_i - x_j})^2\\
                & \textrm{s.t.}  \hspace{0.1in} x_a = 0 , \hspace{0.1in}\expecf{i \sim [n]}{ \norm{x_i}^2 } \leq 6 \E_{i,j\sim[n]} d_{ij}^2
            \end{split}
            \end{equation}
            \item Let $\hat{x}_\calT$ be a draw from the local distribution $\{x_{\calT}\}$. Let $\mu_\calT$ be the pseudodistribution obtained by conditioning on $\{x_i = \hat{x}_i\}_{i \in \calT}$.
            \item Let $\hat{X}_a$ be the embedding $\{\hat{x}_i\}_{i \in [n]}$ where $\hat{x}_i$ is sampled independently from the the $k$-degree local distribution of $\{x_i\}$ in $\mu_\calT$.
        \end{enumerate}
    \end{enumerate}
    \item[Output:] The embedding $\hat{X}_a$ with the lowest \textsf{$k$-EMV}  objective value.
    \end{description}
  \end{algorithm}
\end{mdframed}

\paragraph{Discretization preserves objective value.}

In this section, we will show that a discretized variant of the Euclidean Metric Violation objective with alphabet size $\bigO{\epsilon \cdot \log (\Delta k /\epsilon)}$ preserves the objective value up to an additive error of $\epsilon \E_{i,j\sim[n]} d_{ij}^2$

\begin{lemma}[Discretization incurs small error]
\label{lem:discretization-rs}
    Given a set of distances $\{ d_{i,j} \}_{i,j \in [n]}$ such that $\Delta = \max_{i,j} d_{i,j}$ and $1 = \min_{i,j}d_{i,j}$, $k \in \mathbb{N}$,  and $0<\epsilon <1$, there exists a fixed set $\Sigma \subseteq \mathbb{R}^k$ depending only on $\Delta$ and $\epsilon$ such that the objective
    \[ \OPT_{\textrm{DEMV}} =  \min_{ x_1, \ldots , x_n \in \Sigma } \sum_{i,j \in [n]} \Paren{  d_{i,j} - \norm{x_i - x_j }_2 }^2\]
    satisfies
    \[ \OPT_{\textrm{DEMV}} \leq \OPT_{\textrm{EMV}} + 16 \epsilon \cdot \E_{i,j\sim[n]} d_{ij}^2\,.\]
    Furthermore,
    \[ \vert \Sigma \vert \leq \bigO{\epsilon \cdot \log (\Delta k /\epsilon)} \,.\]
\end{lemma}

We first describe the set $\Sigma$. We have that:
\[ \Sigma = \{0\} \cup \left\{\eps \cdot \sqrt{\frac{1}{k}} (1+\epsilon)^t \, \vert \, t \in [12\log_{1+\epsilon} (k\Delta/\epsilon)] \right\}\,.\]

In order to show that $\OPT_{\mathrm{DEMV}}$ is small, we will bound the value of the objective when taking an (approximately) optimal low diameter solution to the original \textsf{$k$-EMV} problem and assigning each point $x_i$ to its closest point in $\Sigma$. 

First, we must show that such an approximately optimal low-diameter solution exists.
\begin{fact}
\label{fact:rs-diam-opt}
    Let $\{ d_{i,j} \}_{i,j \in [n]}$ be a set of distances such that $\Delta = \max_{i,j} d_{i,j}$ and $1 =\min_{i,j}d_{i,j}$. Then for any $k \in \mathbb{N}$,  and $0<\epsilon <1$, there exists a solution $x_1,\ldots, x_n$ for the \textsf{$k$-EMV}  objective of value $\OPT_{\mathrm{EMV}} + \epsilon \E_{i,j \sim [n]} d_{ij}^2$ where the diameter of the embedding is at most $\left(\Delta/\eps\right)^{O(1)}$.
\end{fact}
\begin{proof}
    Take some optimal solution $x_1^*, \ldots, x_n^*$. We will modify the solution to get a new solution satisfying the diameter conditions with value at most $\OPT_{\mathrm{EMV}} + \epsilon \E_{i,j \sim [n]} d_{ij}^2$.

    Note that we have that $\E_{i, j \sim [n]} (d_{ij} - \norm{x_i^* - x_j^*})^2 \leq \OPT_{\mathrm{EMV}} \leq \Delta^2$. Thus, by Markov's Inequality we have that for all but an $\epsilon^2 / \Delta^4$-fraction of pairs we have that $\norm{x_i^* - x_j^*} \leq \epsilon^{-1} \Delta^3 + \Delta$. Let the set of good points $x_i^*$ be all points that participate in at most $\epsilon/\Delta^2$ pairs such that $\norm{x_i^* - x_j^*}^2 > \epsilon^{-1} \Delta^3 + \Delta$. Note that there are at least $1-\eps/\Delta^2$ many good points. 
    
    Furthermore, every pair of good points $x_i^*, x_j^*$ are not too far apart, specifically $\norm{x_i^* - x_j^*} \leq 2 \left(\epsilon^{-1} \Delta^3 + \Delta\right)$. Note that $i,j$ both have at most $\epsilon/\Delta^2$ neighbors such that $\norm{x_i^* - x_j^*}^2 > \epsilon^{-1} \Delta^3 + \Delta$. Thus, they must have a shared neighbor $\ell$ such that $\norm{x_i^* - x_\ell^*}^2 \leq \epsilon^{-1} \Delta^3 + \Delta$ and $\norm{x_j^* - x_\ell^*}^2 \leq \epsilon^{-1} \Delta^3 + \Delta$. Therefore, by Triangle Inequality $\norm{x_i^* - x_j^*} \leq 2 \left(\epsilon^{-1} \Delta^3 + \Delta\right)$.

    We now note that we can construct a low diameter solution by first taking the solution induced by the good points above and projecting the remaining points to the $2 \left(\epsilon^{-1} \Delta^3 + \Delta\right)$ radius ball containing the good points. The cost on all pairs that only contain good points is preserved. Furthermore, there are only $\eps/\Delta^2 n^2$ many pairs that contain at least one bad point. Note that the procedure above only contracts distances, so the additional cost incurred on pairs after a contraction is at most $\Delta^2$. Thus, the cost of the solution increased by at most $\eps \leq \eps \E_{i, j\sim [n]} d_{ij}^2$, as desired.
\end{proof}

Furthermore, note that we must have that any optimal or approximately optimal solution must be on average contained within a not too large ball. 

\begin{lemma}[Points are contained in a ball]
\label{lem:magnitude-points-rs}
    Let $\{ d_{i,j} \}_{i,j \in [n]}$ be a set of distances such that $\Delta = \max_{i,j} d_{i,j}$ and $1 =\min_{i,j}d_{i,j}$. Then for any $k \in \mathbb{N}$,  and $0<\epsilon <1$, and any embedding $x_1,\ldots, x_n$ with Euclidean Metric Violation objective value at most $2\E_{i,j \sim [n]} d_{ij}^2$, there exists a fixed point $c$ such that
    \[ \E_{i\sim[n]} \norm{x_i - c}^2 \leq 6\E_{i,j\sim[n]} d_{ij}^2\,.\]
    Furthermore, the point $c = x_a$ for some $a \in [n]$.
\end{lemma}
\begin{proof}
    We will show that in expectation, if we pick $c$ as a random $x_a$, the inequality holds, and thus there exists a fixed point $c = x_a$ for some $a$ such that it is true. Specifically, if we sample a uniformly random index $a \in [n]$, then we have by Almost Triangle Inequality that
    \[ \E_{a\sim[n]} \E_{i\sim[n]} \norm{x_i - x_a}^2 = \E_{ai} \left(\norm{x_i-x_a} \pm d_{ia}\right)^2 \leq \E_{ai} 2\left(d_{ia} - \norm{x_i-x_a}\right)^2 + 2d_{ia}^2\,.\]
    The first term is the Euclidean Metric Violation objective value, which is bounded by $2\E_{i,j\sim[n]} d_{ij}^2$. Thus, we have that
    \[\E_{a\sim[n]} \E_{i\sim[n]} \norm{x_i - x_a}^2 \leq 6 \E_{i,j \sim[n]} d_{ij}^2\,.\]
\end{proof}

We can now proceed to the proof of~\cref{lem:discretization-rs}.
\begin{proof}[Proof of~\cref{lem:discretization-rs}]
    Consider an approximately optimal solution to the \textsf{$k$-EMV}  problem with low diameter, which must exist by~\cref{fact:rs-diam-opt}. Note that this solution has low objective value and thus without loss of generality shift it such that the fixed point $c$ satisfying the conditions of~\cref{lem:magnitude-points-rs} is $0$. We then generate a discretized solution by taking each $x_i$ and assigning $i$ to the point $\hat{x}_i$ which is closest in $\Sigma$. We have that
    \begin{equation}
    \label{eqn:main-discretization-triangle-inequality}
    \begin{split}
        \OPT_{\mathrm{DEMV}} &\leq \E_{i,j\sim[n]} (d_{ij} - \norm{\hat{x}_i - \hat{x}_j})^2 \\
        &= \E_{i,j\sim[n]} (d_{ij} - \norm{\hat{x}_i -\hat{x}_j)} \pm \norm{x_i - x_j})^2 \\
        &= \OPT_{\mathrm{EMV}} + \eps \E_{i,j \sim [n]} d_{ij}^2 + 2 \E_{i,j\sim[n]} d_{ij} \\
        &\qquad +\left(\norm{\hat{x}_i - \hat{x}_j} - \norm{x_i - x_j}\right) + \E_{i,j\sim[n]} \left(\norm{\hat{x}_i - \hat{x}_j} - \norm{x_i - x_j}\right)^2 \,.
    \end{split}
    \end{equation}
    Note that by Cauchy-Schwarz, we have that
    \begin{equation}
    \label{eqn:cauchy-schwarz-on-middle-term}
        2 \E_{i,j\sim[n]} d_{ij} \left(\norm{\hat{x}_i - \hat{x}_j} - \norm{x_i - x_j}\right) \leq 2 \sqrt{\E_{i,j\sim[n]} d_{ij}^2} \cdot \sqrt{\E_{i,j\sim[n]} \left(\norm{\hat{x}_i - \hat{x}_j} - \norm{x_i - x_j}\right)^2}  \,,
    \end{equation}
    and thus it suffices to bound $\E_{i,j\sim[n]} \left(\norm{\hat{x}_i - \hat{x}_j} - \norm{x_i - x_j}\right)^2$. Using reverse triangle inequality, we have that
    \begin{equation*}
    \begin{split}
        \E_{i,j\sim[n]} \left(\norm{\hat{x}_i - \hat{x}_j} - \norm{x_i - x_j}\right)^2  & \leq \E_{i,j\sim[n]} \Paren{ \norm{ \hat{x}_i - \hat{x}_j - x_i + x_j } }^2 \\
        & \leq 4 \E_{i\sim[n]} \norm{\hat{x}_i - x_i}^2 \, .
    \end{split}
    \end{equation*}
    
    For any point $x_i$, the amount that each coordinate $\ell$ shifts by is either at most $\epsilon \sqrt{\frac{1}{k} \cdot \E_{i,j \sim [n]} d_{ij}^2}$ (if the value is in the range $\left[-\epsilon \sqrt{\frac{1}{k} \cdot \E_{i,j \sim [n]} d_{ij}^2}, \epsilon \sqrt{\frac{1}{k} \cdot \E_{i,j \sim [n]} d_{ij}^2}\right]$) or $ \epsilon (x_i)_\ell$ (if the value is larger in magnitude than $\epsilon \sqrt{\frac{1}{k} \cdot \E_{i,j \sim [n]} d_{ij}^2}$) since there is always a point in the second case which is at most an $\epsilon$ factor larger in any direction (since the solution has diameter at most $n\Delta$ along any basis vector by~\cref{fact:rs-diam-opt}). Thus, we have that
    \[ \E_{i\sim[n]} \norm{\hat{x}_i - x_i}^2 \leq \E_{i\sim[n]} \left(\epsilon^2 \norm{x_i}^2 + \epsilon^2 \E_{i,j \sim [n]} d_{ij}^2\right)\,.\]
    By~\cref{lem:magnitude-points-rs} we have that the first term is at most $4\epsilon^2 \E_{i,j\sim[n]} d_{ij}^2$ and therefore, $$\E_{i,j\sim[n]} \left(\norm{\hat{x}_i - \hat{x}_j} - \norm{x_i - x_j}\right)^2 \leq 5 \eps^2 \E_{i,j\sim[n]} d_{ij}^2\,.$$
    Substituting this back into \cref{eqn:cauchy-schwarz-on-middle-term} and \cref{eqn:main-discretization-triangle-inequality}, we can conclude that
    \begin{equation*}
        \OPT_{\mathrm{DEMV}} \leq \OPT_{\mathrm{EMV}} + 16 \eps \E_{i,j\sim[n]} d_{ij}^2 \,.  
    \end{equation*}
    Finally, we note that $\vert \Sigma \vert \leq 1 + 2 \log_{1+\eps} (\Delta^{O(1)} \sqrt{k}/\eps^2)\leq \bigO{\epsilon \cdot \log (\Delta k /\epsilon)}$, which completes the proof.
\end{proof}

\paragraph{Constrained LP has Small Value.}

Using~\cref{lem:magnitude-points-rs} it is now easy to prove that there is some $a$ such that the LP in~\cref{algo:raw-stress} has small objective value.

\begin{lemma}
\label{lem:rs-pe-low-value}
    Let $\calD = \{ d_{i,j} \}_{i,j \in [n]}$ be a set of a distance such that $\Delta = \max_{i,j} d_{i,j}$ and $1 = \min_{i,j}d_{i,j}$, $k \in \mathbb{N}$, $k$ be a constant, and $0<\epsilon <1$. Then there exists $a \in [n]$ such that there is an embedding $\hat{x}_1, \ldots, \hat{x}_n$ satisfying
    \begin{enumerate}
        \item $\hat{x}_a = 0$,
        \item $\E_{i\sim[n]} \norm{x_i}^2 \leq 6\E_{i,j\sim[n]} d_{ij}^2$,
        \item $\E_{i,j\sim[n]} (d_{ij} - \norm{\hat{x}_i - \hat{x}_j})^2 \leq OPT_{\mathrm{EMV}} + \epsilon \E_{i,j\sim[n]} d_{ij}^2$.
    \end{enumerate}
\end{lemma}
\begin{proof}
    The proof follows almost directly from~\cref{lem:magnitude-points-rs}, which shows that in the continuous variant of the problem that there exists $a \in [n]$ satisfying the first two conditions (after shifting the embedding such that $x_a = 0$). Furthermore, we note that the process described in the proof of~\cref{lem:discretization-rs}, which transforms an optimal solution to the continuous Euclidean Metric violation problem into a solution over our discrete alphabet, maps points at $0$ to $0$. Thus, by applying this transformation, we have that the third condition is satisfied. 
\end{proof}
Note that this lemma implies that one of our Sherali-Adams LPs will both be feasible and have objective value comparable (up to additive loss) to $\OPT_{\mathrm{EMV}}$.

\paragraph{Analyzing global correlation rounding with potential functions.}

We now show that after $\bigO{k/\epsilon^2}$ many further conditionings, the additive error of independent rounding is small. Note that we can rewrite our objective value as 
\[ \E_{i,j\sim[n]} (d_{ij} - \norm{x_i - x_j})^2  = \E_{i,j\sim[n]} \left[d_{ij}^2 - 2 d_{ij} \norm{x_i - x_j} + \norm{x_i}^2 + \norm{x_j}^2 - 2 \langle x_i, x_j\rangle \right]\,.\]
The contribution to the objective value of the terms $d_{ij}^2, \norm{x_i}^2, \norm{x_j}^2$ is preserved when independently rounding $x_i$, since they depend on the distribution of at most one embedding point. Thus, it remains to argue that after conditioning independent rounding approximately preserves the objective value of $2 d_{ij} \norm{x_i - x_j}$ and $-2\langle x_i, x_j\rangle$.

We have that all the remaining parts of our objective function are potential aligned for appropriately chosen potential functions and thus we will apply~\cref{lem:gcr_linear_combo}. Note that $\Set{ d_{ij} \norm{x_i - x_j}}_{i,j \in [m]}$ is a family of Lipschitz functions with Lipschitz constants $\Set{d_{ij}}_{i,j \in [n]}$ and thus we know by~\cref{lem:lipschitz-potential-aligned} that this portion of the objective function is potential aligned. Furthermore, we have by~\cref{fact:deg-2-potential-aligned} that under bounded variance conditions (which our psuedodistribution satisfies) the remainder of the objective function is also potential aligned. The main remaining piece is to show upper bounds on the initial potential.

\begin{lemma}[Average Entropy Bound]
\label{lem:rs-entropy-bound}
Let $\mu$ be a pseudo-distribution over $\Sigma^k$ which satisfies 
    \[ \E_i \pE_{\mu} \norm{x_i}^2 \leq \E_{i,j \sim [n]} d_{ij}^2\,.\]
    Let $\widetilde{x_i}$ is $x_i(T) E_i + (\pE_\mu x_i) \overline{E}_i$ where the event $E_i$ is $E_i = (\norm{x_i - \E_\mu x_i} \leq C \cdot \sqrt{\pE_\mu \norm{x_i - \pE_\mu x_i}^2})$. Then we have that
    \[ \E_i H(\widetilde{x}_i) \leq O\left(k \log (C k/\epsilon)\right)\,.\]
\end{lemma}

\begin{proof}
    Note that $H(\widetilde{x}_i)$ is at most $\log \gamma_i$, where $\gamma_i$ is the support size of $\widetilde{X}_i$. By definition, we have that $\widetilde{x}_i$ has support contained in a $k$-dimensional cube with side length $2C \sqrt{\E_\mu \norm{x_i - \E_\mu x_i}^2}$ and all points in $\Sigma^k$ differ by at least $\sqrt{\epsilon/k \cdot \E_{i,j \sim [n]} d_{ij}^2}$ in each index, and thus
    \[ H(\widetilde{x}_i) \leq \log \gamma_i \leq \log \left( \left( \frac{2C \sqrt{\E_\mu \norm{x_i - \E_\mu x_i}^2}}{\sqrt{\epsilon/k \cdot \E_{i,j \sim [n]} d_{ij}^2}} \right)^k \right) \leq k \log \left( \frac{2C \sqrt{E \norm{x_i - \E x_i}^2}}{\sqrt{\epsilon/k \cdot \E_{i,j \sim [n]} d_{ij}^2}} \right)\,. \]
    By Jensens, we have that
    \begin{align*}
        \E_i H(\widetilde{x}_i) &\leq \E_i k \log \left( \frac{2C \sqrt{E_\mu \norm{x_i - \E_\mu x_i}^2}}{\sqrt{\epsilon/k \cdot \E_{i,j \sim [n]} d_{ij}^2}} \right) \leq k \log \left( \frac{2C \sqrt{\E_{i\sim[n]} \E_\mu \norm{x_i - \E_\mu x_i}^2}}{\sqrt{\epsilon/k \cdot \E_{i,j \sim [n]} d_{ij}^2}} \right)\,.
    \end{align*}
    Finally, by our bound on the average squared norm of $x_i$, we conclude that
    \[ \E_i H(\widetilde{x}_i) \leq k \log \left( \frac{2C \sqrt{\E_{i\sim[n]} \E_\mu \norm{x_i - \E_\mu x_i}^2}}{\sqrt{\epsilon/k \cdot \E_{i,j \sim [n]} d_{ij}^2}} \right) \leq O\left(k \log (C k/\epsilon)\right)\,.\]
\end{proof}

Now we can prove the following lemma, which is an easy application of~\cref{lem:gcr_linear_combo}.

\begin{lemma}[Rounding \textsf{$k$-EMV} ]
\label{lem:rs-gcr-error}
    Let $\epsilon \geq 0$ and let $\mu$ be an initial pseudo-distribution which satisfies
    \[ \E_{i \sim [n]} \pE_\mu \norm{x_i}^2 \leq 6 \E_{i,j \sim [n]} d_{ij}^2\,.\]
    Then there exists a set $\calT$ of size at most $O\left(k \log (k/\epsilon) /\epsilon^4 \right)$ such that after conditioning on the values of $x_i \in \calT$,
    \[ \left\vert \E_{\{\mu_\calT\}} \pE_{\mu_{\calT}} \E_{i,j \sim [n]} (d_{ij} - \norm{x_i - x_j})^2 - \pE_{\mu^{\otimes}_{\calT}} \E_{i,j \sim [n]} (d_{ij} - \norm{x_i - x_j})^2 \right\vert \leq \epsilon \E_{i,j\sim[n]} d_{ij}^2 \,.\]
\end{lemma}
\begin{proof}
    Note that we can write
    \[ \E_{i,j\sim[n]} (d_{ij} - \norm{x_i - x_j})^2  = \E_{i,j\sim[n]} \left[d_{ij}^2 - 2 d_{ij} \norm{x_i - x_j} + \norm{x_i}^2 + \norm{x_j}^2 - 2 \langle x_i, x_j\rangle \right]\,.\]
    Note that out of these terms the only two which are different in expectation under $\mu_\calT$ versus $\mu_\calT^\otimes$ are $2 d_{ij} \norm{x_i - x_j}$ and $2 \langle x_i, x_j\rangle$ and they are both potential aligned. 
    
    In particular, since the first is an average of Lipschitz functions with Lipschitz constants $d_{ij}$, we have that $2 d_{ij} \norm{x_i - x_j}$ is $(\epsilon \cdot \left( \E_{i \sim [n]} \pE_{\mu} \norm{x_i - \pE_\mu x_i}^2 \right)^{1/2} \cdot \left( \E_{i, j\sim [n]} d_{ij}^2\right)^{1/2}, \epsilon^4)$ potential aligned function for $\mu$ and $\E_{i \sim [n]} H(x_i)$ by~\cref{lem:lipschitz-potential-aligned}. Since we have that $\E_{i \sim [n]} \pE_\mu \norm{x_i}^2 \leq 4 \E_{i,j \sim [n]} d_{ij}^2$ we can further say that it is also $(4\epsilon \cdot \left( \E_{i, j\sim [n]} d_{ij}^2\right), \epsilon^4)$ potential aligned, since 
    \[\E_{i \sim [n]} \pE_{\mu} \norm{x_i - \pE_\mu x_i}^2 \leq \E_{i \sim [n]} \pE_\mu \norm{x_i}^2\,.\]
    Furthermore, we have that $\E_{i,j \sim [n]} 2 \langle x_i, x_j\rangle$ is also $(\epsilon \E_{i,j \sim [n]} d_{ij}^2, \epsilon^2 / (k \cdot \E_{i,j \sim [n]} d_{ij}^2) )$ potential aligned using the fact that $\E_{i \sim [n]} \tr(\widetilde{\Sigma}_i) \leq \E_{i \sim [n]} \pE_\mu \norm{x_i}^2 \leq \E_{i,j \sim [n]} d_{ij}^2$ via our pseudo-distribution constraint in combination with~\cref{fact:deg-2-potential-aligned}.
    Thus, applying~\cref{lem:gcr_linear_combo} we have that 
    \[ \left\vert \E_{\{\mu_\calT\}} \pE_{\mu_{\calT}} \E_{i,j \sim [n]} (d_{ij} - \norm{x_i - x_j})^2 - \pE_{\mu^{\otimes}_{\calT}} \E_{i,j \sim [n]} (d_{ij} - \norm{x_i - x_j})^2 \right\vert \leq \epsilon \E_{i,j\sim[n]} d_{ij}^2 \,,\]
    holds for a set $\calT$ of size 
    \[ \vert \calT \vert \leq O\left(\frac{\E_{i} H(\tilde{x_i})}{\epsilon^4} + \frac{\E_i \tr(\widetilde{\Sigma}_i)}{\epsilon^2/k \cdot \E_{i, j \sim [n]} d_{ij}^2}\right)\,. \]
    Noting that $\E_{i} H(\tilde{x_i}) \leq O(k \log(k/\epsilon)$ by~\cref{lem:rs-entropy-bound} and $\E_{i \sim [n]} \tr(\widetilde{\Sigma}_i) \leq \E_{i \sim [n]} \pE_\mu \norm{x_i}^2 \leq \E_{i,j \sim [n]} d_{ij}^2$ by our psuedodistribution constraint yields that
    \[\vert \calT \vert \leq O\left(k \log (k/\epsilon) /\epsilon^4 \right)\,.\]
\end{proof}

Finally, we can now put the pieces together.
\begin{proof}[Proof of $ (nB)^{\mathcal{O}( k^2\log(k/\eps)/\eps^4)}$ time version of~\cref{thm:raw-stress-main}]
    Note that we can take the best solution produced by rounding all of the pseudo-distributions, so it suffices to show that with high probability, a "good" round of our SDP choice and rounding will produce a solution that is additively close to an optimal solution. 

    Note that by~\cref{lem:rs-pe-low-value} there exists $a$ such that the pseudo-distribution is both feasible and there exists a solution in the support with value at most $\OPT_{\mathrm{EMV}} + \epsilon \E_{i,j\sim[n]} d_{ij}^2$. We will argue that rounding this solution produces an embedding with value at most $\OPT_{\mathrm{EMV}} + \bigO{ \epsilon} \cdot \E_{i,j\sim[n]} d_{ij}^2$ with probability $0.99$. 

    By~\cref{lem:rs-gcr-error} in combination with the fact that conditioning preserves the objective value in expectation, we have that the expected value of the solution sampled in the independent rounding step is $\OPT_{\mathrm{EMV}} + c \cdot \epsilon \cdot \E_{i,j\sim[n]} d_{ij}^2$ for some constant $c$. However, since any possible solution sampled has value at least $\OPT_{\mathrm{EMV}}$, we must have the probability that we sample a solution which has value greater than $\OPT_{\mathrm{EMV}} + 100 c \cdot \epsilon \cdot \E_{i,j\sim[n]} d_{ij}^2$ is at most $0.01$. Finally, letting $\epsilon' = \epsilon / (100 c)$ gives the final claim.

    Finally, note that the total degree required for this process is $k \cdot \vert \calT \vert$ since for each step of conditioning we must fix the value of $k$ variables in the program, one for each coordinate of $x_i$. Thus, the total degree needed is at most $k^2 \log (k/\epsilon) /\epsilon^4$, and thus the degree of our computed pseudo-distributions is sufficiently high. The running time is dominated by the cost at each iteration of finding an optimizing pseudodistribution and thus is bounded by $(n \log \Delta)^{\mathcal{O}( k\log(k/\eps)/\eps^4)}$.
\end{proof}

\subsection{Weighted \textsf{$k$-EMV}: Algorithm and Analysis}

In this section, we extend the algorithm for \textsf{$k$-EMV} to weighted instances, and in particular, every pair of distances need not appear in the objective. 

\begin{theorem}[Weighted \textsf{$k$-EMV}]
\label{thm:regular-raw-stress}
Let $\Set{d_{ij}}_{i,j \in [n]}$ be distances with aspect ratio $\Delta$,  and let $\Set{w_{i,j}}_{i,j \in [n]}$ be weights $0\leq w_{ij} \leq 1$ such that corresponding weighted graph is $(\delta \cdot n)$-regular.
Given an integer $k \in \mathbb{N}$ and $0<\eps<1$, there exists an algorithm that runs in $\Paren{n \Delta}^{ \mathcal{O}(k^2\log(\Delta \cdot k/(\eps \delta)) / (\eps^2 \delta^{11}) ) }$ time and with probability at least $0.99$ outputs an embedding $\Set{ \hat{x}_i }_{i \in [n]}$ such that 
\begin{equation*}
    \expecf{i,j \sim W}{ \Paren{ d_{ij }  - \norm{ \hat{x}_i - \hat{x}_j} }^2  } \leq \OPT_{\textrm{WEMV}} +  \eps \cdot \expecf{i,j\sim W}{ d_{ij}^2 } \,.
\end{equation*}
\end{theorem}

Our argument proceeds by showing that we can obtain an initial variance bound by decomposing the regular graph into a small number of expanders. Then, we show that the rounding error on the weighted instance can be related to the rounding error on the complete instance. We begin by introducing a new expander decomposition for dense, regular graphs.

\begin{mdframed}
  \begin{algorithm}[Additive Approximation Scheme for Weighted \textsf{$k$-EMV} ]
    \label{algo:weighted-raw-stress}\mbox{}
    \begin{description}
    \item[Input:] Distances $ \{d_{ij}\}_{i,j\in[n]}$, weights $\Set{ w_{ij} }_{i,j \in [n]}$ forming a $\delta \cdot n$ regular graph,  target dimension $k \in \mathbb{N}$, target accuracy $0<\epsilon<1$.
    
    \item[Operations:]\mbox{}
    \begin{enumerate}
        \item For every $\calT \subset [n]$ such that $\abs{ \calT } = \bigO{ k^2\log(\Delta \cdot k/(\eps \delta))/ (\eps^2 \delta^{11}) }$, 
        \begin{enumerate}
            \item Consider a grid of size length $\epsilon \cdot \sqrt{\frac{1}{k} \cdot \E_{i, j \sim W} d_{ij}^2}$ over $[-n\Delta , n\Delta]^k$.  Compute a sum-of-squares pseudo-distribution $\mu$ over the grid with  degree-$\bigO{ k^2\log(\Delta \cdot k/(\eps \delta)/ (\eps^2 \delta^{11}) }$ such that $\pE_{\mu}$ optimizes: 
        \begin{equation*}
            \min_{ \pE_{\mu} } \hspace{0.1in} \pE_{\mu} \E_{i,j \sim W} \Paren{ d_{ij} - \norm{x_i - x_j} }^2
        \end{equation*}
            \item Let $\bar{x}_\calT$ be a draw from the local distribution $\{x_{\calT}\}$. Let $\mu_\calT$ be the pseudo-distribution obtained by conditioning on $\{x_i = \bar{x}_i\}_{i \in \calT}$.
            \item Compute the objective value of the embedding $\hat{X}_\calT  = \{\hat{x}_i\}_{i \in [n]}$ where $\hat{x}_i$ is sampled independently from the the $k$-degree local distribution of $\{x_i\}$ in $\mu_\calT$.
        \end{enumerate}
    
    \end{enumerate}
    \item[Output:] The embedding $\hat{X}_\calT$ with the lowest Weighted \textsf{$k$-EMV}  objective value.
    \end{description}
  \end{algorithm}
\end{mdframed}

\begin{remark}
    Note that our discretization preserves the objective value up to an additive factor of $\eps \leq \eps \E_{i,j \sim W} d_{ij}^2$. The proof of this fact follows via a simpler analog of the proof of~\cref{lem:discretization-rs}.
\end{remark}

\newcommand{\vol}{\text{vol}}
\paragraph{Partitioning.}
In particular, we show that an $\delta \cdot n$ regular graph can be partitioned into $\bigO{1/\delta^2}$ components such that the induced sub-graph on each component is a $\poly(\delta)$-expander and at most half the edges are discarded. To set up notation,  for a weighted graph $G = (V,E)$, we write $E(A,B)$ for the total edge weight from $A \subseteq V$ to $B \subseteq V$.
Following \cite{Tanaka2011,Gharan2014}, we use the following notion of multi-way conductance:
\begin{definition}[Multi-Way Conductance]
    Let $G = (V,E)$ be a graph and $k \in \N$.
    Then, the multi-way conductance is defined as
    \[
    \rho_G(k) = \min_{A_1,\ldots,A_k \, \text{disjoint}} \max_{j \in [k] }  \hspace{0.1in}\phi_G(A_j) \,, 
    \]
    where $A_1,\ldots,A_k \subseteq V$.
\end{definition}
Oveis Gharan and Trevisan~\cite{Gharan2014} prove the following result relating graph partitions to multi-way conductance:

\begin{theorem}[\cite{Gharan2014}, Theorem 1.7]
\label{thm:partitioning-into-expanders}
  Let $G = (V,E)$.
  If $\rho_G(k+1) > (1+\delta ) \rho_G(k)$ for some $0 < \delta < 1$, then there exists a partitioning of $V$ into $P_1,\ldots,P_k$ such that for all $i \leq k$, we have (a) $\phi_G(P_i) \leq k \cdot \rho(k)$ and (b) $\phi(G[P_i]) \geq \delta  \cdot \rho(k+1)/(14 k)$, where $G[P_i]$ is the subgraph of $G$ induced by the vertices $P_i$.
\end{theorem}

We will need a simple lemma about the multi-way conductance of regular dense graphs.

\begin{lemma}
\label{fact:dense-multiway}
    Let $G = (V,E)$ be an $n$-vertex graph with edge weights $w_{ij} \in [0,1]$, where every vertex has degree $\delta n$.
    Then for every $k$, $\rho_G(k) \geq 1 - \tfrac 1 {\delta k}$.
\end{lemma}
\begin{proof}
WLOG $k \geq 2$, since the statement is trivial when $k = 1$.
Let $A_1, \ldots, A_k$ be any collection of disjoint subsets of $V$. By disjointness, there exists $A_i$ with $|A_i| \leq n/k$.
The total weight of edges internal to $A$ is at most $\binom{|A|}{2} \leq |A|^2$, while the volume of $A$ is $|A| \delta n$.
So the conductance of $A$ satisfies
\[
\phi_G(A) \geq \frac{\delta n |A| - |A|^2}{\delta n |A|} = 1 - \frac{|A|}{\delta n} \geq 1 - \frac{1}{\delta k} \, ,
\]
and by definition $\rho_G(k) \geq \phi_G(A) \geq (1 - 1/(\delta k)$, as desired. %
\end{proof}

We can put these together to arrive at a partitioning-into-expanders lemma for dense graphs.

\begin{lemma}[Expander Decomposition for regular Graphs]
\label{lem:expander-decomposition}
    Let $G = (V,E)$ be an $n$-vertex regular graph with edges weights $w_{ij} \in [0,1]$, where every vertex has weighted degree $\delta \cdot n$.
    There exists a partition of $V$ into $P_1,\ldots,P_k$ such that each $P_i$ has $\phi(G[P_i]) \geq \Omega(\delta^{10})$, $\phi_G(P_i) \leq 1/4$, and $\sum_{i \neq j} E(P_i,P_j) \leq \tfrac 12 |E|$, for some $k \leq O(1/\delta^2)$.
\end{lemma}
\begin{proof}
    If $\rho(2) \geq \delta^5$ then we are done, since then $G$ is already an $\delta^5$-expander.
    So assume $\rho(2) \leq \delta^5$, and let $\ell$ be the greatest index such that $\rho(\ell) \leq \delta^5$. Observe, by \cref{fact:dense-multiway}, $\rho(1/ (10\delta)) \geq 0.9$ and therefore $\ell < 1/(10\delta)$. 
    We now find $ \ell \leq k \leq 1/\delta^2 +1$ such that $\rho(k+1) \geq (1+ \delta^3) \rho(k)$ and $k \rho(k) \leq 1/4$. 
    
    Let $k_0$ be the first integer such that $\rho(k_0+1) > \tfrac 1 {4(k_0+1)}$.
    By Fact~\ref{fact:dense-multiway}, $\rho(1/\delta^2) \geq 1-\delta$ and so $k_0 \leq 1/\delta^2$.
    If for every $\ell \leq k \leq k_0$ we had $\rho(k+1) < (1+\delta^3) \rho(k)$, then we would have $\rho(k_0+1) \leq (1+\delta^3)^{k_0} \delta^5 \leq \exp( \delta ) \delta^5 < \delta^2 \leq 1/{2(k_0+1)}$, which goes against the definition of $k_0$.
    So there exists $\ell \leq k \leq k_0$ satisfying $\rho(k+1) \geq (1+\delta^3) \rho(k)$ and $k \rho(k) \leq 1/4$.

    Now we apply Theorem~\ref{thm:partitioning-into-expanders} with this value of $k$, obtaining a partition $P_1,\ldots,P_k$ such that $\phi(G[P_i]) \geq \delta^3 \rho(k+1) / (14 k) = \Omega(\delta^{10})$ while $\phi_G(P_i) \leq 1/4$.

    It just remains to check that this implies that $\sum_{i \neq j} E(P_i,P_j) \leq \tfrac 12 |E|$.
    For each $i$, we have $\sum_{j \neq i} E(P_i,P_j) \leq \phi_G(P_i) \cdot \vol(P_i) \leq \vol(P_i)/4 \leq |P_i| \cdot (\delta n /4)$.
    Summing across all $i$ finishes the argument.
\end{proof}

Next, we show that after conditioning on a single point, the variance is bounded for the weighted objective as well: 

\begin{lemma}[Variance after conditioning]
\label{cor:magnitude-points-wrs}
    Let $\{ d_{i,j} \}_{i,j \in [n]}$ be a set of distances and $\Set{w_{ij}}_{i,j \in [n]}$ be a set of non-negative weights such that for all $i \in [n]$, $\sum_{j \in [n]} w_{i,j} = \delta \cdot n$ for some $0<\delta<1$. Then for any pseudo-distribution over $x_1,\ldots, x_n$ for the Weighted \textsf{$k$-EMV}  objective, there exists a set of $k=\bigO{1/\delta^2}$ points $\Set{ x_{j_{\ell} } }_{ \ell  \in [k] }$ such that 
    \[ \E_{i\sim W} \tilde{\Var}( x_i  \mid     x_{j_1}, \ldots , x_{j_k} )   \leq \bigO{ 1/\delta^{10} } \cdot \E_{i,j\sim W} d_{ij}^2\,,\]
    where the distribution $i \sim W$ samples indices $i \propto \sum_{j} W_{ij}$ and the distribution $i, j \sim W$ samples the pair $i,j \propto W_{ij}$.
\end{lemma}
\begin{proof}
Observe, since each row has weight $\eta n$, for any fixed $c= x_a$ 
\begin{equation*}
    \expecf{i \sim W}{ \tilde{\Var}( x_i  \mid     x_{j_1}, \ldots , x_{j_k} ) } = \expecf{ i \sim [n] }{ \tilde{\Var}( x_i  \mid     x_{j_1}, \ldots , x_{j_k} ) }
\end{equation*}
Next, it follows from \cref{lem:expander-decomposition} that there exists a partition $\Set{P_1, \ldots, P_k}$ such that $\phi(G[P_i]) \geq \delta^{10}$. Consider the random process where $x_{j_{\ell}}$ is picked uniformly from $P_{\ell}$. Then,
\begin{equation}
\label{eqn:too-tired-to-name-this}
    \begin{split}
        \expecf{ x_{j_1}, \ldots , x_{j_k}  }{ \expecf{i \sim [n] }{ \tilde{\Var}( x_i  \mid     x_{j_1}, \ldots , x_{j_k} ) } } & \leq  \expecf{ x_{j_1}, \ldots , x_{j_k}  }{ \expecf{i \sim [n] }{ \tilde{\Var}( x_i  \mid     x_{j(i)}} }  \\
        & \leq \expecf{ x_{j_1}, \ldots , x_{j_k}  }{ \expecf{i \sim [n] }{  \pE \norm{ x_i - x_{j(i)}}^2  } }  \\
        & = \frac{1}{n} \sum_{ \ell \in [k] } \frac{\abs{P_\ell } }{\abs{P_\ell }^2 }  \sum_{ i, a \in P_{\ell} } \pE \norm{ x_i - x_a }^2 \\
        & = \frac{1}{n} \sum_{ \ell \in [k] } \abs{P_\ell } \cdot   \expecf{i, a \sim P_\ell }{  \pE \norm{ x_i - x_a }^2  }
    \end{split}
\end{equation}
where $x_{j(i)}$ corresponds to the point selected in the same partition as $x_i$. Further, since $P_\ell$ is an expander, by a Poincar\'e inequality combined with \cref{fact:nonnegative-quadratic}, we have 
\begin{equation}
\label{eqn:averaging-over-expanders}
   \frac{1}{n} \sum_{ \ell \in [k] } \abs{P_\ell } \cdot   \expecf{i, a \sim P_\ell }{  \pE \norm{ x_i - x_a }^2  } \leq \Paren{\frac{1}{\delta^{10}} } \Paren{ \frac{1}{n}  \sum_{ \ell \in [k] }   \Paren{ \frac{ \abs{P_\ell }  }{ \sum_{i, a \in P_\ell } w_{ia}   }  } \sum_{i, a \in P_\ell } w_{ia} \pE \norm{ x_i - x_a }^2  }  
\end{equation}
Next, we show that $\sum_{i, a \in P_{\ell}} w_{ia} \geq \frac{3}{4} \cdot \delta n \cdot  \abs{P_\ell} $. Recall, from \cref{lem:expander-decomposition}, we have $\phi_G(P_i) \leq 1/4$ and therefore, 
\begin{equation*}
    \sum_{i, a \in P_{\ell}}  w_{i,a} = \Abs{ E( P_\ell, P_\ell) } \geq \delta n \cdot \Abs{P_\ell} - \frac{1}{4} \textrm{vol}(P_{\ell}) \geq \frac{3}{4} \delta n \cdot \Abs{P_{\ell}}. 
\end{equation*}
Substituting this back into \cref{eqn:averaging-over-expanders}, we have
\begin{equation*}
    \begin{split}
        \frac{1}{n} \sum_{ \ell \in [k] } \abs{P_\ell } \cdot   \expecf{i, a \sim P_\ell }{  \pE \norm{ x_i - x_a }^2  } & \leq \Paren{\frac{4}{3 \delta^{11}} } \Paren{ \frac{1}{n^2} \sum_{ i, a \in [n]}  w_{ia} \cdot \pE \norm{ x_i - x_a }^2  } \\
         & \leq \Paren{\frac{4}{3 \delta^{10}} } \Paren{\expecf{i,a \sim W}{ \pE \norm{ x_i - x_a }^2   }  }
    \end{split}
\end{equation*}
Then we have by Almost Triangle Inequality that
    \[ \E_{a,i \sim W } \norm{x_i - x_a}^2 =  \E_{a,i \sim W } \left(\norm{x_i-x_a} \pm d_{ia}\right)^2 \leq \E_{a,i \sim W } 2\left(d_{ia} - \norm{x_i-x_a}\right)^2 + 2d_{ia}^2\,.\]
    The first term is $\OPT_{\mathrm{WEMV}}$, which is bounded by $\E_{i,j\sim W} d_{ij}^2$ since there exists a trivial solution achieving this value which just maps every point to $0$. Thus, plugging this back into \cref{eqn:too-tired-to-name-this}, we can conclude that 
    \begin{equation*}
        \expecf{ x_{j_1}, \ldots , x_{j_k}  }{ \expecf{i \sim [n] }{ \tilde{\Var}( x_i  \mid     x_{j_1}, \ldots , x_{j_k} ) } } \leq \Paren{ \frac{8}{\delta^{10}} } \Paren{ \expecf{i,j \sim W}{ d_{i,j}^2 } }
    \end{equation*}
\end{proof}

Note that the above lemma also implies a bound on the entropy after conditioning on the initial set of points:
\begin{lemma}[Weighted \textsf{$k$-EMV} has Bounded Entropy]
\label{lem:weighted-rs-entropy}
    Let $\mu$ be a pseudo-distribution over a grid of length $\epsilon$ which satisfies
    \[ \E_{i \sim  W } \pE_\mu  \tilde{\Var}(x_i ) \leq \bigO{1/\delta^{10}} \E_{i,j \sim  W } d_{ij}^2\,.\]
    Let $\widetilde{x_i}$ is $x_i(T) E_i + (\pE_\mu x_i) \overline{E}_i$ where the event $E_i$ is $E_i = (\norm{x_i - \E_\mu x_i} \leq C \cdot \sqrt{\pE_\mu \norm{x_i - \pE_\mu x_i}^2})$. Then we have that
    \[ \E_{i \sim W} H(\widetilde{x}_i) \leq \mathcal{O}\left(k \log (C \Delta k/(\epsilon \delta^11))\right)\,.\]
\end{lemma}

The proof of this lemma follows the proof of~\cref{lem:rs-entropy-bound}.

\begin{lemma}[Rounding Weighted \textsf{$k$-EMV} ]
\label{lem:weighted-rs-gcr-error}
    Let $\epsilon \geq 0$ and let $\mu$ be an initial pseudo-distribution which satisfies
    \[ \E_{i \sim  W } \pE_\mu  \tilde{\Var}(x_i ) \leq \bigO{1/\delta^{10}} \E_{i,j \sim  W } d_{ij}^2\,.\]
    Then there exists a set $\calT$ of size at most $\bigO{k \log (k/\epsilon) /\epsilon^2 \delta }$ such that after conditioning on the values of $x_i \in \calT$,
    \[ \left\vert \E_{\{\mu_\calT\}} \pE_{\mu_{\calT}} \E_{i,j \sim W } (d_{ij} - \norm{x_i - x_j})^2 - \pE_{\mu^{\otimes}_{\calT}} \E_{i,j \sim  W} (d_{ij} - \norm{x_i - x_j})^2 \right\vert \leq \epsilon \cdot  \E_{i,j\sim W } d_{ij}^2 \,.\]
\end{lemma}
\begin{proof}
    Note that we can write
    \[ \E_{i,j\sim W } (d_{ij} - \norm{x_i - x_j})^2  = \E_{i,j\sim W  } \left[d_{ij}^2 - 2 d_{ij} \norm{x_i - x_j} + \norm{x_i}^2 + \norm{x_j}^2 - 2 \langle x_i, x_j\rangle \right]\,.\]
    Note that out of these terms the only two which are different in expectation under $\mu_\calT$ versus $\mu_\calT^\otimes$ are $2 d_{ij} \norm{x_i - x_j}$ and $2 \langle x_i, x_j\rangle$ and therefore, 
    \begin{equation}
    \label{eqn:rounding-cost-on-weighted-raw-stress}
    \begin{split}
         & \left\vert \E_{\{\mu_\calT\}} \pE_{\mu_{\calT}} \E_{i,j \sim W } (d_{ij} - \norm{x_i - x_j})^2 - \pE_{\mu^{\otimes}_{\calT}} \E_{i,j \sim  W} (d_{ij} - \norm{x_i - x_j})^2 \right\vert \\
         & \leq   2 \underbrace{ \Abs{ \E_{\{\mu_\calT\}} \pE_{\mu_{\calT}}  \expecf{i,j \sim W}{  \Iprod{x_i , x_j } } - \pE_{\mu_{\calT}^\otimes} \expecf{i,j \sim W}{  \Iprod{x_i, x_j}}  } }_{\eqref{eqn:rounding-cost-on-weighted-raw-stress}.(1) }  \\
          & \quad  \quad + 2 \underbrace{ \Abs{  \E_{\{\mu_\calT\}} \pE_{\mu_{\calT}}  \expecf{i,j \sim W}{ d_{ij} \cdot \norm{x_i - x_j} } - \pE_{\mu_{\calT}^\otimes} \expecf{i,j \sim W}{ d_{ij} \cdot \norm{x_i - x_j} }   } }_{\eqref{eqn:rounding-cost-on-weighted-raw-stress}.(2)}  \, 
    \end{split}
    \end{equation}
    Focusing on term \eqref{eqn:rounding-cost-on-weighted-raw-stress}.(1), 
    \begin{equation*}
    \begin{split}
        \eqref{eqn:rounding-cost-on-weighted-raw-stress}.(1) & \leq  \expecf{i,j \sim W}{  \Abs{ \E_{\{\mu_\calT\}} \pE_{\mu_{\calT}} \Iprod{x_i , x_j}    -  \pE_{\mu_{\calT}^\otimes } \Iprod{x_i , x_j}  }   } \\
        & \leq \Paren{ \frac{1}{\delta} }  \expecf{i,j \sim [n]}{  \Abs{ \E_{\{\mu_\calT\}} \pE_{\mu_{\calT}} \Iprod{x_i , x_j}    -  \pE_{\mu_{\calT}^\otimes } \Iprod{x_i , x_j}  }   } 
    \end{split}
    \end{equation*}
    Following the proof of~\cref{lem:weighted-rs-gcr-error}, we have that that the function $\expecf{i,j \sim [n]}{ \Abs{\Iprod{x_i , x_j} } }$ is 
    \begin{equation*}
        \Paren{ \eps \cdot \E_{i,j \sim W } d_{ij}^2  ,  \hspace{0.1in}\Paren{\frac{\eps^2 \delta^{10}}{k }}\cdot  \E_{i,j \sim W } d_{ij}^2   }-\textrm{potential aligned}\,,
    \end{equation*}
    for the potential $\Phi = \expecf{i \sim [n]}{ \tr(\tilde{\Sigma}_i) }$.
    Next, for  term \eqref{eqn:rounding-cost-on-weighted-raw-stress}.(2) we need to modify the potential function for Lipschitz functions as follows: similar to the proof of \cref{lem:lipschitz-rounding-error}, let $\widetilde{x_i} = x_i E_i + (\pE_{\mu}[ x_i] ) \overline{E}_i$, where $E_i$ be the event that $\norm{x_i - \pE_\mu x_i} \leq C \cdot \sqrt{\pE_\mu \norm{x_i - \pE_\mu x_i}^2}$. Then,

    \begin{equation}
    \label{eqn:split-for-lipschitz-term-under-w}
    \begin{split}
        \eqref{eqn:rounding-cost-on-weighted-raw-stress}.(2) & \leq   \underbrace{ \left\vert \E_{\{\mu'\}} \E_{i,j \sim W} \left[\pE_{\mu^{'\otimes}} d_{ij}\cdot \norm{ \widetilde{x}_i-  \widetilde{x}_j } - \pE_{\mu'} d_{ij} \cdot \norm{ \widetilde{x}_i-  \widetilde{x}_j} \right] \right\vert }_{\eqref{eqn:split-for-lipschitz-term-under-w}.(1)} \\
        &  \qquad + \underbrace{ \left\vert \E_{\{\mu'\}} \E_{i,j \sim W } \pE_{\mu} \left[ \left( d_{ij}\norm{ \widetilde{x}_i-\widetilde{x}_j } - d_{ij} \norm{ x_i - x_j} \right) \right] \right\vert}_{\eqref{eqn:split-for-lipschitz-term-under-w}.(2)} \\
        &  \qquad + \underbrace{\left\vert \E_{\{\mu'\}} \E_{i,j \sim W} \pE_{\mu^{'\otimes}} \left[\left( d_{ij}\norm{ y_i- y_j } - d_{ij}\norm{ \widetilde{y}_i-\widetilde{y}_j}\right)\right] \right\vert}_{\eqref{eqn:split-for-lipschitz-term-under-w}.(3)} \,. 
    \end{split}
    \end{equation}
    Starting with \eqref{eqn:split-for-lipschitz-term-under-w}.(1), we have 
    \begin{equation}
    \label{eqn:yet-another-equation-to-bound}
        \begin{split}
            & \left\vert \E_{\{\mu'\}} \E_{i,j \sim W} \left[\pE_{\mu^{'\otimes}} d_{ij}\cdot \norm{ \widetilde{x}_i-  \widetilde{x}_j } - \pE_{\mu'} d_{ij} \cdot \norm{ \widetilde{x}_i-  \widetilde{x}_j} \right] \right\vert \\
            & \leq \E_{\{\mu'\}} \E_{i,j \sim W} \left[ \tv \Paren{ \mu'_{\{\widetilde{x_i}, \widetilde{x_j}\}}, \mu'_{\{\widetilde{x_i}\} \otimes \{ \widetilde{x_j} \}}} \cdot 2d_{ij} \cdot C \cdot \left(\sqrt{\pE_\mu \norm{x_i - \pE_\mu x_i}^2} + \sqrt{\pE_\mu \norm{x_j - \E x_j}^2} \right) \right] \\
        &\leq 4C \E_{\{\mu'\}} \E_{i \sim  W } \left[ \sqrt{\pE_\mu \norm{x_i - \pE_\mu x_i}^2} \cdot  \sum_{j \in [n]}  \Paren{ \frac{ w_{ij}\cdot d_{ij} }{\delta n} } \tv \Paren{ \mu'_{\{\widetilde{x_i}, \widetilde{x_j}\}}, \mu'_{\{\widetilde{x_i}\} \otimes \{ \widetilde{x_j} \}}}   \right] \\
        & \leq 4C \Paren{ \expecf{i \sim W}{    \pE_\mu \norm{x_i - \pE_\mu x_i}^2 } }^{\frac12} \cdot \underbrace{ \Paren{\E_{\{\mu'\}} \E_{i \sim  W } \Paren{ \sum_{j \in [n]}  \Paren{ \frac{ w_{ij}\cdot d_{ij} }{\delta n} } \tv \Paren{ \mu'_{\{\widetilde{x_i}, \widetilde{x_j}\}}, \mu'_{\{\widetilde{x_i}\} \otimes \{ \widetilde{x_j} \}}}}^2  }^{\frac12}}_{\eqref{eqn:yet-another-equation-to-bound}.(1)} 
        \end{split}
    \end{equation}
    By Cauchy-Schwarz followed by Pinsker's we have
    \begin{equation}
    \begin{split}
         \eqref{eqn:yet-another-equation-to-bound}.(1) & \leq   \Paren{ \expecf{i \sim W}{ \Paren{ \sum_{j \in [n]}  \Paren{ \frac{ w_{ij} }{\delta n}   } d_{i,j}^2  } \cdot  \Paren{ \E_{\{\mu'\}}   \sum_{j \in [n]} \Paren{ \frac{ w_{ij} }{\delta n}   }  \tv \Paren{ \mu'_{\{\widetilde{x_i}, \widetilde{x_j}\}}, \mu'_{\{\widetilde{x_i}\} \otimes \{ \widetilde{x_j} \}}}^2    } }   }^{1/2}  \\
         & \leq \Paren{ \expecf{i \sim W}{ \Paren{ \sum_{j \in [n]}  \Paren{ \frac{ w_{ij} }{\delta n}   } d_{i,j}^2  } \cdot  \Paren{ \E_{\{\mu'\}}   \sum_{j \in [n]} \Paren{ \frac{ w_{ij} }{\delta n}   }     \mathsf{I}(\mu'_{\{\widetilde{x_i}, \widetilde{x_j}\}}; \mu'_{\{\widetilde{x_i}\} \otimes \{ \widetilde{x_j} \}})    } }   }^{1/2} 
    \end{split}
    \end{equation}
    Let $\psi$ be the distribution where $i$ is sampled proportional to $\sum_{j \in [n]}  \Paren{ \frac{ w_{ij} }{\delta n}   } d_{i,j}^2$. Then, 
    \begin{equation*}
    \begin{split}
        \Paren{ \eqref{eqn:yet-another-equation-to-bound}.(1)}^2 & \leq \cdot \expecf{i,j \sim W}{ d_{i,j}^2 } \cdot  \E_{\Set{\mu'}} \expecf{i \sim \psi}{ \sum_{j \in [n]} \Paren{ \frac{ w_{ij} }{\delta n}   } \Paren{  H(\widetilde{x}_j) - H(\widetilde{x}_j \vert \widetilde{x}_i)  }    } \\
        & \leq \Paren{ \frac{1}{\delta} } \expecf{i,j \sim W}{ d_{i,j}^2 } \cdot \E_{\Set{\mu'}} \expecf{i \sim \psi}{  \expecf{j \sim [n]}{  H(\widetilde{x}_j) - H(\widetilde{x}_j \vert \widetilde{x}_i)  }    } \,. 
    \end{split}
    \end{equation*}
Substituting this back into \cref{eqn:yet-another-equation-to-bound}, we have
\begin{equation}
\begin{split}
    \eqref{eqn:split-for-lipschitz-term-under-w}.(1) & \leq \Paren{ \frac{4C}{\delta} } \Paren{ \expecf{i \sim W}{    \pE_\mu \norm{x_i - \pE_\mu x_i}^2 } }^{\frac12} \Paren{ \expecf{i,j \sim W}{ d_{i,j}^2 } \cdot \E_{\Set{\mu'}} \expecf{i \sim \psi}{  \expecf{j \sim [n]}{  H(\widetilde{x}_j) - H(\widetilde{x}_j \vert \widetilde{x}_i)  }    } }^{\frac12} \\
    & \leq \Paren{ \frac{4C}{\delta^6} } \cdot \expecf{i,j \sim W}{ d_{i,j}^2 }  \cdot \Paren{  \E_{\Set{\mu'}} \expecf{i \sim \psi}{  \expecf{j \sim [n]}{  H(\widetilde{x}_j) - H(\widetilde{x}_j \vert \widetilde{x}_i)  }    } }^{\frac12}
\end{split}
\end{equation}
Now, we observe that the proof of \cref{lem:tilde-vs-not-diff} remains unchanged if we switch to $i,j \sim W$, and we can conclude that 
\begin{equation*}
\begin{split}
   \eqref{eqn:split-for-lipschitz-term-under-w}.(2),  \eqref{eqn:split-for-lipschitz-term-under-w}.(3) & \leq \bigO{1/C} \left(\E_{i \sim W } \pE_\mu \norm{x_i - \pE_\mu x_i}^2\right)^{1/2} \cdot \left( \E_{i,j\sim W} d_{ij}^2\right)^{1/2} \\
   & \leq \bigO{1/(C\delta^5 ) } \expecf{i,j \sim W}{ d_{ij}^2 }  
\end{split}
\end{equation*}
Setting $C = \bigO{ 1/ (\eps \delta^5) }$, we can conclude that 
if $ \eqref{eqn:rounding-cost-on-weighted-raw-stress}.(2) > \eps \expecf{i,j \sim W}{ d_{i,j}^2 } $, the decrease in potential is at least $ \eps^2\delta^{11} $. Further, since the potential starts out at at most $\bigO{ k  \log( \Delta \cdot k/(\eps \delta))) }$ by~\cref{lem:weighted-rs-entropy} it is bounded by $\eps$ after $\bigO{  k \log( \Delta \cdot k/(\eps \delta))/(\eps^2 \delta^{11})  }$ rounds.  Since both our potentials are now aligned, we can conclude the proof.

\end{proof}

We can now prove \cref{thm:regular-raw-stress} similar to how we proved \cref{thm:raw-stress-main}, and note that the running time is dominated by computing a pseudo-distribution of degree $$\bigO{k^2 \log( \Delta \cdot k/(\eps \delta))/(\eps^2 \delta^{11})},$$ since each round of conditioning requires degree $k$.

\section{Entry-wise Low-Rank Approximation}
In this section, we present an additive approximation scheme for entry-wise rank-$1$ approximation.

\begin{theorem}[Entry-wise LRA]
\label{thm:lra}
Given a matrix $A \in \R^{n \times m}$, integer $k \in N$, even $p \in N$, and $0<\eps<1$, there exists an algorithm that runs in $n^{ \mathcal{O}_p(\log^2 (n)/\poly_p(\eps))}$ time and outputs matrices $\hat{u} \in \R^{n}, \hat{v} \in \R^{m}$ such that with probability at least $0.99$, 
\begin{equation*}
    \norm{ A - \hat{u} \hat{v}^\top }_p^p \leq \OPT_{\textrm{LRA}} + \eps \cdot \Norm{A}_p^p\,. 
\end{equation*}
\end{theorem}

The algorithm we give for this problem is based on the global correlation rounding algorithm and thus the key technical innovation is to show that the objective function can be decomposed into a sum of potential-aligned functions.

\subsection{Rounding Polynomial Growth Functions}
\label{sec:poly-growth-potential}
In this section, we will show that functions of the form $\sum_{i \in [n], j \in [m]} A_{ij}^q \langle u_i, v_j \rangle^{p-q}$ are potential aligned for an appropriately chosen potential function and a certain class of pseudo-expectations, where $p,q$ are integers. Note that depending on the degree $q$ we will need a different potential function -- we first prove this in the case where $q < p/2$ and then for the case where $q \in [p/2, p-1]$.

\subsubsection{Rounding High Degree Polynomial Growth Functions}

\begin{lemma}[Polynomial-growth functions are potential aligned]
\label{lem:poly-potential-aligned}
    Let $0 < \epsilon \leq 1$ and let $\mu$ be some pseudo-distribution. Let $p, q \in \mathbb{N}^{>0}$ such that $q < p/2$ and let $\{A_{ij}\}_{i \in [n], j \in [m]} \in \mathbb{R}$ be a collection of real coefficients. Furthermore, let $\mu$ satisfy the following inequalities:
    \begin{enumerate}
        \item $\sum_{i \in [n]} \E_{\mu} u_i^p \leq 2^{p} \sqrt{\sum_{i \in [n], j \in [m]} A_{ij}^p}$
        \item $\sum_{j \in [m]} \E_{\mu} v_j^p \leq 2^{p} \sqrt{\sum_{i \in [n], j \in [m]} A_{ij}^p}$
        \item $\sum_{i \in [n], j \in [m]} \left(\E_{\mu} u_i^p\right) \cdot \left(\E_{\mu} v_j^p\right) \leq 2^{p+1} \sum_{i \in [n], j \in [m]} A_{ij}^p$
    \end{enumerate}
    Further, let $E_j^{v}$ be the event that $\abs{v_j} \leq C ( \pE_{\mu} v_j^p )^{1/p}$ Suppose $\mu'$ is a random pseudo-distribution such that $\E \{\mu'\} = \{\mu\}$ and everything in the support of $\mu'$ also satisfies the inequalities above. Furthermore, $\Phi$ be a function such that the inequalities above imply that
    \[ \Phi(\mu') \leq \Phi_{\max} \leq 2^{p-1} \cdot \log n \cdot \sqrt{\sum_{i, j \in [m]} A_{ij}^p}\,,\]
    where $\Phi = \sum_{j \in [m]} H(\widetilde{v}_j) \E_\mu v_j^p$, and $\tilde{v}_j = v_j \cdot E_j^{v}$. Then $\sum_{i \in [n], j \in [m]} A_{ij}^q (u_i v_j)^{p-q}$ is
    \[\left(1/\poly_{p,q}(\epsilon) \left(\sum_{i \in [n], j \in [m]} A_{ij}^p\right)^{1/2}, \Omega_p(\epsilon/\log n) \left(\sum_{i \in [n], j \in [m]} A_{ij}^p\right)^{1/2}\right)\]
    potential aligned for $\mu$ and $\Phi$.
\end{lemma}

Before proving this lemma, we will first show that if a pseudoexpectation satisfies a certain approximate pairwise independence property that the difference in expectation over $\mu$ versus $\mu^{\otimes}$ is bounded.
\begin{lemma}
\label{lem:polynomial-rounding-error}
    Let $0 < \epsilon \leq 1$ and let $\mu$ be some pseudo-distribution. Suppose $\mu'$ is a random psuedodistribution such that $\E \mu' = \mu$. Furthermore, let $E_i^u$ be the event that $\abs{u_i} \leq C \left( \E_{\mu} u_i^p \right)^{1/p}$. Furthermore, let $\widetilde{u}_i = u_i E_i^u$. Let $E_j^v$ and $\widetilde{v}_j$ be defined analogously for all $\{v_j\}_{j \in [m]}$. Then, if we have that,
    \[ \E_{\{\mu'\}} \sum_{i \in [n]} \left(\sum_{j \in [m]} A_{ij}^p\right) \left( \sum_{j \in [m]} \tv \Paren{ \mu_{\{\widetilde{u_i}, \widetilde{v_j}\}}, \mu_{\{\widetilde{u_i}\} \otimes \{ \widetilde{v_j} \}}}^{\frac{p}{p-q} } \left( \E_{\mu} v_j^p\right) \right)^{\frac{p-q}{q}} \leq \epsilon \left(\sum_{i \in [n], j \in [m]} A_{ij}^p \right)^{\frac{p+q}{2q}}\,,\]
    then 
    \begin{equation*}
    \begin{split}
        &\left\vert \E_{\{\mu'\}} \sum_{i \in [n], j \in [m]} \pE_{\mu'} A_{ij}^q (u_i v_j)^{p-q} - \sum_{i \in [n], j \in [m]} \pE_{\mu'} A_{ij}^q (u_i v_j)^{p-q} \right\vert \\
        &\leq C^{2(p-q)}\epsilon^{q/p} \left( \sum_{i \in [n]}  \E_{\mu} u_i^p \right)^{\frac {p-q} p} \cdot \left(\sum_{i \in [n], j \in [m]} A_{ij}^p \right)^{\frac{p+q}{2p}} \\
        & \qquad
        + 1/C^q \cdot \left(\sum_{i \in [n], j \in [m]} A_{ij}^p \right)^{\frac q p} \left( \sum_{i \in [n], j \in [m]} \pE_{\mu} \left(u_i v_j\right)^p \right)^{\frac{p-q}{p}}\,.
    \end{split}
    \end{equation*} 
\end{lemma}

In the proof of~\cref{lem:polynomial-rounding-error} we will also need the following lemma, which controls the difference between $u_i, \widetilde{u}_i$ and $v_j, \widetilde{v}_j$ under $\mu$ and $\mu^{\otimes}$.

\begin{lemma}
\label{lem:poly-tilde-vs-not}
    Let $0 < \epsilon \leq 1$ and let $\mu$ be some pseudo-distribution. Suppose $\mu'$ is a random pseudo-distribution such that $\E \mu' = \mu$. Furthermore, let $E_i^u$ be the event that $\abs{u_i} \leq C \left( \E_{\mu} u_i^p \right)^{1/p}$. Furthermore, let $\widetilde{u}_i = u_i E_i^u$. Let $E_j^v$ and $\widetilde{v}_j$ be defined analogously for all $\{v_j\}_{j \in [m]}$. Then,
    \[ \left\vert \E_{\{\mu'\}} \sum_{i \in [n], j \in [m]} \pE_{\mu'} \left[ \left( A_{ij}^q \left(\widetilde{u}_i \widetilde{v}_j\right)^{p-q} - A_{ij}^q \left(u_i v_j\right)^{p-q} \right) \right] \right\vert \leq \frac{1}{C^q}  \left(\sum_{i \in [n], j \in [m]} A_{ij}^p \right)^{\frac q p} \left( \sum_{i \in [n], j \in [m]} \pE_{\mu} \left(u_i v_j\right)^p \right)^{\frac {p-q} p}\,,\]
    and
    \[ \left\vert \E_{\{\mu'\}} \sum_{i \in [n], j \in [m]} \pE_{(\mu')^{\otimes}} \left[ \left( A_{ij}^q \left(\widetilde{u}_i \widetilde{v}_j\right)^{p-q} - A_{ij}^q \left(u_i v_j\right)^{p-q} \right) \right] \right\vert \leq \frac{1}{C^q} \left(\sum_{i \in [n], j \in [m]} A_{ij}^p \right)^{\frac q p} \left( \sum_{i \in [n], j \in [m]} \pE_{\mu} \left(u_i v_j\right)^p \right)^{\frac{p-q} p}\,.\]
\end{lemma}

We will first use~\cref{lem:poly-tilde-vs-not} to prove ~\cref{lem:polynomial-rounding-error} and then return to its proof.

\begin{proof}
    Let $E_i^u$ be the event that $\abs{u_i} \leq C \left( \E_{\mu} u_i^p \right)^{1/p}$. Furthermore, let $\widetilde{u}_i = u_i E_i^u$. Let $E_j^v$ and $\widetilde{v}_j$ be defined analogously for all $\{v_j\}_{j \in [m]}$. Let $f_{ij}(u_i, v_j) = A_{ij}^q \left(u_i v_j \right)^{p-q}$.  Note that we have by Triangle Inequality that
    \begin{equation}
    \label{eq:split-tilde-lra}
    \begin{split}
        & \left\vert \E_{\{\mu'\}} \sum_{i \in [n], j \in [m]} \pE_{\mu'} f_{ij}(u_i, v_j) - \sum_{i \in [n], j \in [m]} \pE_{(\mu')^\otimes} f_{ij}(u_i, v_j) \right\vert \\
        & \quad \leq \left\vert \E_{\{\mu'\}} \sum_{i \in [n], j \in [m]} \pE_{\mu'} f_{ij}(\widetilde{u}_i, \widetilde{v}_j) - \pE_{\{\mu\}} \sum_{i \in [n], j \in [m]} \pE_{(\mu')^{\otimes}} f_{ij}(\widetilde{u}_i, \widetilde{v}_j) \right\vert \\
        & \qquad + \left\vert \E_{\{\mu'\}} \sum_{i \in [n], j \in [m]} \pE_{\mu'} f_{ij}(u_i, v_j) - f_{ij}(\widetilde{u}_i, \widetilde{v}_j)\right\vert + \left\vert \E_{\{\mu'\}} \sum_{i \in [n], j \in [m]} \pE_{(\mu')^{\otimes}} f_{ij}(u_i, v_j) - f_{ij}(\widetilde{u}_i, \widetilde{v}_j)\right\vert
    \end{split}
    \end{equation}
    We will focus on bounding the first term, bounds on the second and third will follow via~\cref{lem:poly-tilde-vs-not}. Note that $\widetilde{u}_i$ has magnitude at most $C \left( \E_{\mu} u_i^p\right)^{1/p}$ and similarly $\widetilde{v}_j$ always has magnitude at most $C \left( \E_{\mu} v_j^p\right)^{1/p}$. Thus, we have with probability $1$ that
    \begin{align*}
        \vert f_{ij}(\widetilde{u}_i, \widetilde{v}_j) \vert &= \vert A_{ij}^{q} \left(\widetilde{u}_i \widetilde{v}_j \right)^{p-q} \vert \\
        &\leq \vert A_{ij}^q \vert \cdot C^{2(p-q)} \cdot \left( \E_{\mu} u_i^p\right)^{(p-q)/p} \left( \E_{\mu} v_j^p\right)^{(p-q)/p} \\
    \end{align*}
    Note that since there exists a coupling where $\mu'$ and $(\mu')^{\otimes}$ are equal on $\widetilde{u}_i, \widetilde{v}_j$ except with probability $\tv \Paren{ \mu'_{\{\widetilde{u_i}, \widetilde{v_j}\}}, \mu'_{\{\widetilde{u_i}\} \otimes \{ \widetilde{v_j} \}}}$, we have that
    \begin{align*}
        &\left\vert \E_{\{\mu'\}} \sum_{i \in [n], j \in [m]} \pE_{\mu'} f_{ij}(\widetilde{u}_i, \widetilde{v}_j) -  \sum_{i \in [n], j \in [m]} \pE_{(\mu')^{\otimes}} f_{ij}(\widetilde{u}_i, \widetilde{v}_j) \right\vert \\
        &\quad\leq \left\vert \E_{\{\mu'\}} \sum_{i \in [n], j \in [m]} \tv \Paren{ \mu'_{\{\widetilde{u_i}, \widetilde{v_j}\}}, \mu'_{\{\widetilde{u_i}\} \otimes \{ \widetilde{v_j} \}}} \cdot \vert A_{ij}^q \vert \cdot C^{2(p-q)} \cdot \left( \E_{\mu} u_i^p\right)^{(p-q)/p} \left( \E_{\mu} v_j^p\right)^{(p-q)/p}\right\vert \\
        &\quad\leq C^{2(p-q)} \left\vert \E_{\{\mu'\}} \sum_{i \in [n]} \left( \E_{\mu'} u_i^p\right)^{(p-q)/p} \sum_{j \in [m]} \tv \Paren{ \mu'_{\{\widetilde{u_i}, \widetilde{v_j}\}}, \mu'_{\{\widetilde{u_i}\} \otimes \{ \widetilde{v_j} \}}} \cdot \vert A_{ij}^q \vert \cdot \left( \E_{\mu} v_j^p\right)^{(p-q)/p} \right\vert\,.
    \end{align*}
    Applying H\"older with $q/p$ and $(p-q)/p$ we have that
    \begin{align*}
        &\left\vert \E_{\{\mu'\}} \sum_{i \in [n], j \in [m]} \pE_{\mu} f_{ij}(\widetilde{u}_i, \widetilde{v}_j) -  \sum_{i \in [n], j \in [m]} \pE_{\mu^{\otimes}} f_{ij}(\widetilde{u}_i, \widetilde{v}_j) \right\vert \\
        &\quad\leq C^{2(p-q)} \left(\E_{\{\mu'\}} \sum_{i \in [n]}  \E_{\mu} u_i^p \right)^{\frac{p-q}{p}}   \left( \E_{\{\mu'\}}\sum_{i \in [n]} \left(\sum_{j \in [m]} \tv \Paren{ \mu'_{\{\widetilde{u_i}, \widetilde{v_j}\}}, \mu'_{\{\widetilde{u_i}\} \otimes \{ \widetilde{v_j} \}}}   \vert A_{ij}^q \vert  \left( \E_{\mu} v_j^p\right)^{\frac{p-q}{p}}\right)^{\frac{p}{q}}\right)^{\frac{q}{p}}\,.
    \end{align*}
    We now focus on bounding $\E_{\{\mu'\}} \left(\sum_{j \in [m]} \tv \Paren{ \mu'_{\{\widetilde{u_i}, \widetilde{v_j}\}}, \mu'_{\{\widetilde{u_i}\} \otimes \{ \widetilde{v_j} \}}} \cdot \vert A_{ij}^q \vert \cdot \left( \E_{\mu} v_j^p\right)^{\frac{p-q}{p}}\right)^{\frac{p}{q}}$. 
    
    Note that by once again applying H\"older with $q/p$ and $(p-q)/p$ we have that 
    \begin{align*}
        & \E_{\{\mu'\}} \left(\sum_{j \in [m]} \tv \Paren{ \mu'_{\{\widetilde{u_i}, \widetilde{v_j}\}}, \mu'_{\{\widetilde{u_i}\} \otimes \{ \widetilde{v_j} \}}} \cdot \vert A_{ij}^q \vert \cdot \left( \E_{\mu} v_j^p\right)^{(p-q)/p}\right)^{p/q} \\
        &\quad\leq \E_{\{\mu'\}} \left(\sum_{j \in [m]} A_{ij}^p\right) \left( \sum_{j \in [m]} \tv \Paren{ \mu'_{\{\widetilde{u_i}, \widetilde{v_j}\}}, \mu'_{\{\widetilde{u_i}\} \otimes \{ \widetilde{v_j} \}}}^{p/(p-q)} \left( \E_{\mu} v_j^p\right) \right)^{(p-q)/q}
    \end{align*}
    However, by our assumption, we have that this is at most $\epsilon \cdot \left(\sum_{i \in [n], j \in [m]} A_{ij}^p \right)^{\frac{p+q}{2q}} $ and combining this all we have that
    \begin{align*}
        \left\vert \E_{\{\mu'\}} \sum_{i \in [n], j \in [m]} \pE_{\mu} f_{ij}(\widetilde{u}_i, \widetilde{v}_j) -  \sum_{i \in [n], j \in [m]} \pE_{\mu^{\otimes}} f_{ij}(\widetilde{u}_i, \widetilde{v}_j) \right\vert &\leq C^{2(p-q)} \left( \sum_{i \in [n]}  \E_{\mu} u_i^p \right)^{\frac{p-q}{p}}  \left(\epsilon \left(\sum_{i \in [n], j \in [m]} A_{ij}^p \right)^{\frac{p+q}{2q}}\right)^{\frac{q}{p}}\,.
    \end{align*}
    Furthermore, we can bound the second and third term in~\cref{eq:split-tilde-lra} via~\cref{lem:poly-tilde-vs-not}, which completes the proof.
\end{proof}

We now prove~\cref{lem:poly-tilde-vs-not}.
\begin{proof}[Proof of~\cref{lem:poly-tilde-vs-not}]
    We will bound 
    \[ \left\vert \E_{\{\mu'\}} \sum_{i \in [n], j \in [m]} \pE_{\mu'} \left[ \left( A_{ij}^q \left(\widetilde{u}_i \widetilde{v}_j\right)^{p-q} - A_{ij}^q \left(u_i v_j\right)^{p-q} \right) \right] \right\vert\,,\] 
    and the expression with $(\mu')^{\otimes}$ can be bounded via the same sequence of inequalities. We break up this term as follows:
    \begin{equation}
    \label{eqn:introduce-indicators-for-E-poly-lra}
    \begin{split}
        &\left\vert \E_{\{\mu'\}} \sum_{i \in [n], j \in [m]} \pE_{\mu'} \left[ \left( A_{ij}^q \left(\widetilde{u}_i \widetilde{v}_j\right)^{p-q} - A_{ij}^q \left(u_i v_j\right)^{p-q} \right) \right] \right\vert \\
        &\leq \underbrace{ \left\vert \E_{\{\mu'\}} \sum_{i \in [n], j \in [m]} \pE_{\mu'} \left[ \mathbb{1}[E_i^u \wedge E_j^v]\left( A_{ij}^q \left(\widetilde{u}_i \widetilde{v}_j\right)^{p-q} - A_{ij}^q \left(u_i v_j\right)^{p-q}\right) \right] \right\vert}_{\eqref{eqn:introduce-indicators-for-E-poly-lra}.(1) } \\
        &\hspace{0.2pt}\quad+ \underbrace{ \left\vert \E_{\{\mu'\}} \sum_{i \in [n], j \in [m]} \pE_{\mu'} \left[ \mathbb{1}[\overline{E}_i^u \vee \overline{E}_j^v]\left( A_{ij}^q \left(\widetilde{u}_i \widetilde{v}_j\right)^{p-q} - A_{ij}^q \left(u_i v_j\right)^{p-q} \right) \right] \right\vert}_{\eqref{eqn:introduce-indicators-for-E-poly-lra}.(2)} \,.
    \end{split}
    \end{equation}
    Note that when $E_i^u$ and $E_j^v$ both occur then $\widetilde{u}_i = u_i$ and $\widetilde{v}_j = v_j$, so thus $\eqref{eqn:introduce-indicators-for-E-poly-lra}.(1) =0$, and it remains to bound term \eqref{eqn:introduce-indicators-for-E-poly-lra}.(2). Note that via Triangle Inequality we have that this term is at most
    \begin{equation}
    \label{eqn:tail-events-lra}
    \begin{split}
        \E_{\{\mu'\}}  & \sum_{i \in [n], j \in [m]} \pE_{\mu'} \left[ \mathbb{1}[\overline{E}_i^u \vee \overline{E}_j^v]\left( \vert A_{ij} \vert ^q \abs{\widetilde{u}_i \widetilde{v}_j}^{p-q} - \vert A_{ij} \vert^q \abs{u_i v_j}^{p-q} \right) \right] \\
        & \leq 2 \underbrace{ \E_{\{\mu'\}} \sum_{i \in [n], j \in [m]} \pE_{\mu'} \left[ \mathbb{1}[\overline{E}_i^u \vee \overline{E}_j^v]\left( \vert A_{ij} \vert^q  \abs{u_i v_j}^{p-q} \right) \right] }_{\eqref{eqn:tail-events-lra}.(1)} \,,
    \end{split}
    \end{equation}
    since $\widetilde{u}_i, \widetilde{v}_j$ always contract values towards $0$ compared to $u_i, v_j$.
    Furthermore, by H\"older we have that
    \begin{equation*}
        \begin{split}
            \eqref{eqn:tail-events-lra}.(1) & \leq \left(\E_{\{\mu'\}} \sum_{i \in [n], j \in [m]} A_{ij}^p \pE_{\mu'} \mathbb{1}[\overline{E}_i^u \vee \overline{E}_j^v] \right)^{\frac{q}{p}} \left(\E_{\{\mu'\}} \sum_{i \in [n], j \in [m]} \pE_{\mu'} \left(u_i v_j\right)^p \right)^{\frac{p-q}{p}}\\
            & \leq \left(\sum_{i \in [n], j \in [m]} A_{ij}^p \pE_{\mu} \mathbb{1}[\overline{E}_i^u \vee \overline{E}_j^v] \right)^{q/p} \left( \sum_{i \in [n], j \in [m]} \pE_{\mu} \left(u_i v_j\right)^p \right)^{(p-q)/p}\, ,
        \end{split}
    \end{equation*}
    where the last inequality follows since  $\E_{\{\mu'\}} \pE_{\mu'} = \pE_{\mu}$.
    Note that we have that $\mathbb{1}[\overline{E}_i^u \vee \overline{E}_j^v] \leq \mathbb{1}[\overline{E}_i^u] + \mathbb{1}[\overline{E}_j^v]$ and for all $i,j$ we have that $\Pr[\overline{E}_i^u], \Pr[\overline{E}_j^v] \leq 1/C^p$. Thus, we have that
    \[ \eqref{eqn:tail-events-lra}.(1) \leq 1/C^q \cdot \left(\sum_{i \in [n], j \in [m]} A_{ij}^p \right)^{q/p} \left( \sum_{i \in [n], j \in [m]} \pE_{\mu} \left(u_i v_j\right)^p \right)^{(p-q)/p} \,.\]

\end{proof}

Now we are finally able to use~\cref{lem:polynomial-rounding-error} to show that $\sum_{i \in [n], j \in [m]} A_{ij}^q \left(u_i v_j\right)^{p-q}$ is \emph{potential aligned}.

\begin{proof}[Proof of~\cref{lem:poly-potential-aligned}]
    Note that by~\cref{lem:polynomial-rounding-error} (setting $C = 1/\poly(\epsilon)$) and using our bounds on $\sum_{i \in [n]} \E_{\mu} u_i^p, \sum_{i \in [n], j \in [m]} \left(\E_{\mu} u_i^p\right) \cdot \left(\E_{\mu} v_j^p\right)$ from the lemma assumptions we know that if the difference in expectation under $\mu$ versus $\mu^{\otimes}$ is large then we must have that
    \[ \E_{\{\mu'\}} \sum_{i \in [n]} \left(\sum_{j \in [m]} A_{ij}^p\right) \left( \sum_{j \in [m]} \tv \Paren{ \mu_{\{\widetilde{u_i}, \widetilde{v_j}\}}, \mu_{\{\widetilde{u_i}\} \otimes \{ \widetilde{v_j} \}}}^{\frac{p}{p-q} } \left( \E_{\mu} v_j^p\right) \right)^{\frac{p-q}{q}} \geq \epsilon \left(\sum_{i \in [n], j \in [m]} A_{ij}^p \right)^{\frac{p+q}{2q}}\,. \]
    Consider a distribution $\psi$ which samples $i$ with probability proportional to $\left(\sum_{j \in [m]} A_{ij}^p\right)$. Then we have by rewriting the above inequality that 
    \[ \E_{\{\mu'\}} \left(\sum_{i \in [n], j \in [m]} A_{ij}^p\right) \left[\E_{i \sim \psi} \left( \sum_{j \in [m]} \tv \Paren{ \mu'_{\{\widetilde{u_i}, \widetilde{v_j}\}}, \mu'_{\{\widetilde{u_i}\} \otimes \{ \widetilde{v_j} \}}}^{\frac{p}{p-q} } \left( \E_{\mu} v_j^p\right) \right)^{\frac{p-q}{p}}\right] \geq \epsilon \left(\sum_{i \in [n], j \in [m]} A_{ij}^p \right)^{\frac{p+q}{2q}} \,. \]
    Furthermore, by Pinskers, we know that
    \[ \E_{\{\mu'\}} \left[\E_{i \sim \psi} \left( \sum_{j \in [m]} \left( H(\widetilde{v}_j) - H(\widetilde{v}_j \vert \widetilde{u}_i)\right)^{p/(2(p-q))} \left( \E_{\mu} v_j^p\right) \right)^{(p-q)/q}\right] \geq \epsilon \left(\sum_{i \in [n], j \in [m]} A_{ij}^p \right)^{(p-q)/(2q)}\,.\]
    Note that by rearranging the inner sum and applying H\"older on the inner sum over $j$ (with $p/(2(p-q))$ and $(p-2q)/(2(p-q))$) we have that %
    \begin{align*}
        &\left( \sum_{j \in [m]} \left( H(\widetilde{v}_j) - H(\widetilde{v}_j \vert \widetilde{u}_i)\right)^{p/(2(p-q))} \left( \E_{\mu} v_j^p\right) \right)^{(p-q)/q} \\
        &\leq \left( \sum_{j \in [m]} \left( H(\widetilde{v}_j) - H(\widetilde{v}_j \vert \widetilde{u}_i)\right) \left(\E_{\mu} v_j^p\right) \right)^{p/(2q)} \left(\sum_{j \in [m]} \E_{\mu} v_j^p\right)^{(p-2q)/(2q)}
    \end{align*}
    Thus, we have concluded that
    \begin{align*}
        \epsilon \left(\sum_{i \in [n], j \in [m]} A_{ij}^p \right)^{\frac{p-q}{2q}} \leq \E_{\{\mu'\}} \left[\E_{i \sim \psi} \left( \sum_{j \in [m]} \left( H(\widetilde{v}_j) - H(\widetilde{v}_j \vert \widetilde{u}_i)\right) \left(\E_{\mu} v_j^p\right) \right)^{\frac{p}{2q}}\right] \cdot \left(\sum_{j \in [m]} \E_{\mu} v_j^p\right)^{\frac{p-2q}{2q}}
    \end{align*}
    We can now observe that the decrease in potential is at most the maximum value of the potential, or $\Phi_{\max}$. Thus, we can also conclude that
        \begin{align*}
        \epsilon \left(\sum_{i \in [n], j \in [m]} A_{ij}^p \right)^{\frac{p-q}{2q}} \leq \E_{\{\mu'\}} \left[\E_{i \sim \psi} \left( \sum_{j \in [m]} \left( H(\widetilde{v}_j) - H(\widetilde{v}_j \vert \widetilde{u}_i)\right) \left(\E_{\mu} v_j^p\right) \right)\right] \cdot \Phi_{\max}^{\frac{p-2q}{2q}} \cdot \left(\sum_{j \in [m]} \E_{\mu} v_j^p\right)^{\frac{p-2q}{2q}}
    \end{align*}
    Note that since
    \[ \Phi_{\max} \leq \log n \cdot  O_p\left(  \sum_{i \in [n], j \in [m]} A_{ij}^p\right)^{1/2}\]
    and 
    \[ \sum_{j \in [m]} \E_{\mu} v_j^p \leq O_p\left( \sum_{i \in [n], j \in [m]} A_{ij}^p\right)^{1/2} \]
    we have that
    \[ \E_{\{\mu'\}} \left[\E_{i \sim \psi} \left( \sum_{j \in [m]} \left( H(\widetilde{v}_j) - H(\widetilde{v}_j \vert \widetilde{u}_i)\right) \left(\E_{\mu} v_j^p\right) \right)\right] \geq \Omega_p(\epsilon/\log n) \left( \sum_{i \in [n], j \in [m]} A_{ij}^p\right)^{1/2}\,,\]
    and thus our potential function decreases by $\Omega_p(\epsilon) \left( \sum_{i \in [n], j \in [m]} A_{ij}^p\right)^{1/2}$ in expectation over the choice of $i$ and thus there is a fixed $i$ which also achieves this decrease. 
\end{proof}

\subsubsection{Rounding Low Degree Polynomial Growth Functions}

\begin{lemma}
\label{lem:low-deg-poly-aligned}
    Let $0 < \epsilon \leq 1$ and let $\mu$ be some pseudo-distribution. Let $p, q \in \mathbb{N}^{>0}$ such that $q \geq p/2$ and let $\{A_{ij}\}_{i \in [n], j \in [m]} \in \mathbb{R}$ be a collection of real coefficients. Furthermore, let $\mu$ satisfy:
    \[\sum_{i \in [n]} \pE_{\mu} u_i^p \leq \sqrt{\sum_{i \in [n], j \in [m]} A_{ij}^p}
    ,.\]
    Suppose $\mu'$ is a random psuedodistribution such that $\E_{\{\mu'\}} \mu' = \mu$ and everything in the support of $\mu'$ also satisfy 
    \[\sum_{i \in [n]} \pE_{\mu'} u_i^p \leq \sqrt{\sum_{i \in [n], j \in [m]} A_{ij}^p}
    ,.\]
    Then $\sum_{i \in [n], j \in [m]} A_{ij}^q (u_i v_j)^{p-q}$ is 
    \[(\epsilon \sum_{i \in [n], j \in [m]} A_{ij}^p, \Omega_p(\poly_{p,q} (\epsilon)) \left(\sum_{i \in [n], j \in [m]} A_{ij}^p\right)^{1/2})\text{-potential aligned} \]
    for $\sum_{j \in [m]} \Var(v_j^{p-q})^{p/(2(p-q))}$ and $\mu$.
\end{lemma}

In the proof of~\cref{lem:low-deg-poly-aligned} we will need the following fact:
\begin{fact}
\label{fact:decr-of-powers}
    For all $c\geq 1$ and positive numbers $0 <= a' <= a$ we have that
    \[ (a-a')^c \leq a^c - (a')^c\,.\]
\end{fact}
\begin{proof}[Proof of~\cref{lem:low-deg-poly-aligned}]
    If we have that 
    \[\left\vert \E_{\{\mu'\}} \sum_{i \in [n], j \in [m]} \pE_{\mu'} A_{ij}^q (u_i v_j)^{p-q} - \sum_{i \in [n], j \in [m]} \pE_{(\mu')^\otimes} A_{ij}^q (u_i v_j)^{p-q} \right\vert \geq \epsilon \sum_{i \in [n], j \in [m]} A_{ij}^p\,,\]
    then it also holds that
    \[ \sum_{i \in [n], j \in [m]} \vert A_{ij} \vert^q \E_{\{\mu'\}} \Cov_{\mu'}(u_i^{p-q}, v_j^{p-q}) \geq \epsilon \sum_{i \in [n], j \in [m]} A_{ij}^p\,.\]
    Applying~\cref{fact:var-reduction} we have that
    \[ \sum_{i \in [n], j \in [m]} \E_{\{\mu'\}} \vert A_{ij} \vert^q \sqrt{\Paren{\Var_{\mu'} (v_j^{p-q}) - \E_{u_i \sim \mu'} \Var_{\mu'} (v_j^{p-q} \vert u_i^{p-q})} \cdot \Paren{\Var_{\mu'} (u_i^{p-q})}} \geq \epsilon \sum_{i \in [n], j \in [m]} A_{ij}^p\,. \]
    We can further bound the LHS of the inequality via applying H\"older with $(p-q)/p$ and $q/p$: 
    \begin{align*}
        &\sum_{i \in [n], j \in [m]} \E_{\{\mu'\}} \vert A_{ij} \vert^q \sqrt{\left(\Var_{\mu'} (v_j^{p-q}) - \E_{u_i \sim \mu'} \Var_{\mu'}(v_j^{p-q} \vert u_i^{p-q})\right) \cdot \left(\Var_{\mu'} (u_i^{p-q})\right)} \\
        &\qquad= \sum_{i \in [n]} \E_{\{\mu'\}} \sqrt{ {\Var_{\mu'}(u_i^{p-q})}} \sum_{j \in [m]} \vert A_{ij} \vert^q \sqrt{\Var_{\mu'}(v_j^{p-q}) - \E_{u_i \sim \mu'} \Var_{\mu'}(v_j^{p-q} \vert u_i^{p-q})} \\
        &\qquad\leq \left( \sum_{i \in [n]} \E_{\{\mu'\}} \Var_{\mu'}(u_i^{p-q})^{p/(2(p-q))} \right)^{(p-q)/p}  \\
        &\qquad\qquad\left( \sum_{i \in [n]} \E_{\{\mu'\}} \left( \sum_{j \in [m]} \vert A_{ij} \vert^q \sqrt{\Var_{\mu'}(v_j^{p-q}) - \E_{u_i \sim \mu'} \Var_{\mu'}(v_j^{p-q})}\right)^{p/q}\right)^{q/p} \,.
    \end{align*}
    Note that by constraints and Jensens we have that $\Var_{\mu'}(u_i^{p-q})^{p/(2(p-q))} \leq \pE_{\mu'} u_i^p$ and thus
    \[ \sum_{i \in [n]} \E_{\{\mu'\}} \Var_{\mu'}(u_i^{p-q})^{p/(2(p-q))} \leq \sum_{i \in [n]} \pE_{\mu'} u_i^p \leq O_p\left(\sqrt{\sum_{i \in [n], j \in [m]} A_{ij}^p} \right)\,.\]
    Furthermore, by applying H\"older again with $(p-q)/p$ and $q/p$ we have that 
    \begin{align*}
        &O_p(1) \cdot \left( \sum_{j \in [m]} \vert A_{ij} \vert^q \sqrt{\Var_{\mu'}(v_j^{p-q}) - \E_{u_i \sim \mu'} \Var_{\mu'}(v_j^{p-q})}\right)^{p/q} \\
        &\qquad\leq \left(\sum_{j \in [m]} A_{ij}^p\right) \left( \sum_{j \in [m]} \left( \Var_{\mu'}(v_j^{p-q}) - \E_{u_i \sim \mu'} \Var_{\mu'}(v_j^{p-q}) \right)^{p/(2(p-q))}\right)^{(p-q)/q}\,.
    \end{align*}
    Combining all of these steps, we have that
    \begin{align*}
        \begin{aligned}
        \epsilon \sum_{i \in [n], j \in [m]} A_{ij}^p &\leq O_p\left( 1 \right)\left( \sum_{i \in [n], j \in [m]} A_{ij}^p\right)^{\frac {p-q} {2p}} \\
        &\qquad \left(\sum_{i \in [n]} \E_{\{\mu'\}} \left(\sum_{j \in [m]} A_{ij}^p\right) \left( \sum_{j \in [m]} \left( \Var_{\mu'}(v_j^{p-q}) - \E_{u_i \sim \mu'} \Var_{\mu'}(v_j^{p-q}) \right)^{\frac p {2(p-q)}}\right)^{\frac {p-q} q} \right)^{\frac q p}\,.
        \end{aligned}
    \end{align*}
    Equivalently, if we define $\psi$ as the distribution which samples $i$ proportional to $\sum_{j \in [m]} A_{ij}^p$ we have that
    \begin{equation}
    \label{eqn:var-bound-lra-start}
    \begin{split}
        \epsilon \sum_{i \in [n], j \in [m]} A_{ij}^p &\leq O_p\left( 1 \right) \left( \sum_{i \in [n], j \in [m]} A_{ij}^p\right)^{\frac {p+q} {2p}} \cdot \\
        & \qquad \underbrace{ \left(\E_{i \sim \psi} \E_{\{\mu'\}} \left( \sum_{j \in [m]} \left( \Var_{\mu'}(v_j^{p-q}) - \E_{u_i \sim \mu'} \Var_{\mu'}(v_j^{p-q}) \right)^{\frac p {2(p-q)}}\right)^{\frac {p-q} {q}} \right)^{\frac q p} }_{\eqref{eqn:var-bound-lra-start}.(1)}\,.
    \end{split}
    \end{equation}
    Applying~\cref{fact:decr-of-powers} and Jensens (since $p/(2(p-q)) \geq 1$) we have that
    \begin{equation*}
    \begin{split}
        \eqref{eqn:var-bound-lra-start}.(1) & \leq    \left(\E_{i \sim \psi} \E_{\{\mu'\}} \left( \sum_{j \in [m]} \left( \Var_{\mu'}(v_j^{p-q}) - \E_{u_i \sim \mu'} \Var_{\mu'}(v_j^{p-q}) \right)^{\frac p {2(p-q)}}\right)^{\frac {p-q} {q}} \right)^{\frac q p}  \\
        &\leq   \left(\E_{i \sim \psi} \E_{\{\mu'\}} \left( \sum_{j \in [m]} \Var_{\mu'}(v_j^{p-q})^{\frac p {2(p-q)}} - \left(\E_{u_i \sim \mu'} \Var_{\mu'}(v_j^{p-q}) \right)^{\frac p {2(p-q)}}\right)^{\frac {p-q} {q}} \right)^{\frac q p}  \\
        &\leq    \left(\E_{i \sim \psi} \E_{\{\mu'\}} \left( \sum_{j \in [m]} \Var_{\mu'}(v_j^{p-q})^{\frac p {2(p-q)}} - \E_{u_i \sim \mu'}  \left(\Var_{\mu'}(v_j^{p-q}) \right)^{\frac p {2(p-q)}}\right)^{\frac {p-q} {q}} \right)^{\frac q p} \\
        &\leq  \left(\E_{i \sim \psi} \E_{\{\mu'\}}  \sum_{j \in [m]} \Var_{\mu'}(v_j^{p-q})^{\frac p {2(p-q)}} - \E_{u_i \sim \mu'}  \left(\Var_{\mu'}(v_j^{p-q}) \right)^{\frac p {2(p-q)}} \right)^{\frac {p-q} p} \,. 
    \end{split}
    \end{equation*}
    
    Thus, we have concluded that the decrease in potential is at least
    \[ \Omega_p\left( \epsilon  \left( \sum_{i \in [n], j \in [m]} A_{ij}^p \right)^{\frac {p-q}{2p}} \right)^{\frac p {p-q}} = \Omega_p(\poly_{p,q} (\epsilon)) \left(\sum_{i \in [n], j \in [m]} A_{ij}^p\right)^{1/2}\,.  \]
\end{proof}

\subsection{Entrywise Low Rank Approximation: Algorithm and Analysis}

We can now proceed to the algorithm and proof of~\cref{thm:lra}

\begin{mdframed}
  \begin{algorithm}[$\ell_p$ LRA Additive Approximation]
    \label{algo:lra}\mbox{}
    \begin{description}
    \item[Input:] Matrix $A \in \mathbb{R}^{n \times m}$ and target accuracy $0<\epsilon<1$.
    
    \item[Operations:]\mbox{}
    \begin{enumerate}
        \item Let $\calS$ be a $\frac{\epsilon^p}{n^3} \left(\sum_{i \in [n], j \in [m]} A_{ij}^p\right)^{1/2}$ grid of \[\left[-2 \left(\sum_{i \in [n], j \in [m]} A_{ij}^p\right)^{1/2}, 2 \left(\sum_{i \in [n], j \in [m]} A_{ij}^p\right)^{1/2}\right]\,.\]
        \item For each $s_u \in \calS, s_v \in \calS$ and set $\calT$ of size at most $O_p\left(\log^2 (n)  /\poly_p(\epsilon) \right)$:
        \begin{enumerate}
            \item Let $\Sigma$ be a $\frac{\epsilon^p}{n^2} \cdot \left(\sum_{i \in [n], j \in [m]} A_{ij}^p\right)^{1/(2p)}$ grid of 
            \[\left[-\left(\sum_{i \in [n], j \in [m]} A_{ij}^p\right)^{1/(2p)}, \left(\sum_{i \in [n], j \in [m]} A_{ij}^p\right)^{1/(2p)}\right]\,.\]
            \item Let $\mu$ be a $O_p\left(\log^2 (n)  /\poly_p(\epsilon) \right)$-degree Sherali-Adams pseudo-distribution over $\left(\Sigma\right)^{n+m}$ such that $\pE_{\mu}$ optimizes 
            \begin{align}
                &\min_{\pE_{\mu}} \quad \pE_{\mu} \sum_{i\in[n],j\in[m]} (A_{ij} - u_i v_j)^p \notag \\
                &\text{s.t.\ } \forall \text{ monomials } r, \notag \\
                &\quad \sum_{i \in [n]} \pE_{\mu} r \cdot u_i^p = \pE_{\mu} r \cdot s_u \label{eq:constraint1} \\
                &\quad \sum_{j \in [m]} \pE_{\mu} r \cdot v_j^p = \pE_{\mu} r \cdot s_v \label{eq:constraint2}
            \end{align}
            \item Let $\hat{x}_\calT$ be a draw from the local distribution $\{x_{\calT}\}$. Let $\mu_\calT$ be the pseudodistribution obtained by conditioning on $\{u_i = \hat{u}_i\}_{i \in \calT}$.
            \item Let $\hat{u}, \hat{v}$ be the solution where $\hat{u}_i, \hat{v}_j$ are sampled independently from their local distributions in $\mu_\calT$.
        \end{enumerate}
    \end{enumerate}
    \item[Output:] The embedding $\hat{u}, \hat{v}$ with the lowest LRA objective value over all iterations of the loop.
    \end{description}
  \end{algorithm}
\end{mdframed}

We will first argue that our discretization of a subset of $\mathbb{R}$ and optimizing over this set with the above constraints approximately preserves the value of the objective for at least one iteration of the loop. We will then argue that for some set $\calT$, the rounding procedure also produces a solution which in expectation has value close to the value of the LP.

Note that the analysis below is stated for the case of $m=n$ but easily generalizes to the case of unequal $m,n$. %

\paragraph{Value of the LP under discretization and constraints.}

Before we show that at least one of the LPs in the algorithm has good objective value, we will first need to give a characterization of optimal solutions to~\cref{prob:lra}.
\begin{fact}
\label{fact:lra-sol-char}
    Let $p$ be an even integer. For any matrix $A \in \mathbb{R}^{n\times n}$ there exists an optimal solution $\{u_i\}_{i \in [n]}$ and $\{v_j\}_{j \in [m]}$ to~\cref{prob:lra} such that:
    \begin{enumerate}
        \item $\sum_{i \in [n]} u_i^p = \sum_{j \in [m]} v_j^p \leq 2^{p-1} \sqrt{\sum_{i, j \in [m]} A_{ij}^p}$.
        \item $\sum_{i \in [n], j \in [m]} u_i^p v_j^p \leq 2^{p} \sum_{i, j \in [m]} A_{ij}^p$.
    \end{enumerate}
\end{fact}
\begin{proof}
    Consider an optimal solution $\{u_i\}_{i \in [n]}$ and $\{v_j\}_{j \in [m]}$. Note that we have by Almost Triangle Inequality that
    \[ \sum_{i \in [n], j \in [m]} u_i^p v_j^p = \sum_{i \in [n], j \in [m]} (u_iv_j \pm A_{ij})^p \leq 2^{p-1} \left( \OPT_{\mathrm{LRA}} + \sum_{i \in [n], j \in [m]} A_{ij}^p \right)\,.\]
    However, $\OPT_{\mathrm{LRA}} \leq \sum_{i \in [n], j \in [m]} A_{ij}^p$ since the solution which sets $u_i = v_j = 0$ for all $i,j$ achieves this value. Thus, we have that 
    \[ \sum_{i \in [n], j \in [m]} u_i^p v_j^p \leq 2^{p} \sum_{i, j \in [m]} A_{ij}^p \,,\]
    as desired. Furthermore,
    \[ \sum_{i \in [n], j \in [m]} u_i^p v_j^p = \left(\sum_{i \in [n]} u_i^p \right) \left( \sum_{j \in [m]} v_j^p\right)\,,\]
    and since the value of the solution is invariant to dividing all $u_i$ by some constant $c$ and multiplying $v_j$ by the same constant we can always scale $u_i, v_j$ to get a new optimal solution $u_i', v_j'$ such that 
    \[ \sum_{i \in [n]} (u_i')^p = \sum_{j \in [m]} (v_j')^p\,,\]
    which completes the proof.
\end{proof}

We now argue that our discretized LRA program has value which is comparable to the original instance in at least one iteration of the loop.
\begin{lemma}
\label{lem:lra-discretized}
    There exists $\Sigma$, $\calS$ such that
    \[ \abs{\Sigma} \leq O_p(n^2/\epsilon^p), \abs{\calS} \leq O_p(n/\epsilon^p)\,,\]
    and there exists $\{\hat{u}_i\}_{i \in [n]}, \{\hat{v}_j\}_{j \in [m]} \in \Sigma$ such that
    \[ \sum_{i \in [n], j \in [m]} (A_{ij} - \hat{u}_i \hat{v}_j)^p \leq \OPT_{\mathrm{LRA}} + \epsilon \left( \sum_{i \in [n], j \in [m]} A_{ij}^p\right)\,,\]
    and 
    \begin{align*}
        \sum_{i \in [n]} \hat{u}_i^p, \sum_{j \in [m]} \hat{v}_j^p \in \calS\,.
    \end{align*}
\end{lemma}
\begin{proof}
    First, we begin by noting that there exists an optimal solution which satisfies
        \begin{enumerate}
        \item $\sum_{i \in [n]} u_i^p = \sum_{j \in [m]} v_j^p \leq 2^{p-1} \sqrt{\sum_{i, j \in [m]} A_{ij}^p}$.
        \item $\sum_{i \in [n], j \in [m]} u_i^p v_j^p \leq 2^{p} \sum_{i, j \in [m]} A_{ij}^p$,
    \end{enumerate}
    by~\cref{fact:lra-sol-char}. As a result, we also have that in this solution all $u_i, v_j$ have magnitude at most $\left(\sum_{i \in [n], j \in [m]} A_{ij}^p\right)^{1/(2p)}$. Let $\hat{u}_i, \hat{v}_j$ be some small perturbation of $u_i, v_j$. Whenever 
    \[ \vert \hat{u}_i - u_i \vert , \vert \hat{v}_j - v_j \vert \leq O_p\left(\frac{\epsilon^p}{n^2} \right) \left(\sum_{i \in [n], j \in [m]} A_{ij}^p\right)^{1/2p}\,\]
    and $\vert u_i \vert, \vert v_j \vert \leq \left(\sum_{i \in [n], j \in [m]} A_{ij}^p\right)^{1/2p}$ we have that
    \[ \vert \hat{u}_i \hat{v}_j - u_i v_j \vert \leq O_p\left(\frac{\epsilon^p}{n^2} \right) \left(\sum_{i \in [n], j \in [m]} A_{ij}^p\right)^{1/p}\,.\]
    Therefore, for such $\hat{u}_i, \hat{v}_j$ we have that
    \begin{align*}
        & \sum_{i \in [n], j \in [m]} (A_{ij} - \hat{u}_i \hat{v}_j)^p = \sum_{i \in [n], j \in [m]} (A_{ij} - \hat{u}_i \hat{v}_j \pm u_iv_j)^p \\
        &\leq \sum_{i \in [n], j \in [m]} (A_{ij} - u_i v_j)^p + \sum_{i \in [n], j \in [m]} \sum_{0 \leq k < p} (A_{ij} - u_iv_j)^k \left(u_i v_j - \hat{u}_i \hat{v}_j\right)^{p-k} \\
        &\leq \sum_{i \in [n], j \in [m]} (A_{ij} - u_i v_j)^p + \sum_{0 \leq k < p} \left( \sum_{i, j \in [m]} (A_{ij} - u_iv_j)^p \right)^{\frac{k}{p}} \left( \sum_{i, j \in [m]} \left(u_i v_j - \hat{u}_i \hat{v}_j\right)^{p}  \right)^{\frac{p-k}{p}}\,,
    \end{align*}
    where the last inequality is via H\"older with $k/p$ and $(p-k)/p$. Using our bound on $\vert \hat{u}_i \hat{v}_j - u_i v_j \vert$ we have that
    \[ \left( \sum_{i, j \in [m]} (A_{ij} - u_iv_j)^p \right)^{\frac{k}{p}} \left( \sum_{i, j \in [m]} \left(u_i v_j - \hat{u}_i \hat{v}_j\right)^{p}  \right)^{\frac{p-k}{p}} \leq \OPT_{\mathrm{LRA}}^{\frac{k}{p}} \cdot O_p\left( n^2 \cdot \frac{\epsilon^p}{n^2} \cdot \left(\sum_{i \in [n], j \in [m]} A_{ij}^p\right)\right)^{\frac{p-k}{p}}\,.\]
    Using that $\OPT_{\mathrm{LRA}} \leq \left(\sum_{i \in [n], j \in [m]} A_{ij}^p\right)$ and for a sufficiently small constant $O_p(1)$ we have that 
    \[ \sum_{i \in [n], j \in [m]} (A_{ij} - \hat{u}_i \hat{v}_j)^p \leq \sum_{i \in [n], j \in [m]} (A_{ij} - u_i v_j)^p + \epsilon \left( \sum_{i \in [n], j \in [m]} A_{ij}^p \right)\,.\]
    Thus, it suffices to show that there exist small enough sets $\Sigma, \calS$ such that for all $u_i, v_j$ such $\hat{u}_i, \hat{v}_j$ sufficiently close exist. First, note that allowing $\hat{u}_i, \hat{v}_j$ to belong to a $O_p\left(\frac{\epsilon^p}{n^2} \right) \left(\sum_{i \in [n], j \in [m]} A_{ij}^p\right)^{1/(2p)}$-grid of $\left[-\left(\sum_{i \in [n], j \in [m]} A_{ij}^p\right)^{1/(2p)}, \left(\sum_{i \in [n], j \in [m]} A_{ij}^p\right)^{1/(2p)}\right]$ suffices to have that for all $u_i,v_j$ there is $\hat{u}_i, \hat{v}_j$ which is $O_p\left(\frac{\epsilon^p}{n^2} \right) \left(\sum_{i \in [n], j \in [m]} A_{ij}^p\right)^{1/(2p)}$ close. Furthermore, this grid has size $O_p(n^2/\epsilon^p)$.

    We now consider the possible $\ell_p$ norms of $\hat{u}, \hat{v}$. Note that since their entries belong to the grid above they must be of the form
    \[ \left(C \cdot \frac{\epsilon^p}{n^2} \right) \left(\sum_{i \in [n], j \in [m]} A_{ij}^p\right)^{1/(2p)}\,,\]
    for some constant $C$ and thus all $\ell_p$ norms of these vectors have the form 
    \[ \left(C \cdot \frac{\epsilon^p}{n} \right) \left(\sum_{i \in [n], j \in [m]} A_{ij}^p\right)^{1/2}\,.\]
    There are at most $O_p(n/\epsilon^p)$ possible values of the norm since 
    \[ \sum_{i \in [n]} u_i^p = \sum_{j \in [m]} v_j^p \leq 2^{p-1} \sqrt{\sum_{i, j \in [m]} A_{ij}^p}\,,\]
    and thus $\norm{\hat{u}}_p^p, \norm{\hat{v}}_p^p \leq 2^{p} \sqrt{\sum_{i, j \in [m]} A_{ij}^p}$ so we can take the set of all possible norm values as $\calS$.
\end{proof}

\paragraph{Analyzing global correlation rounding with potential functions.}

Note that we can write our objective function as follows:
\[ \sum_{i, j \in [m]} \left(A_{ij} - u_i v_j \right)^p = \sum_{0 \leq q \leq p} \sum_{i \in [n], j \in [m]} \binom{p}{k} A_{ij}^q (u_i v_j)^{p-q}\]
We will aim to invoke~\cref{lem:gcr_linear_combo} to show that the rounding error is small. In~\cref{sec:poly-growth-potential} we showed that most components of the objective function are potential-aligned, so it suffices to bound their respective potential functions and analyze the remaining term which is degree $p$ in both $u,v$.
\begin{lemma}
\label{lem:lra-potentials-bounded}
    Let $0 < \epsilon \leq 1$ and let $\mu$ be some pseudo-distribution over $\Sigma$ as defined in~\cref{lem:lra-discretized} such that
    \[ \sum_{j \in [m]} \E_{\mu} v_j^p \leq 2^{p}\sqrt{\sum_{i \in [n], j \in [m]} A_{ij}^p}\,.\]
    Furthermore, let $E_i^v$ be the event that $\abs{v_j} \leq \poly_p(1/\epsilon) \left( \E_{\mu} v_j^p \right)^{1/p}$. Furthermore, let $\widetilde{v}_j = v_j E_j^v$. Then we have that 
    \[ \sum_{j \in [m]} H(\tilde{v}_j) \E_\mu v_{j}^p \leq \log \left(\poly_p(1/\epsilon) \cdot n\right) \cdot \sqrt{\sum_{i, j \in [m]} A_{ij}^p}\,.\]
    In addition, for all $q \leq p/2$ we have that
    \[ \sum_{j \in [m]} \Var(v_j^q)^{p/(2q)} \leq \sqrt{\sum_{i \in [n], j \in [m]} A_{ij}^p}\,.\]
\end{lemma}
\begin{proof}
    Note that we have that 
    \[ \sum_{j \in [m]} H(\tilde{v}_j) \E_\mu v_{j}^p \leq \left( \max_{j \in [m]} H(\tilde{v}_j)\right) \sum_{j \in [m]} \E_\mu v_{j}^p \leq \left( \max_{j \in [m]} H(\tilde{v}_j)\right) \sqrt{\sum_{i, j \in [m]} A_{ij}^p}\,.\]
    Thus, it suffices to bound $\max_{j \in [m]} H(\tilde{v}_j)$, which reduces to bounding the maximum alphabet size of $\tilde{v}_j$. Note that $\tilde{v}_j$ has magnitude at most $\poly_p(1/\epsilon) \cdot \left(E_\mu v_j^p\right)^{1/p}$ and takes value on a grid of size $O_p\left(\frac{\epsilon^p}{n^2} \right) \left(\sum_{i \in [n], j \in [m]} A_{ij}^p\right)^{1/2p}$. Thus, it has at most
    \[ n^2/\epsilon^p \cdot \frac{\poly_p(1/\epsilon) \cdot \left(E_\mu v_j^p\right)^{1/p}}{\left(\sum_{i \in [n], j \in [m]} A_{ij}^p\right)^{1/2p}}\]
    possible values. Furthermore, since 
    \[ \sum_{j \in [m]} \E_{\mu} v_j^p \leq 2^{p}\sqrt{\sum_{i \in [n], j \in [m]} A_{ij}^p}\,.\]
    we also have that
    \[ \max_{j \in [m]} \E_{\mu} v_j^p \leq n \cdot 2^{p}\sqrt{\sum_{i \in [n], j \in [m]} A_{ij}^p}\,,\]
    and thus the largest possible alphabet size of $\widetilde{v}_j$ is $O_p(\poly_p(1/\epsilon) \cdot n^3)$. We can therefore bound 
    \[ \sum_{j \in [m]} H(\tilde{v}_j) \E_\mu v_{j}^p  \leq \left( \max_{j \in [m]} H(\tilde{v}_j)\right) \sqrt{\sum_{i, j \in [m]} A_{ij}^p} \leq \log \left(n / \epsilon\right) \cdot \sqrt{\sum_{i, j \in [m]} A_{ij}^p}\,.\]
    Furthermore, we have by Jensens that
    \[\sum_{j \in [m]} \Var(v_j^q)^{p/(2q)} \leq \sum_{j \in [m]} \E_{\mu} v_j^p \leq \sqrt{\sum_{i, j \in [m]} A_{ij}^p}\,,\]
    which completes the proof.
\end{proof}

We can now show that the rounding error is small after conditioning. 

\begin{lemma}[Rounding LRA]
\label{lem:lra-gcr-error}
    Let $\epsilon \geq 0$ and let $\mu$ be an initial pseudo-distribution which satisfies
    \[ \sum_{i \in [n]} \pE_{\mu} r \cdot u_i^p = \pE_{\mu} r \cdot s_u \,,\]
    and 
    \[ \sum_{j \in [m]} \pE_{\mu} r \cdot v_j^p = \pE_{\mu} r \cdot s_v \,\]
    for any monomial $r$ and some $s_u, s_v \leq 2^p \sqrt{\sum_{i \in [n], j \in [m]} A_{ij}^p}$. 
    Then there exists a set $\calT$ of size at most $O_p\left(\log^2 (n)  /\poly_p(\epsilon) \right)$ such that after conditioning on the values of $v_j \in \calT$,
    \[ \left\vert \E_{\{\mu_\calT\}} \pE_{\mu_{\calT}} \sum_{i, j \in [m]} \left(A_{ij} - u_i v_j \right)^p - \pE_{\mu^{\otimes}_{\calT}} \sum_{i, j \in [m]} \left(A_{ij} - u_i v_j \right)^p \right\vert \leq \epsilon \sum_{i \in [n], j \in [m]} A_{ij}^p \,.\]
\end{lemma}
\begin{proof}
    Note that we can rewrite our objective function as follows:
    \[ \sum_{i, j \in [m]} \left(A_{ij} - u_i v_j \right)^p = \sum_{0 \leq q \leq p} \sum_{i \in [n], j \in [m]} \binom{p}{k} A_{ij}^q (u_i v_j)^{p-q}\,.\]
    When $q = 0$ then we have that
    \[ \sum_{i, j \in [m]} \pE_{\mu} u_i^p v_j^p = \sum_{i \in [n]} \pE_\mu u_i^p \left( \norm{v}_p^p\right) = \sum_{i \in [n]} s_v \pE_\mu u_i^p  = s_v \cdot s_u\,.\]
    Note that this is also the pseudoexpectation of this term under independent rounding, so the contribution to the objective value of this term is always preserved in expectation under independent rounding. Furthermore, by~\cref{lem:lra-potentials-bounded},~\cref{lem:poly-potential-aligned}, and~\cref{lem:low-deg-poly-aligned} we have that all the remaining terms are $\left(\epsilon \left( \sum_{i, j \in [m]} A_{ij}^p\right), \frac{\poly_p(\epsilon)}{\log n} \cdot \left( \sum_{i, j \in [m]} A_{ij}^p \right)\right)$-potential aligned and thus by~\cref{lem:gcr_linear_combo} it suffices to take 
    \[ \vert \calT \vert \leq O_p\left(\log^2 (n)  /\poly_p(\epsilon) \right)\,.\]
\end{proof}

We can now prove \cref{thm:lra} similar to how we proved \cref{thm:raw-stress-main}, and note that the running time is dominated by computing a pseudo-distribution of degree $\bigO{\log^2 n /\poly_p(\epsilon)}$.

\printbibliography

\appendix

\section{Euclidean Metric Violation Runtime Speedup}
\label{sec:emv-speedup}

In this section, we will show how to use the rounding procedure described in~\cref{sec:rs} for complete EMV via a straightforward application of the fast LP solver in~\cite{gs12fastsdp}. Specifically, we will complete the proof of the following theorem from~\cref{sec:rs}:
\begin{theorem}[Additive Approximation Scheme for \textsf{$k$-EMV}]
\label{thm:raw-stress-speedup}
Given a set of distances $\calD = \Set{d_{i,j}}_{i,j \in [n]}$, $k \in \mathbb{N}$, where each $d_{ij}$ can be represented in $B$ bits, and $0<\epsilon <1$, there exists an algorithm that runs in $ B^{\mathcal{O}( k^2\log(k/\eps)/\eps^4)} \poly(n)$ time and outputs points $\Set{ \hat{x}_1, \ldots, \hat{x}_n }_{i \in [n]} \in \R^k$ such that with probability at least $0.99$, 
\begin{equation*}
    \E_{i,j\sim[n]} \Paren{ d_{i,j} - \norm{ \hat{x}_i - \hat{x}_j }_2 }^2 \leq \OPT_{\textrm{EMV}} + \eps  \cdot \E_{i,j\sim[n]} d_{i,j}^2. 
\end{equation*}
\end{theorem}

Before we prove~\cref{thm:raw-stress-speedup}, we will first give an overview of the tools from~\cite{gs12fastsdp} which we will use to improve the runtime of~\cref{thm:raw-stress-main}. Guraswami and Sinop give an algorithm for matching the guarantees of rounding algorithms that fall under the following ``local rounding algorithm'' structure (see section 2 of~\cite{gs12fastsdp}). A local rounding algorithm has the following form and is determined by two procedures, $\textsc{SEED}$ and $\textsc{FEASIBLE}$\footnote{In~\cite{gs12fastsdp} they require that $\textsc{SEED}$ and $\textsc{FEASIBLE}$ are deterministic procedures for simplicity but state that their work also applies to randomized $\textsc{SEED}$ procedures. We will use a randomized $\textsc{SEED}$ procedure in our application.}:
\begin{mdframed}
  \begin{algorithm}[Local Rounding Algorithm]
    \label{algo:local-rounding}\mbox{}
    \begin{description}
    \item[Input:] $x \in K_N$ where $K_N \subseteq \mathbb{R}^N$ is some convex body 
    
    \item[Operations:]\mbox{}
    \begin{enumerate}
        \item Let $S(0) \subseteq [N]$ be the initial solution fragment and $y(0) \leftarrow x_{S(0)}$ be the induced solution on those coordinates.
        \item For $i \in [\ell]$:
        \begin{enumerate}
            \item Fail if $\textsc{FEASIBLE}(y(i))$ asserts infeasible.
            \item If $i < \ell$ read more of the solution by setting $S(i+1) \leftarrow \textsc{SEED}(y(i))$ and $y(i+1) \leftarrow x_{S(i+1)}$.
        \end{enumerate}
        \item Perform rounding using only $S(\ell)$ and $y(\ell)$.
    \end{enumerate}
    \end{description}
  \end{algorithm}
\end{mdframed}

Note that the algorithm as stated in~\cref{algo:raw-stress} reads the entire solution in enumerating over all possible seed sets of size $k^2 \log(k/\eps)/\eps^4$. However, we note that this brute force method of determining the correct seed set is unnecessary, and we will overview later in this section how a good seed set can be sampled from a distribution which only depends on the input distances $\calD$.

Guraswami and Sinop give the following speedup for algorithms based on local rounding algorithms:
\begin{theorem}[Fast Solver for Local Rounding (Theorem 19 of~\cite{gs12fastsdp})]
\label{thm:gs-speedup}
    Let $N \in \mathbb{N}^{>0}$ and let $\{\Pi_S\}_{S \subseteq [N]}$ be a set of subspaces, represented by their projection matrices and
    associated with subsets of $[N]$. Furthermore, with each subspace $\Pi_S$, let $K_S \subseteq \Pi_S[0,1]^N$ be an associated convex body such that
    \[ \Pi_S \subseteq \Pi_T \implies \Pi_S K_T \subseteq K_S \,.\]
    Let $n, s$ be positive integers and let \textsc{Feasible}, \textsc{Round} and \textsc{Seed} be functions such that:
    \begin{itemize}
        \item \textbf{\textsc{Feasible}}$_s : \Pi_S \mathbb{Q}^N \to \{\text{feasible}, \Pi_S \mathbb{Q}^N\}$. 
        On input $S \subseteq N, y \in \Pi_S \mathbb{Q}^N$, it asserts feasible if $y \in K_S$ or returns 
        $c \in \Pi_S \mathbb{Q}^N : \|c\|_\infty = 1$ such that $\forall x \in K_S : \langle c, x \rangle < \langle c, y \rangle$
        in time $\mathrm{poly}(\mathrm{rank}(\Pi_S))$.
    
        \item \textbf{\textsc{Seed}}$_s : K_S \to 2^N$. 
        Given $S \subseteq [N]$ and $y \in \Pi_S \mathbb{Q}^N$, it returns subset $S' \supset S$ such that 
        $\frac{\mathrm{rank}(\Pi_{S'})}{\mathrm{rank}(\Pi_S)} \le s$ when $S \neq \emptyset$, and 
        $\mathrm{rank}(\Pi_{S'}) \le n$ when $S = \emptyset$. Its worst case running time is bounded by 
        $\mathrm{poly}(\mathrm{rank}(\Pi_S))$ (or $\mathrm{poly}(n)$ in the case of $S = \emptyset$).

        \item \textbf{\textsc{Round}}$_S : K_S \to L^{[n]}$. 
        On inputs $S \subseteq N$ and $y \in K_S$, returns an approximation to the original problem in time 
        $\mathrm{poly}(\mathrm{rank}(\Pi_S))$.
    \end{itemize}

    Then there exists an algorithm which runs in time $[s^\ell n \log(1/\varepsilon_0)]^{O(\ell)}$ 
    (\emph{compare this with the straightforward algorithm which runs in time $N^{O(1)} \log(1/\varepsilon_0)$}) 
    and achieves the following guarantee: Provided that $\mathrm{vol}(K) \ge \varepsilon_0$, it outputs $y^* \in K_{S(\ell)}$ 
    and $S(0),\ldots, S(\ell)$ set for all $i$:
    \begin{align}
        \Pi_{S(i)} y^* &\in K_{S(i)}, \tag{10} \\
        S(i+1) &= \textsc{Seed}_{S(i)}(y^*). \tag{11}
    \end{align}
    Otherwise it asserts $\mathrm{vol}(K) < \varepsilon_0$.
\end{theorem}

Before we proceed to the proof of~\cref{thm:raw-stress-speedup}, we will first give a version of the rounding algorithm in~\cref{sec:rs} which follows the local rounding algorithm framework.

\begin{mdframed}
  \begin{algorithm}[Local Rounding Algorithm for $k$-$\mathrm{EMV}$]
    \label{algo:local-emv}\mbox{}
    \begin{description}
    \item[Input:] $\{\mu_S\}_{\vert S \vert \leq 2k^2 \log (k/\epsilon) /\epsilon^4 + 2}$ be local distributions over subsets of $\{x_i\}_{i \in [n]}$
    
    \item[Operations:]\mbox{}
    \begin{enumerate}
        \item Let $S_1$ be a subset of size $k^2 \log (k/\epsilon) /\epsilon^4$ where $i$ is sampled proportional to $\E_{j \sim [n]} d_{ij}^2$ and let $S_2$ be a uniformly random subset of size $k^2 \log (k/\epsilon) /\epsilon^4$. 
        \item Read the distribution of all $(2k^2 \log (k/\epsilon) /\epsilon^4 + 2)$-sized subsets including $S_1 \cup S_2$, and ensure that they are consistent on the local distribution of $S_1 \cup S_2$ and that the LP constraints in~\cref{algo:raw-stress} are satisfied.
        \item Sample an assignment to the $x_i \in S_1 \cup S_2$ from the local distribution and condition on the assignment.
        \item Using only the prior local distributions viewed in step 2, assign all other variables by sampling them independently from their marginal distributions. 
    \end{enumerate}
    \end{description}
  \end{algorithm}
\end{mdframed}

We now return to the proof of~\cref{thm:raw-stress-speedup}. 

\begin{proof}[Proof of~\cref{thm:raw-stress-speedup}]
    We will explain how to view the rounding algorithm in~\cref{sec:rs} as a local rounding algorithm and then apply~\cref{thm:gs-speedup}. First, we note that rather than brute forcing over all possible seed sets, it instead suffices to sample the seed set from a distribution that is independent of the LP solution. Specifically, in the analysis of the rounding algorithm, we show that a good point to condition on exists via showing that a random point suffices. 
    
    Recall that we separately analyzed the rounding error on different parts of the objective function via decomposing it into a sum. We can write the $k$-$\mathrm{EMV}$ objective function as follows:
    \[ \E_{i,j \sim [n]} d_{ij}^2 - 2 d_{ij} \norm{x_i-x_j} + x_i^2 + x_j^2 - 2x_i x_j\,.\]
    Note that only two terms in this sum contribute to the rounding error. To ensure rounding error is small on the $-2x_ix_j$ portion of the objective it suffices to sample a uniformly random subset to condition on. To ensure rounding error is small on the $-2d_{ij} \norm{x_i - x_j}$ portion of the objective function, it suffices to sample a subset where $i$ is sampled proportional to $\E_{j \sim [n]} d_{ij}^2$. If we select our set to condition on randomly via sampling from the above distributions, then we achieve the same rounding error guarantees of~\cref{thm:raw-stress-main}. Thus the rounding in~\cref{algo:local-emv} also produces a solution of cost $\OPT_{\mathrm{EMV}} + \eps \E_{i,j \sim [n]} d_{ij}^2$ with probability $0.99$.

    Thus, we can achieve the desired guarantees with an algorithm with $O(1)$ rounds of \textsc{SEED} where we look at at most $n^2$ local distributions, each of size $2k^2 \log (k/\epsilon) /\epsilon^4 + 2$. Since each $x_i$ takes on values in $\Sigma^k$, where $\vert \Sigma \vert \leq \bigO{\epsilon \cdot \log (\Delta k /\epsilon)}$, our seed procedure satisfies $\frac{\mathrm{rank}(\Pi_{S'})}{\mathrm{rank}(\Pi_S)} \le O \left(n^2 \cdot \left(\epsilon \cdot \log (\Delta k /\epsilon)\right)^{k^2 \log (k/\epsilon) /\epsilon^4}\right)$.

    For \textsc{FEASIBLE}, we simply check that the local distributions on subsets read so far are all consistent and that the size $2$ local distributions satisfy the constraints on the LP in~\cref{algo:raw-stress}. These are all linear inequality constraints, so there exists an efficient separation oracle.

    For \textsc{ROUND}, we look at the local distributions of size $1$ and round them to an assignment. This can be done in $\poly(n, B, k, 1/\eps)$ time and only requires the subsets read so far, since we have access to the local distributions of size $1$ from the \textsc{SEED} procedure described above.
    
    Thus, we can apply~\cref{thm:gs-speedup} to complete the proof of~\cref{thm:raw-stress-speedup}.
\end{proof}

\end{document}